\pdfoutput=1
\documentclass[11pt,letterpaper]{article}

\newif\ifLMfonts
\LMfontsfalse
\newif\ifAuthorYear
\AuthorYearfalse

\usepackage[T1]{fontenc}
\usepackage[utf8]{inputenc}
\usepackage{amsmath,amssymb,amsthm}

\ifLMfonts
  \usepackage{lmodern}%
  \newcommand{\TitleSize}{\fontsize{22}{26}\selectfont}
  \newcommand{\SubtitleSize}{\fontsize{13}{16.5}\selectfont}
\else
  \usepackage{newpxtext}%
  \usepackage{newpxmath}%
  \usepackage[scaled=0.95,lining]{FiraSans}%
  
  \newcommand{\TitleSize}{\fontsize{22}{26}\selectfont}
  \newcommand{\SubtitleSize}{\fontsize{13}{16.5}\selectfont}
\fi
\usepackage[letterspace=80]{microtype}

\usepackage[letterpaper,left=0.98in,right=0.98in,top=0.85in,bottom=0.95in,headheight=22pt,headsep=0.28in,footskip=0.42in]{geometry}

\usepackage{xcolor}
\definecolor{ink}{HTML}{1A1F2B}%
\definecolor{navy}{HTML}{1F3A63}%
\definecolor{refcool}{HTML}{2A62B0}%
\definecolor{citewarm}{HTML}{96515B}%
\definecolor{muted}{HTML}{5A6578}%
\definecolor{urlslate}{HTML}{3D4F6B}%
\definecolor{rulegray}{HTML}{C9D3E0}%
\definecolor{qrule}{HTML}{86A2C9}%
\definecolor{webred}{HTML}{B3262E}%
\definecolor{thmback}{HTML}{EEF4FB}%
\definecolor{lemback}{HTML}{F3F5F8}
\definecolor{lembar}{HTML}{6B7A90}
\definecolor{absback}{HTML}{F5F8FC}%
\definecolor{TableHead}{HTML}{E7ECF2}%
\definecolor{TableBand}{HTML}{F6F8FA}%
\colorlet{PaperInk}{ink}
\colorlet{PaperAccent}{refcool}
\colorlet{PaperCite}{citewarm}
\colorlet{PaperRule}{rulegray}
\colorlet{Ink}{ink}
\colorlet{Teal}{refcool}
\colorlet{Ochre}{qrule}

\usepackage{booktabs,array,tabularx,colortbl}
\newcolumntype{Y}{>{\raggedright\arraybackslash}X}
\usepackage{enumitem}
\setlist{nosep}
\setlist[description]{font={\sffamily\bfseries\color{navy}},itemsep=3pt}
\usepackage{graphicx,tikz}
\usetikzlibrary{arrows.meta,positioning}
\usepackage{titlesec,fancyhdr,needspace}
\usepackage[skins,breakable]{tcolorbox}
\usepackage[font=small,labelfont={sf,bf,color=navy},labelsep=period,skip=6pt,singlelinecheck=false]{caption}
\ifAuthorYear
  \usepackage[authoryear,round]{natbib}%
  \renewcommand{\cite}{\citep}%
\else
  \usepackage[numbers,square,sort&compress]{natbib}
\fi
\newcommand{\conceptasset}[2]{\IfFileExists{assets/#1}{\includegraphics[width=\linewidth]{assets/#1}}{\IfFileExists{paper/assets/#1}{\includegraphics[width=\linewidth]{paper/assets/#1}}{#2}}}
\newcommand{\conceptcaption}[3]{\IfFileExists{assets/#1}{\caption{#2}}{\IfFileExists{paper/assets/#1}{\caption{#2}}{\caption{#3}}}}

\titleformat{\section}{\sffamily\Large\bfseries\color{navy}}{\textcolor{refcool}{\thesection}}{0.7em}{}
\titlespacing*{\section}{0pt}{3.0ex plus 1ex minus .2ex}{1.4ex plus .2ex}
\titleformat{\subsection}{\sffamily\large\bfseries\color{navy}}{\textcolor{refcool}{\thesubsection}}{0.6em}{}
\titlespacing*{\subsection}{0pt}{2.3ex plus .8ex minus .2ex}{1.0ex plus .2ex}
\titleformat{\subsubsection}{\sffamily\normalsize\bfseries\color{navy}}{\textcolor{refcool}{\thesubsubsection}}{0.5em}{}
\titlespacing*{\subsubsection}{0pt}{1.8ex plus .6ex minus .2ex}{0.8ex plus .2ex}
\titleformat{\paragraph}[runin]{\sffamily\bfseries\color{navy}}{}{0pt}{}
\titlespacing*{\paragraph}{0pt}{1.2ex plus .4ex minus .2ex}{0.55em}

\renewcommand{\headrulewidth}{0.5pt}
\renewcommand{\headrule}{\hbox to\headwidth{\color{rulegray}\leaders\hrule height \headrulewidth\hfill}}
\fancypagestyle{firstpage}{\fancyhf{}\renewcommand{\headrulewidth}{0pt}\renewcommand{\headrule}{}}%

\makeatletter
\renewcommand{\@makefnmark}{\hbox{\@textsuperscript{\normalfont\color{refcool}\@thefnmark}}}
\renewcommand{\footnoterule}{\kern-3pt{\color{rulegray}\hrule width 0.28\columnwidth height 0.5pt}\kern2.6pt}
\makeatother

\tcbset{barbox/.style={enhanced,frame hidden,sharp corners,boxrule=0pt,breakable=false}}
\newcommand{\PaperKeywords}{\mbox{visual provenance}; \mbox{watermarking}; \mbox{passive detection}; \mbox{message recovery}; \mbox{authenticated provenance}; \mbox{code-based rendering}; \mbox{direct visual generation}}
\renewenvironment{abstract}{%
  \begin{tcolorbox}[barbox,colback=absback,borderline west={3pt}{0pt}{refcool},left=11pt,right=11pt,top=6pt,bottom=6pt,before skip=7pt,after skip=6pt]
  {\sffamily\bfseries\footnotesize\color{navy}\lsstyle ABSTRACT}\par\vspace{2pt}\small}{%
  \par\vspace{3pt}{\sffamily\footnotesize\color{muted}\textbf{\lsstyle KEYWORDS}\quad\PaperKeywords}\par\vspace{2pt}{\sffamily\footnotesize\hypersetup{urlcolor=webred}\color{webred}\textbf{\lsstyle VIDEO}\quad\url{https://youtu.be/14SMl0d_e48}\hspace{1.4em}\textbf{\lsstyle WEB}\quad\url{https://zhenggao-30.github.io/Rethinking-Visual-Provenance/}}\end{tcolorbox}}
\newtcolorbox{tldrbox}{barbox,colback=thmback,borderline west={3pt}{0pt}{navy},left=11pt,right=11pt,top=5pt,bottom=5pt,before skip=6pt,after skip=7pt,fontupper=\small}
\newtcolorbox{takeawaybox}{enhanced,breakable,frame hidden,sharp corners,boxrule=0pt,colback=thmback,borderline west={2.6pt}{0pt}{refcool},left=10pt,right=9pt,top=5pt,bottom=5pt,before skip=9pt,after skip=9pt,fontupper=\small}

\newtheoremstyle{agendaquestion}{0pt}{0pt}{\normalfont}{}{\sffamily\bfseries\color{navy}}{.}{0.55em}
  {\thmname{#1}\thmnumber{ #2}\thmnote{\ \textmd{(#3)}}}
\theoremstyle{agendaquestion}
\newtheorem{question}{Question}
\newtheoremstyle{boxed}{0pt}{0pt}{\itshape}{}{\sffamily\bfseries\color{navy}}{.}{0.5em}
  {\thmname{#1}\thmnumber{ #2}\thmnote{\ \textmd{(#3)}}}
\newtheoremstyle{boxedrm}{0pt}{0pt}{\normalfont}{}{\sffamily\bfseries\color{navy}}{.}{0.5em}
  {\thmname{#1}\thmnumber{ #2}\thmnote{\ \textmd{(#3)}}}
\theoremstyle{boxed}
\newtheorem{theorem}{Theorem}[section]
\newtheorem{proposition}[theorem]{Proposition}
\newtheorem{lemma}[theorem]{Lemma}
\newtheorem{corollary}[theorem]{Corollary}
\theoremstyle{boxedrm}
\newtheorem{definition}[theorem]{Definition}
\newtheorem{example}[theorem]{Example}
\newtheorem{remark}[theorem]{Remark}
\tcbset{thmbox/.style={enhanced,breakable,frame hidden,sharp corners,boxrule=0pt,
  left=8pt,right=8pt,top=5pt,bottom=5pt,before skip=9pt,after skip=9pt}}
\tcbset{coolbox/.style={thmbox,colback=thmback,borderline west={2.4pt}{0pt}{refcool}}}
\tcbset{slatebox/.style={thmbox,colback=lemback,borderline west={2.4pt}{0pt}{lembar}}}
\tcbset{warmbox/.style={thmbox,colback=white,borderline west={2.4pt}{0pt}{citewarm}}}
\tcbset{linebox/.style={blanker,breakable,left=9pt,before skip=7pt,after skip=7pt,borderline west={1.2pt}{0pt}{rulegray}}}
\tcbset{qbox/.style={blanker,breakable,left=9pt,before skip=7pt,after skip=7pt,borderline west={1.4pt}{0pt}{qrule}}}
\tcolorboxenvironment{theorem}{coolbox}
\tcolorboxenvironment{proposition}{coolbox}
\tcolorboxenvironment{corollary}{coolbox}
\tcolorboxenvironment{lemma}{slatebox}
\tcolorboxenvironment{definition}{warmbox}
\tcolorboxenvironment{remark}{linebox}
\tcolorboxenvironment{example}{linebox}
\tcolorboxenvironment{question}{qbox}
\theoremstyle{plain}

\makeatletter
\renewenvironment{proof}[1][\proofname]{\par\pushQED{\qed}\normalfont\topsep6\p@\@plus6\p@\relax
  \trivlist\item[\hskip\labelsep{\sffamily\bfseries\color{navy}#1.}]\ignorespaces}{\popQED\endtrivlist\@endpefalse}
\makeatother

\newcommand{\pflab}[2]{{\color{refcool}\ensuremath{\langle#1\rangle}#2}}
\newenvironment{steps}{\begin{list}{}{%
    \setlength{\leftmargin}{2.2em}\setlength{\labelwidth}{2.0em}\setlength{\labelsep}{0.2em}%
    \setlength{\itemsep}{2pt}\setlength{\parsep}{1pt}\setlength{\topsep}{2pt}\setlength{\partopsep}{0pt}%
    }}{\end{list}}
\newcommand{\step}[2]{\item[\pflab{#1}{#2}.]}
\newcommand{\qedstep}[2]{\item[\pflab{#1}{#2}.]{\sffamily\small\color{navy}Q.E.D.}}
\newcommand{\pf}{\par{\sffamily\small\color{muted}Proof:}\ }
\newcommand{\sref}[2]{\pflab{#1}{#2}}

\providecommand{\Let}{{\sffamily\small\color{navy}Let:}\ }

\newcommand{\qlabrun}[1]{{\sffamily\bfseries\color{navy}#1.}\enspace}
\newcommand{\qlab}[1]{\par\noindent{\sffamily\bfseries\color{navy}#1.}\enspace}

\makeatletter
\newcommand{\PaperTocSep}[1]{\addvspace{7pt}{\noindent\sffamily\bfseries\footnotesize\color{ink}\lsstyle\MakeUppercase{#1}}\par\nobreak\vskip2pt{\color{rulegray}\hrule height 0.5pt}\vskip1pt}
\newcommand{\PaperTocLine}[5]{%
  \ifnum #1>\c@tocdepth \else
    \vskip \z@ \@plus.2\p@
    {\leftskip #2\relax \rightskip \@tocrmarg \parfillskip -\rightskip
     \parindent #2\relax\@afterindenttrue \interlinepenalty\@M \leavevmode
     \@tempdima #3\relax \advance\leftskip \@tempdima \null\nobreak\hskip -\leftskip
     {#4}\nobreak
     {\color{rulegray}\leaders\hbox{$\m@th\mkern\@dotsep mu\hbox{.}\mkern\@dotsep mu$}\hfill}\nobreak
     \hb@xt@\@pnumwidth{\hfil #5}\par}%
  \fi}
\renewcommand{\tableofcontents}{%
  \begingroup
  \setlength{\parskip}{0pt}\color{ink}\sffamily\fontsize{10.2}{12.4}\selectfont
  \hypersetup{linkcolor=ink}%
  \renewcommand{\@dotsep}{4.2}%
  \def\@pnumwidth{1.6em}\def\@tocrmarg{2.6em}%
  \def\l@section##1##2{\addvspace{3.6pt plus 1pt}{\bfseries\PaperTocLine{1}{0em}{1.9em}{##1}{##2}}}%
  \def\l@subsection{\PaperTocLine{2}{1.9em}{2.6em}}%
  {\noindent\color{ink}\sffamily\bfseries\Large Contents\par}\vspace{1pt}%
  {\color{refcool}\hrule height 1.2pt}\vspace{3pt}%
  \@starttoc{toc}%
  \endgroup}
\makeatother

\let\PaperOrigAppendix\appendix
\renewcommand{\appendix}{\PaperOrigAppendix
  \addtocontents{toc}{\protect\PaperTocSep{Appendices}}%
  \noindent{\color{refcool}\rule{\linewidth}{2pt}}\par\vspace{5pt}
  \noindent{\sffamily\bfseries\Large\color{navy}Appendices}\hfill{\sffamily\footnotesize\color{muted}\lsstyle RETHINKING VISUAL PROVENANCE}\par\vspace{3pt}
  {\color{rulegray}\hrule height 0.5pt}\par\vspace{2pt}}



\newcommand{\Prb}{\mathbb{P}}

\newcommand{\TV}{\operatorname{TV}}

\title{Rethinking Visual Provenance:\\Detection and Watermarking Across\\Direct Visual Generation and LLM-Driven Code Rendering}
\author{}
\date{4 October 2026}
\usepackage[colorlinks=true,linkcolor=refcool,citecolor=citewarm,urlcolor=urlslate,
  bookmarksnumbered=true,pdfdisplaydoctitle=true,pdflang={en},
  pdftitle={Rethinking Visual Provenance: Detection and Watermarking Across Direct Visual Generation and LLM-Driven Code Rendering},
  pdfauthor={Zheng Gao, Xiaoyu Li, Zhicheng Bao, Jiaojiao Jiang},pdfsubject={Research agenda},
  pdfkeywords={visual provenance; watermarking; passive detection; message recovery; authenticated provenance; code-based rendering; direct visual generation}]{hyperref}
\newcommand{\conceptassetw}[3]{\IfFileExists{assets/#2}{\includegraphics[width=#1\linewidth]{assets/#2}}{\IfFileExists{paper/assets/#2}{\includegraphics[width=#1\linewidth]{paper/assets/#2}}{#3}}}
\newcommand{\PaperAuthors}{Zheng Gao\quad Xiaoyu Li\quad Zhicheng Bao\quad Jiaojiao Jiang}
\newcommand{\PaperAffiliations}{UNSW Sydney\quad\textbullet\quad\href{mailto:zheng.gao1@unsw.edu.au}{zheng.gao1@unsw.edu.au}}

\makeatletter
\renewcommand{\maketitle}{%
  \thispagestyle{firstpage}\enlargethispage{0.6in}%
  \noindent{\color{refcool}\rule{\linewidth}{2.2pt}}\par\vspace{6pt}
  \noindent{\sffamily\bfseries\TitleSize\color{navy}Rethinking Visual Provenance:\par}\vspace{3pt}
  \noindent{\sffamily\SubtitleSize\color{citewarm}Detection and Watermarking Across\\Direct Visual Generation and LLM-Driven Code Rendering\par}\vspace{9pt}
  \noindent{\sffamily\fontsize{11.5}{15}\selectfont\color{ink}\PaperAuthors\par}\vspace{3pt}
  \noindent{\sffamily\footnotesize\color{muted}\PaperAffiliations\par}\vspace{5pt}
  {\color{rulegray}\hrule height 0.5pt}\vspace{1pt}}
\makeatother

\begin{document}
\color{ink}
\maketitle
\begin{abstract}
AI systems create images and videos with image/video generation models or by writing code and graphics descriptions that are then rendered. These routes can produce similar visible artifacts but expose different representations, intervention points, and provenance evidence. We develop a production-centered framework that compares detection and watermarking across both routes. An explicit verification specification distinguishes passive inference, message recovery, and authenticated provenance. We organize image, video, source-code, and rendering-aware watermarks by production stage. We examine the different requirements of generated images and video, plots and SVG, programmable video, and agent-composed workflows. Documented Claude, OpenAI, and rendering-tool interfaces connect the framework to concrete systems. We pose ten scoped research questions on identifiability, observability, fair comparison across stages, recoverable payload, reconstruction, synchronization, composition, hybrid local contribution, and private production-event authentication. The result is a conceptual research agenda grounded in published methods, inspected interfaces, and elementary boundary examples. It reports no experiments and claims no new theorems; its appendix results are elementary calculations, and documentation and source inspection establish interfaces, not empirical robustness.
\end{abstract}
{\centering\refstepcounter{figure}\label{fig:routes}%
\conceptassetw{0.88}{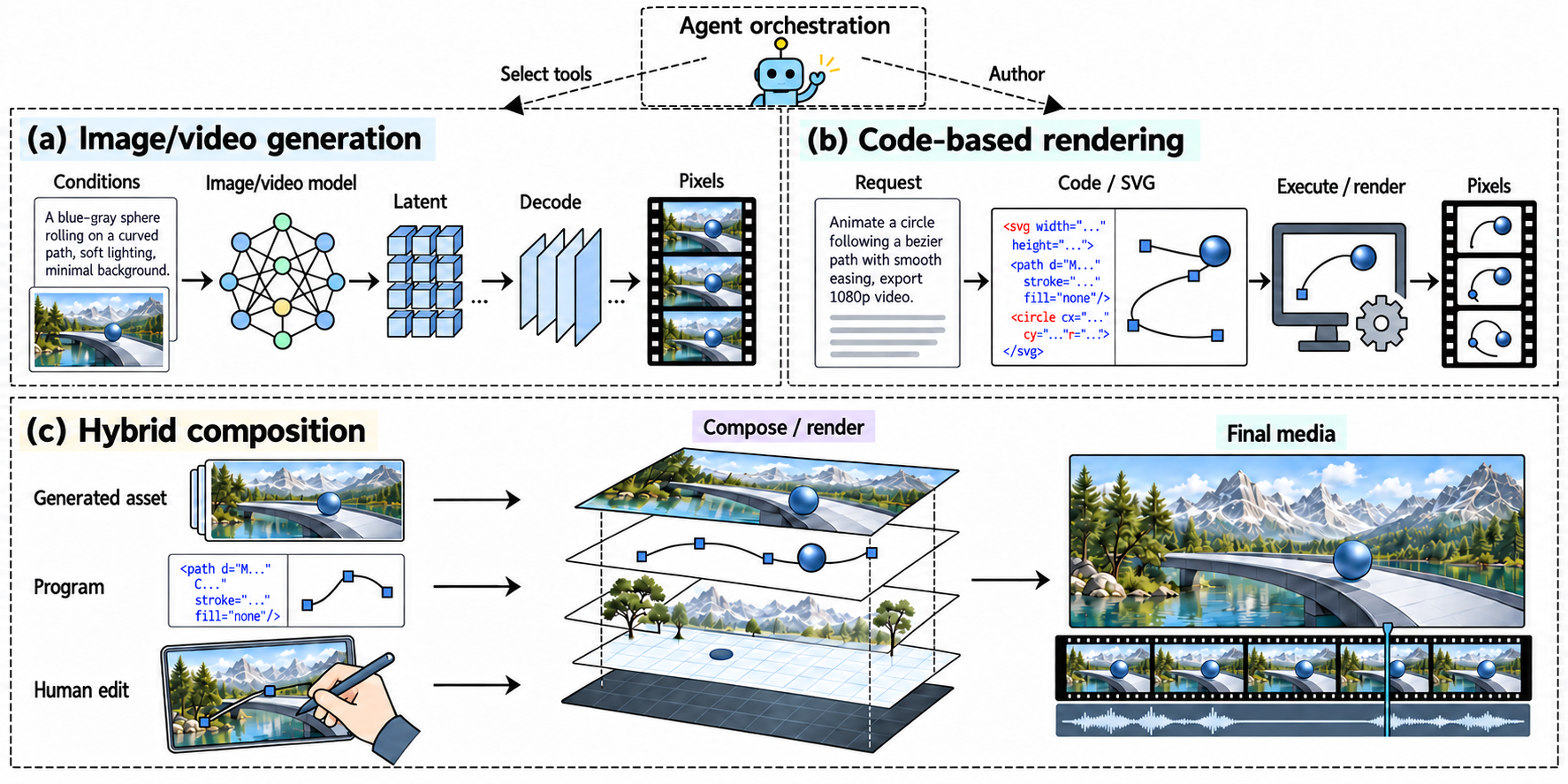}{%
\resizebox{\linewidth}{!}{\begin{tikzpicture}[>=Stealth,font=\small,
box/.style={draw=Ink!45,rounded corners=2pt,minimum height=9mm,text width=29mm,align=center,inner sep=4pt}]
\node[box,fill=Ink!5] (a) at (0,0) {Human or agent\\controller};
\node[box,fill=Teal!8] (n) at (4,1) {Learned\\generation};
\node[box,fill=Teal!8] (l) at (8,1) {Latent / image\\decoding};
\node[box] (v) at (12,1) {Raster image\\or video};
\node[box,fill=Ochre!10] (p) at (4,-1) {Code / vector\\description};
\node[box,fill=Ochre!10] (r) at (8,-1) {Interpretation\\and rendering};
\node[box] (w) at (12,-1) {Raster image\\or video};
\draw[->,draw=Ink] (a.east) -- (n.west);
\draw[->,draw=Ink] (a.east) -- (p.west);
\draw[->,draw=Ink] (n) -- (l);\draw[->,draw=Ink] (l) -- (v);
\draw[->,draw=Ink] (p) -- (r);\draw[->,draw=Ink] (r) -- (w);
\end{tikzpicture}}}\par}
\vspace{2pt}
\begin{tldrbox}
\textbf{\sffamily\color{navy}Figure~\thefigure.}\enspace\textbf{\sffamily TL;DR}\quad Image/video generation and code-based rendering can each produce AI visual media; an agent can orchestrate either, and hybrid outputs combine generated assets, programs, and human edits (schematic illustration, not a measurement or an assertion of pixel equality). Compare their watermarks through what can be changed, what reaches the verifier, and which provenance assertion the evidence supports.
\end{tldrbox}

\clearpage
\begingroup
\setlength{\parskip}{0pt}
\sffamily\fontsize{10}{12}\selectfont
\hypersetup{linkcolor=black}
\makeatletter
\renewcommand{\@dotsep}{5.2}
\renewcommand{\l@section}[2]{\vskip 4pt\@dottedtocline{1}{0em}{1.8em}{\bfseries #1}{\bfseries #2}}
\renewcommand{\l@subsection}{\@dottedtocline{2}{1.8em}{2.6em}}
\makeatother
\setcounter{tocdepth}{2}
\renewcommand{\contentsname}{\textcolor{black}{Contents}}
\tableofcontents
\endgroup
\clearpage
\section{Introduction}
\label{sec:introduction}

A visual explanation can come from image/video models (\emph{image/video generation}) or from rendered code and graphics descriptions (\emph{code-based rendering}).\footnote{These are working labels for composable operations, not disjoint model families. The title's ``direct visual generation'' means image/video generation; ``direct'' includes latent decoding and does not exclude code implementation or agent orchestration. Code-based rendering interprets explicit programs or declarative graphics such as SVG.} Diffusion models provide image and video examples, and the documented Claude Code--Remotion workflow provides a code-based example \cite{ho2022,rombach2022ldm,claudeCodeOverview,remotionCodingAgents}.

Viewers may receive neither source nor intermediate state nor execution records, so the visible artifact reveals only part of its production.

Consider an animation of a ball that follows a curve and carries a speed label. It could be synthesized, calculated and rendered, or composed from a generated background, programmed motion, and recorded narration. These implementations serve the same communication task but expose different intervention and observation points (Figure~\ref{fig:routes}), potentially coordinated by a human or agent.

\subsection{From production choices to observable evidence}

\textbf{The output format does not determine the provenance question.} An MP4 file may contain synthesized footage, rendered graphics, recorded material, or a composition. A static image can arise from a learned image tool or executed plotting code, two distinct operations in OpenAI's documentation \cite{openaiCodeInterpreter,openaiImageGeneration}. Claude's documented custom visuals yield an HTML or SVG description that can be downloaded and inspected \cite{claudeCustomVisuals}. Receiving that description and receiving a raster snapshot provide different observations.

The distinction concerns operations rather than model families: generators can accept language, assistants can author code, and programs can include generated assets. An LLM agent can control several operations and does not define an exclusive origin class. Remotion's Prompt to Video template combines generated script, images, and voiceover \cite{remotionPromptToVideo}, which motivates a workflow of connected synthesis, rendering, and composition stages.
\subsection{Three challenges for detection and watermarking}

\textbf{Evidence must reach the intended verifier.} An interpreter may ignore a source mark, and a screenshot may omit file credentials. Changes to emitted colors can reach pixels. Video Seal marks video, SrcMarker targets program representations, and CopyRNeRF connects a marked representation to rendered views \cite{videoseal2024,yang2024srcmarker,luo2023copyrnerf}. Comparing them requires locating both intervention and observation, not assuming interchangeable carriers.

\textbf{A signal needs a precise interpretation.} Passive inference estimates a production label without a mark. Watermark verification tests for a mark (detection) or estimates its embedded message (message recovery). Authentication checks an issuer-bound assertion about content or about a production event. A production event is one invocation of a production step, such as one generation call, one rendering run, or one export. It is not the event that the video depicts.

In this article, detection covers passive inference and watermark detection; message recovery and authentication are the other verification tasks (Section~\ref{sec:recovery}). A mark applied by a rendering service can concern processing of human-authored code, and a texture mark concerns an asset, but neither identifies every contributor. The Coalition for Content Provenance and Authenticity (C2PA) specifies Content Credentials \cite{c2pa}. Their assertions, signed claims, and trust relationships make the scope of the authenticated assertion explicit \cite{c2pa}.

\textbf{Robustness depends on access and permissions.} Recompression, source editing, rerendering, and clean-asset replacement act on different objects (Section~\ref{sec:stages}). Pixel-only and source-visible attackers therefore have different capabilities, and deployments place code, logs, and keys under different control. Remotion's server, browser, and Lambda backends provide concrete examples \cite{remotionSSR,remotionClientRender,remotionLambda}. A recovery claim must declare the accessible state and the allowed transformations.

\subsection{A production-centered framework}\label{sec:spec-list}

We state the following seven items together as one verification specification (Section~\ref{sec:formal-contract}). It declares:
\begin{enumerate}[label=(\arabic*)]
\item the production specification;
\item the encoder's allowed observations and interventions;
\item the verifier's observation;
\item the attacker's view and allowed policies;
\item the task-preservation predicate;
\item the assertion under test, together with the verification task;
\item the null family of unmarked observation laws.
\end{enumerate}
Key generation, issuance queries, and verifier permissions complete it. The specification distinguishes a pixel-only authorship test, a source-visible message-recovery decoder, and a trusted export record by the assumptions each needs.

Evidence can remain partial. An authenticated asset may have undocumented later composition, and a negative detector result may leave several histories possible. Preserving uncertainty and allowing abstention prevents origin evidence from being overinterpreted as complete authorship, authorization, or truth of the depicted event.

\subsection{Scope and organization}\label{sec:scope}

We connect image/video watermarking, source-code and rendering-aware precedents, documented workflows, and ten scoped research questions. Elementary boundary examples show when information is lost or transmitted. The questions set restricted targets for definitions, constructions, attacks, or conditional limits.

\textbf{Limits of the article.} The article is a conceptual research agenda. We report no experiments, no new theorems, no empirical benchmark, and no new construction claimed to meet the ten questions. Documentation and source inspection establish interfaces, not empirical robustness or measured watermark performance. The numbered appendix results are standard facts or direct calculations that make the boundary examples checkable. The image versions of Figures~\ref{fig:routes}, \ref{fig:bottleneck}, \ref{fig:image-stages}, \ref{fig:rerendering}, and~\ref{fig:hybrid-evidence} are schematic illustrations made with an image-generation tool and carry no empirical content.

\textbf{Organization.}
\begin{description}[style=multiline,leftmargin=6.6em]
\item[Section~\ref{sec:background}] fixes the operations, observation interfaces, verification specification, and tasks.
\item[Section~\ref{sec:stages}] organizes watermarks and related tasks by production stage.
\item[Section~\ref{sec:applications}] examines visual applications.
\item[Section~\ref{sec:systems}] studies documented deployments.
\item[Sections~\ref{sec:agenda}, \ref{sec:conclusion}] pose the ten research questions and conclude.
\item[Appendices] \ref{app:observation}--\ref{app:temporal} give elementary results and calculations; Table~\ref{tab:roadmap} is their roadmap. Appendix~\ref{app:glossary} gives the candidate locations (Table~\ref{tab:locations}), a glossary, and the notation (Table~\ref{tab:notation}). Appendix~\ref{app:mechanism} gives the mechanism map (Table~\ref{tab:mechanism-map}) with notes on adjacent work.
\end{description}

\section{Background and Conceptual Foundations}
\label{sec:background}
\label{sec:information}

\begin{takeawaybox}
\textbf{Takeaway: The production process and the received artifact are different objects.} An image or video records visible content; a production history records operations, inputs, and dependencies. A verifier normally receives only part of that history. This section fixes the objects needed to compare both production routes before surveying intervention mechanisms.
\end{takeawaybox}

\subsection{Production operations and observation interfaces}

\textbf{Image/video generation and code-based rendering are operations that a workflow can combine.} Image/video generation uses an image/video model to produce pixels or a representation that is then decoded into pixels. Code-based rendering executes or interprets an explicit drawing program, animation, or graphics description under a specified environment. The two operations can overlap: neural rendering combines learned components with graphics operations \cite{tewari2020neuralrendering}. A program can call an image/video model or reuse generated footage, and an LLM can write SVG or plotting code.

The comparison follows the representations and operations that a workflow exposes, not an exhaustive partition by model architecture. Static versus temporal output is a separate axis. We use four words for production:
\begin{description}[leftmargin=6.6em,labelwidth=5.6em,labelsep=0.9em,align=left,topsep=3pt,partopsep=0pt]
\item[Route] the informal contrast between image/video generation and code-based rendering;
\item[Operation] a node of the production model in Section~\ref{sec:formal-contract};
\item[Workflow] a named real system or a described composition of operations;
\item[Pipeline] a specific implemented chain.
\end{description}

\textbf{A claim about a verifier of one representation cannot silently transfer to a verifier of another.} The same production can expose several artifacts: source code, a scene or asset representation, a browser preview, decoded pixels, an encoded file, and accompanying credentials. Fonts, resources, viewport, frame rate, and runtime can affect the map between them. Figure~\ref{fig:routes} summarizes the operations, and Figure~\ref{fig:bottleneck} isolates the observation boundary.

\subsection{A production model and verification specification}
\label{sec:formal-contract}

\textbf{Stating the assumptions of a verification specification prevents a change of observation, attacker, or assertion from being mistaken for a comparison of watermark algorithms.} Following W3C PROV, provenance concerns the entities, activities, and agents involved in producing an artifact \cite{w3cPROV}. For the bounded comparisons, we model production by a finite directed acyclic graph $\mathcal G=(\mathcal N,\mathcal E)$.\footnote{This functional production graph and the verification specification below are modeling choices for this article. They organize operations, observations, permissions, and assertions; they do not replace a provenance interchange standard or authenticate a declared history.} Each node $v$ emits an artifact $Y_v$ in a declared space and applies a map
\[
 Y_v=F_v\bigl((Y_u)_{u\in\operatorname{pa}(v)},Z_v\bigr).
\]
Here $\operatorname{pa}(v)$ is the set of parents of $v$. We fix the external input nodes to a task instance $x$. The additional inputs and random variables $Z=(Z_v)_v$ have a specified joint distribution $\mu_x$, and they need not be independent.

We can include the randomness of a generator, runtime, or renderer in $Z$. Each node map is defined on allowed inputs; allowed timeouts or failed invocations produce declared failure artifacts. A topological execution defines the production state $S$, which contains the declared graph, inputs, environments, realized values of $Z$, and emitted artifacts. A claimed record of $S$ is not automatically an authenticated record of execution.

A generation node uses a specified image/video model, including any separate decoder it requires. A rendering node interprets an executable or declarative description under an environment. A composite implementation can use several nodes, or one node with stated component operations. A rendering description can reference model-generated assets. The labels alone do not determine the available state or watermark controls.

An LLM agent is a controller that may select or supply inputs to either node type. Its identity is separate from the operation producing a particular visual contribution. The first mathematical problems fix $\mathcal G$. An adaptive agent instead requires a bounded policy and a distribution over its resulting execution traces, not an assumed fixed history.

Let $O(S)$ specify the verifier's observation. Examples are a raster image, quantized frames, an SVG source file, or an exported file with credentials. Include public design parameters, known models, and side information explicitly. An SVG-to-PNG export changes $O$; it is not simply a change of filename. For finite questions, bound file length, precision, raster dimensions, and clip length. Continuous representations require their own measurable-space and precision assumptions.

We collect the assumptions in the following verification specification
\[
 \mathfrak C=(\mathfrak P,\mathcal I,O,\mathcal A,U,\chi,\mathcal Q_0).
\]
Its seven symbols name the items of the list in Section~\ref{sec:spec-list}, in the same order. The production specification $\mathfrak P=(\mathcal G,(F_v)_{v\in\mathcal N},(\mu_x)_x)$ consists of the graph, the node maps, and the input laws $\mu_x$ of the task instances $x$. The symbol $\mathcal I$ denotes the encoder's allowed observations and interventions.

Key generation, issuance queries, and verifier permissions complete this specification. Intervention stages alone do not specify their access. A source-editing attacker and a received-image editor belong to different $\mathcal A$ even when both return pixels.

An encoder or attacker may use only its declared view. An upstream intervention records allowed changes to inputs or node maps, then executes the affected downstream nodes under a declared rerun-randomness law. It is not permission to replace every artifact independently.

Write task preservation as $U_x(s_0,s')\in\{0,1\}$ for a declared baseline state $s_0$ and a marked or attacked state $s'$.\footnote{Task preservation means satisfying the stated application requirements, for example keeping chart values and labels correct. It need not preserve every pixel or every behavior of the original program.} Absolute task specifications may omit $s_0$. Specify any coupling of randomized baseline and marked executions.

In the first bounded models, require task preservation almost surely for encoding and allowed attack policies. A model that allows task-preservation failures must instead state their probability and treatment. Section~\ref{sec:recovery} states recovery and false-alarm requirements.

\subsection{Passive inference, message recovery, and authentication}
\label{sec:recovery}

\textbf{Passive inference, watermark verification, and authentication support different assertions.} \emph{Passive inference} estimates a production label without arranging to mark the output. \emph{Watermark verification} includes \emph{watermark detection}, which tests for a mark, and \emph{message recovery} (payload extraction), which estimates its embedded message.

\emph{Authentication} checks an issuer-bound assertion under declared trust assumptions. Our questions about production events additionally bind an artifact or allowed derivative to an authorized production event. Signed assertions and content binding follow C2PA terminology (Appendix~\ref{app:glossary}; \cite{c2pa}). State the guarantees of the three tasks and the assertions they support separately.

Watermark Forensics separates detection, attribution, payload extraction, and localization \cite{li2026watermarkforensics}. These statistical objectives help specify the visual verifier's task, but the source assumptions and rate results of that work need a separate model before they transfer to rendering or video editing.

For example, a trusted rendering service can mark a video made from human-authored code. The mark could support an assertion about service participation, but it does not identify the code's author.

An animation can retain a mark from a generated texture while separate authoring produces most of the composition, and a provider might sign a file it converts without being the original source of the ideas it represents. None of these situations invalidates the narrower evidence, but its interpretation must travel with the detection protocol.

Absence of a detected mark does not certify that no AI contributed. An unknown generator, changed carrier, or unavailable evidence may leave the verifier unable to decide. Conversely, recovering a mark does not establish the truth of the depicted event or the completeness of the reported history. The verifier should be able to reject or abstain when the evidence does not support the assertion under test.

Take $\varepsilon,\alpha\in[0,1]$. For message recovery, let key generation yield issuance and verification material $(K_e,K_v)$. An allowed intervention with message $m$ and an allowed attack policy $a$ produces state $S_a^{K_e,m}$. A decoder $D$ returns a message or abstention $\bot$. One possible correct-recovery requirement averages over key generation and is uniform over fixed tasks, messages, and policies:
\[
 \Prb\{D(O(S_a^{K_e,m}),K_v)=m\}\geq1-\varepsilon
 \quad\text{for every allowed }x,m,a.
\]

Recovery is unconditional on the task-preservation check. The probability includes the specified key generation, production, encoder, attacker, and decoder randomness. Four points limit the requirement:
\begin{enumerate}[label=(\alph*),topsep=3pt,partopsep=0pt]
\item An adaptive policy is fixed before those draws, although its actions can depend on its declared view.
\item Average-case completeness instead needs a distribution over task instances, messages, and policies (Remark~\ref{rem:worst-average}).
\item Fixed-key reliability (Remark~\ref{rem:key-average}), and message or task selection after key-dependent observations, require separate quantifiers or a selection game.
\item We declare the joint law of keys and production variables $Z$ and do not infer independence from their separate notation.
\end{enumerate}

A false alarm occurs when the decoder returns a message for an unmarked observation. Let each $Q\in\mathcal Q_0$ be a specified joint law of unmarked observations and verifier key material. A possible requirement is
\[
 \sup_{Q\in\mathcal Q_0}\Prb_{(o,k_v)\sim Q}
 \{D(o,k_v)\ne\bot\}\leq\alpha.
\]
This is a proposed specification, not a guarantee of any cited system. Message recovery does not establish $\chi$ automatically: an authenticated production-event assertion additionally needs an authorized relation between production events and content, and an issuance/forgery game. Section~\ref{sec:security-agenda} develops those questions.

\subsection{Information can disappear without every watermark disappearing}\label{sec:identical-observation}

\textbf{Rendering can collapse some distinctions while transmitting others.}
Two different descriptions may render to exactly the same pixels in a fixed environment (Definition~\ref{def:render-fiber}). Suppose a uniformly selected source message leaves the full observation unchanged and has no message-correlated side information. Then no decoder that sees only that observation, with randomness independent of the message, can beat the best constant guess (Proposition~\ref{prop:identical} and Corollary~\ref{cor:exact-equality}). That guess succeeds with probability $1/N$ among $N$ equally likely messages. This is a conditional boundary, not a claim that one can erase watermarks.

Source inspection may still distinguish the descriptions (Remark~\ref{rem:invisible}). An intervention that affects visible geometry, texture, or decoded content is a different case. Its information can reach the output if the map and task constraints allow it.

Conditional equality is also different from equality of population distributions. Human and model authors can choose programs with different probabilities even when the same procedure renders each program. Example~\ref{ex:mixture} is an artificial case: a human class selects an all-black program with probability three quarters, an LLM class with probability one quarter, and each otherwise selects an all-white program, so the clip laws differ.

Appendix~\ref{app:observation} gives the finite Bayes calculation (Proposition~\ref{prop:bayes}) and the common-channel contraction argument (Proposition~\ref{prop:contraction}). These elementary boundaries clarify the agenda without asserting universal watermark erasure or universal route identifiability.

\begin{figure}[ht]
\centering
\conceptasset{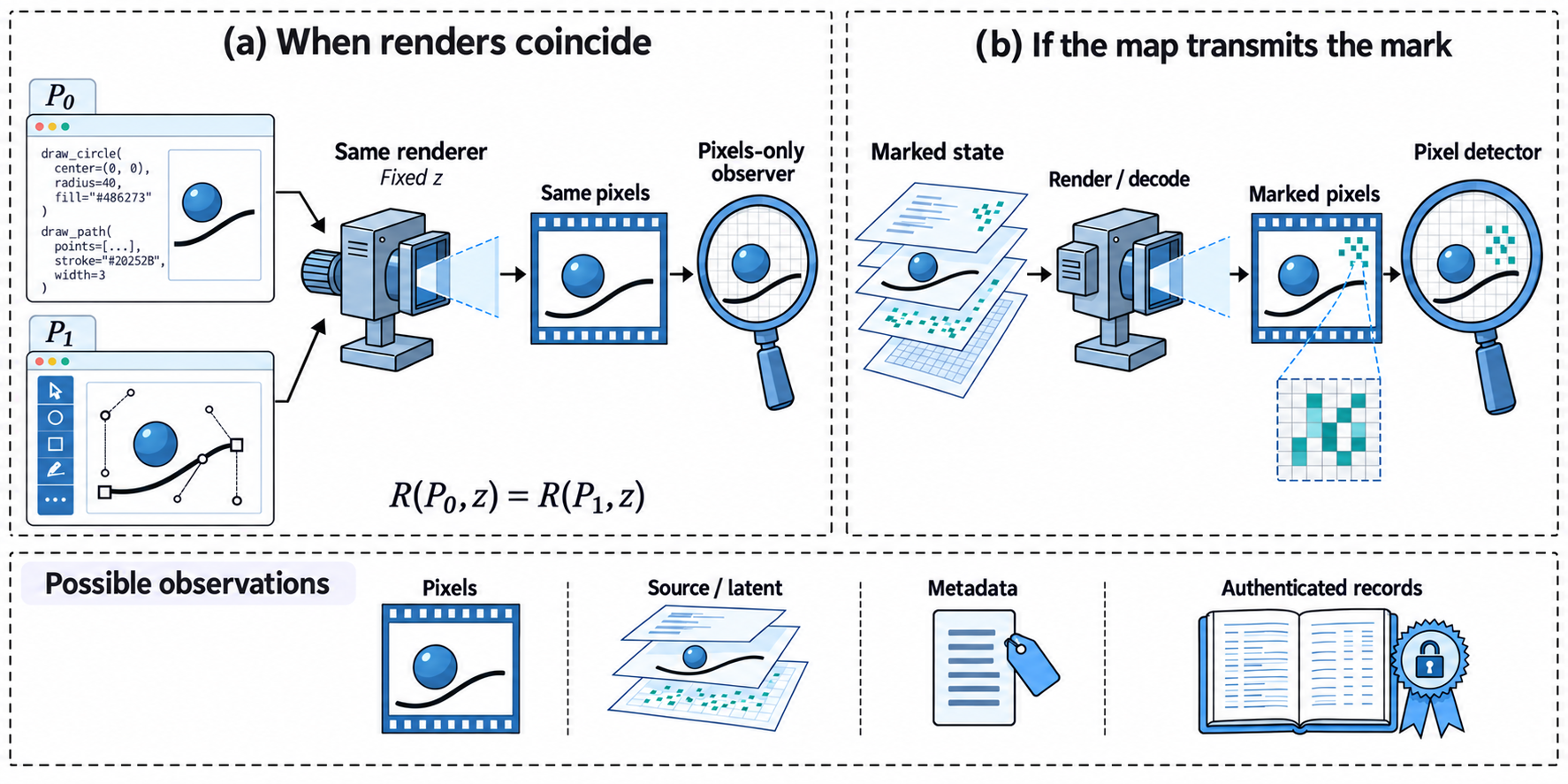}{%
\resizebox{0.95\linewidth}{!}{\begin{tikzpicture}[>=Stealth,font=\small,
box/.style={draw=Ink!45,rounded corners=2pt,minimum height=9mm,text width=25mm,align=center,inner sep=4pt}]
\node[box,fill=Ochre!10] (p0) at (0,1) {Description $P_0$};
\node[box,fill=Ochre!10] (p1) at (0,-1) {Description $P_1$};
\node[box,fill=Ink!5] (r) at (4,0) {Fixed renderer $R$\\and environment $z$};
\node[box,fill=Teal!8] (v) at (8,0) {Same observation\\$o_0$};
\node[box] (a) at (4,2) {Source observer\\can receive $P_i$};
\node[box] (b) at (12,0) {Pixel observer\\receives $o_0$};
\draw[->,draw=Ink] (p0) -- (r);\draw[->,draw=Ink] (p1) -- (r);
\draw[->,draw=Ink] (r) -- (v);\draw[->,draw=Ink] (v) -- (b);
\draw[->,dashed,draw=Ink] (p0.north) |- (a.west);
\draw[->,dashed,draw=Ink] (p1.west) -- ++(-0.6,0) |- (a.north);
\end{tikzpicture}}}
\conceptcaption{observation-pastel-v10.png}{\textbf{What reaches the verifier depends on the production map.} Different descriptions can coincide under a fixed renderer, while an intervention that changes the output can transmit a mark. In the figure, $R$ is the renderer, $z$ the environment, and $P_0$ and $P_1$ are two descriptions (Appendix~\ref{app:observation}). The source snippets and rendering apparatus are schematic; the left panel stipulates equal renders rather than demonstrating them by execution. Pixels, source or latent state, metadata, and authenticated records are distinct possible observations.}{A conditional observation bottleneck. Distinct descriptions $P_0$ and $P_1$ can give the same declared render under a fixed environment $z$. A source observer may receive the description, while the pixel observer receives only the common observation $o_0$. This example establishes neither universal erasure nor an authorship guarantee.}
\label{fig:bottleneck}
\end{figure}

\subsection{Scope of evidence}

\textbf{The article combines published watermark mechanisms, inspected official product and renderer interfaces, and explicitly conditioned mathematical examples.} The source library records versions and reading depth: several newer methods are abstract-screened, and selected central methods were read in technical detail. The mechanism map (Table~\ref{tab:mechanism-map}) lists, for each family, the observation boundary of its examples; its entries are precedents, not security guarantees. Section~\ref{sec:scope} states what documentation and source inspection establish. The ten Questions in Section~\ref{sec:agenda} are bounded in scope.

\section{Watermarking Across Production Stages}
\label{sec:stages}
\label{sec:programs}

\textbf{The intervention point determines which evidence can reach the verifier.}
A producer can mark a finished image, constrain a generator's internal state, transform source code, modify an asset, or instrument a renderer. These are places where a producer can place a mark, not progressively stronger capability levels. Several can occur in one production graph, and the verifier need not observe the representation that carries the mark. We organize the literature around that separation: what changes during production, what the verifier receives, and what the visual task permits to change.

The five stages map to the production graph (Section~\ref{sec:formal-contract}) and to Table~\ref{tab:locations} as follows.
\begin{description}[leftmargin=10.4em,labelwidth=9.6em,labelsep=0.8em,align=left,topsep=3pt,partopsep=0pt]
\item[Output pixels] the final pixel artifact; row Final pixels.
\item[Generation states] $Z_v$ and $F_v$ of a generation node; rows Representation, Emission stage.
\item[Source] the input description of a rendering node; rows Request / text (code tokens only) and Representation.
\item[Assets] an input artifact of a rendering node, or of a later generation node; row Representation.
\item[Execution and export] $F_v$ of a rendering node and its exports; rows Emission stage, Final pixels, and Side information (an observation, not a node).
\end{description}
The map is not one-to-one: Representation spans three stages, and Emission stage and Final pixels span two each.

\subsection{Output pixels: a common baseline}

\textbf{A producer can mark the output of either production route after the route produces pixels.}
A raster embedder receives an image or video and a message, then modifies the visual signal. Its downstream decoder can operate without knowing which of image/video generation, code-based rendering, or a camera produced the input.

HiDDeN learns image embedding and decoding through specified distortion layers \cite{hidden2018}. StegaStamp extends the channel to physical capture, with localization and alignment before message recovery \cite{stegastamp2019}. Their differing channels show why the verifier's observation must include more than a file extension. A downloaded image and a photograph of its display present different recovery problems.

Resolution and localization create further image-specific requirements. TrustMark learns raster embedding and extraction, scales a learned embedding residual to the cover image's resolution, and separately studies watermark removal for re-watermarking \cite{trustmark2023}. Watermark Anything produces spatial detection scores and local message predictions, supporting recovery of multiple messages from composed images \cite{wam2024}. These mechanisms expose different units of evidence: a whole-image message, a detected region, or a message associated with that region. Detecting overlapping marked regions does not establish recovery of every overlapping message. The reported image results of these methods do not by themselves establish temporal consistency or complete composition history.

For video, RivaGAN uses learned embedding and extraction with attention \cite{rivagan2019}. DVMark distributes information across spatial and temporal scales \cite{dvmark2021}. ItoV adapts image watermarking by combining temporal and channel dimensions \cite{itov2023}.

Video Seal computes pixel residuals on selected frames, propagates them temporally, and aggregates extracted soft bits across frames \cite{videoseal2024}. Its discussion of artifacts during fast motion illustrates a concrete constraint: an economical embedding schedule must still respect the changing scene. Transform-domain temporal redundancy in DTCWT-SVD \cite{dtcwt2022} and region selection in FlowMark \cite{flowmark2026} offer further examples of designing around video structure.

This shared baseline is essential for a fair comparison. An animation rendered from code can receive an output watermark through the same interface as model-generated footage. What that mark supports is an assertion about the marking event or associated message. It does not, by itself, distinguish the route that produced the unmarked content. Processing the exported pixels and producing a fresh export from accessible upstream material are different operations. A claim about the first does not establish resistance to the second.

\subsection{Generation states, models, and decoders}

\textbf{Internal marks can exploit generation structure.} Figure~\ref{fig:image-stages} compares image intervention and extraction interfaces, each requiring its own verification access.

Tree-Ring places a Fourier pattern in initial diffusion noise \cite{treering2023}. Gaussian Shading uses message-dependent Gaussian sampling with keyed randomization \cite{gaussianshading2024}. The studied extraction procedures of both methods use diffusion inversion and auxiliary watermark information. This differs from receiving an image and applying a standalone learned extractor. Inversion accuracy, compatible generation procedures, and access to the model belong in the mechanism description.

SLICE associates separately extracted semantic factors with partitions of an initial diffusion latent \cite{gao2026slice}. It then compares reconstructed and inverted partitions during verification. This content-sensitive design also makes geometric registration part of the verification.

For video, VideoShield encodes an encrypted template through signs of initial Gaussian noise and extracts through the corresponding video-model machinery \cite{videoshield2025}. VideoMark uses frame-wise pseudorandom-code messages and a temporal matching procedure to accommodate sequence edits \cite{videomark2025}. Its treatment of the difficult first frame in image-to-video generation is a reminder that temporal positions can have unequal extraction quality.

SIGMark addresses generation-time noise marking, with sequence organization designed around temporal disturbances \cite{sigmark2026}. Terms such as ``blind extraction'' should retain each paper's specified access model and do not imply that a verifier needs neither a model nor a key.

A second family changes the generation machinery. Stable Signature fine-tunes a latent image decoder so that a learned image extractor can recover a signature without diffusion inversion \cite{stablesignature2023}. For video, LVMark modulates decoder weights \cite{lvmark2024}. Video Signature selectively fine-tunes decoder components \cite{vidsig2025}. SPDMark uses key-indexed parameter displacements \cite{spdmark2025}.

These methods make different commitments about producer control, output variation, and verifier resources. A conventional vector rasterizer does not automatically expose an intervention of this kind. Conversely, a program can call a marked neural decoder and inherit its output as an asset. The useful comparison is thus between accessible components in the actual pipeline, not between product labels. Reported processing robustness and distribution-preservation properties remain scoped to the individual methods and their assumptions.

\begin{figure}[!htbp]
\centering
\conceptasset{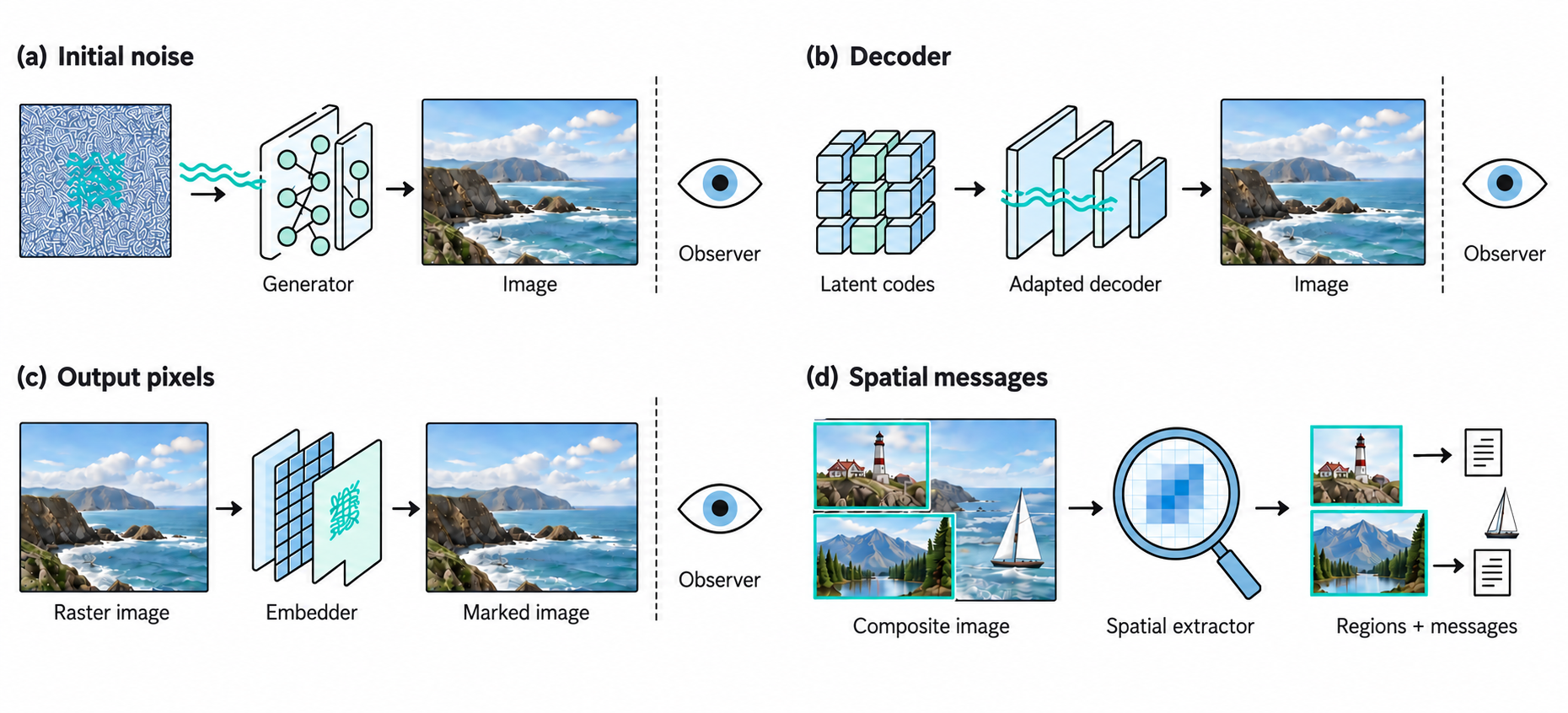}{\begin{tabularx}{\linewidth}{@{}lX@{}}\toprule
Initial noise & Marked sampling state; verification may require model inversion.\\
Decoder & Adapted generation component; a learned extractor reads the output.\\
Output pixels & A raster embedder modifies the exported image.\\
Spatial messages & A spatial extractor associates regions with message predictions.\\
\bottomrule\end{tabularx}}
\caption{\textbf{Image watermarking changes different objects.} Initial noise, an adapted decoder, and raster pixels are distinct intervention points \cite{treering2023,stablesignature2023,trustmark2023}. Spatial extraction can return localized messages \cite{wam2024}. Eye icons indicate the received images. The text specifies the keys, models, or extractors that each method requires. Teal traces identify conceptual carriers, not measured image residuals.}
\label{fig:image-stages}
\end{figure}

\subsection{Executable and declarative source}

\textbf{Source watermarking has a well-defined observation space of its own.}
SWEET places token-level marks at selected high-entropy positions \cite{lee2024sweet}. STONE gates marking by syntax category \cite{kim2026stone}. AST-guided code watermarking learns token placement with structural guidance \cite{guo2025astwatermark}.

Transformation-based approaches instead encode through choices among source representations. ACW uses transformations such as refactoring and reordering \cite{li2024acw}. SrcMarker combines structural transformations with variable substitutions \cite{yang2024srcmarker}. These mechanisms identify information in suspect code. Readers should not take their utility and attack studies as guarantees that one can extract the same information from executed frames.

For a drawing program, if the declared execution and observation rules ignore an allowed identifier change, that change leaves the observed image unchanged. A pixel-only verifier gains no distinguishing signal from that change at fixed inputs and environment. This statement concerns the specified observation. A source inspector, a metadata reader, or an execution monitor may still distinguish the files. Reflection, external state, and timing prevent us from declaring arbitrary source transformations render-preserving without examining the interpreter.

Source interventions can instead affect visible parameters. A program may select a line style, texture, layout, or motion pattern that carries information. Then the central constraint becomes task preservation. The change must preserve the chart's values, the explanation's meaning, or the animation's required events. That possibility requires its own encoding and verification design. Existing code-token detection supplies useful precedents for source control and attack access, but does not settle this downstream visual problem.

\subsection{Assets and representations designed for rendering}

\textbf{Rendering can transmit a mark deliberately placed in an asset.}
CopyRNeRF modifies a radiance field's color representation and trains a decoder on rendered image patches \cite{luo2023copyrnerf}. WateRF likewise studies marks recovered from rendered views of a modified representation \cite{jang2024waterf}. These are concrete precedents for a signal crossing the representation-to-image boundary. These methods permit visual changes and optimize rendering quality, so they do not conflict with the elementary observation (Section~\ref{sec:identical-observation}) that exact preservation of the verifier's input leaves no new signal.

GaussianMarker makes the observation distinction especially explicit by using separate decoders for rendered images and Gaussian parameters \cite{huang2024gaussianmarker}. The parameter decoder has information unavailable to an ordinary image viewer. 3D-GSW \cite{jang2025gsw} optimizes a modified Gaussian representation through rendering. A fixed image-message decoder reads a frequency subband of the resulting view. These methods concern marked representations and selected transformations and do not authenticate every program that later loads the representation.

An asset can also cross into another generator. LoT-Pass studies source-image marks passing through image-to-video generation \cite{lotpass2025}. The wider implication for hybrid pipelines is precise: each downstream operation defines a further transmission channel. Rendering the same marked asset again can retain its signal. Replacing it with clean material, rebuilding the scene, or bypassing an added representation component is a different attacker capability. A provenance assertion should name the marked contribution and the allowed operations on it.

\subsection{Execution, rendering, and export}

\textbf{The renderer can be an intervention point and an evidence producer.}
The inspected Remotion server-rendering path separates composition selection, evaluation at a frame index, frame capture, and media encoding \cite{remotionSSR,remotionRendererMedia,remotionRendererPerFrame,remotionRendererCapture}. These boundaries identify places where a system designer could modify emitted pixels or record an execution event. They are observations about the workflow, not a claim that Remotion implements such a watermark. Its browser-side rendering path has a different implementation, which reinforces the need to declare the actual backend \cite{remotionClientRender}.

A managed renderer could apply a keyed pixel mark, issue a signed receipt, or attach a content credential. These mechanisms produce different observations: decoded pixels, a separately retained record, and structured file-associated evidence. C2PA provides a framework for signed provenance assertions and content binding (Appendix~\ref{app:glossary}) \cite{c2pa}. Reading a credential and recovering a pixel message remain distinct verification operations. Image/video generation services expose analogous service and export boundaries even when no user-authored program takes part.

The supported assertion must follow the intervention. A renderer's evidence may establish that it processed an artifact under stated conditions. Establishing who authored the input, which assets contributed, or whether the producer invoked a particular model requires additional bindings. If callers can export through another renderer, the availability of that bypass is part of the attacker model.

\subsection{What the stages imply together}

\begin{takeawaybox}
\textbf{Takeaway: Follow the signal through the observed path.} Declare the marked object, the verifier's observation (with its side information), the allowed downstream transformations, the allowed visual changes, and the assertion under test. These items belong to the verification specification of Section~\ref{sec:formal-contract}. Source inspection, retained-asset identification, and signed export records play complementary roles, and a comparison among them must expose their different observations.
\end{takeawaybox}

\textbf{None of these three sources of evidence establishes complete authorship.} A service may add no mark, or callers may bypass the service, and someone may copy a marked component into another composition. VideoMarkBench separates removal, forgery, and attacker access \cite{videomarkbench2025}. Application constraints matter too: a signal acceptable in a photograph may be unacceptable in a scientific plot or timed explanation.

\section{Provenance Across Visual Media Applications}
\label{sec:applications}

Photographs, plots, SVG diagrams, and narrated animations allow different changes and expose different verifier observations. The following examples turn production-stage comparisons into application-specific research settings; they are not experimental findings.

Figure~\ref{fig:real-examples} makes these settings concrete with selected video frames and static graphics. Media from code-based rendering can incorporate elaborate scenes and imported assets, and an image model can generate a diagram. Thus visual style does not define the workflow.

\begin{figure}[p]
\centering
\begin{minipage}[t]{0.485\linewidth}
\raggedright\sffamily\small
\textbf{(a) Video generation}\par
{\footnotesize Veo 3.1 / official preview}\par\vspace{3pt}
{\color{black}\hrule height 0.35pt}\vspace{2pt}
\conceptasset{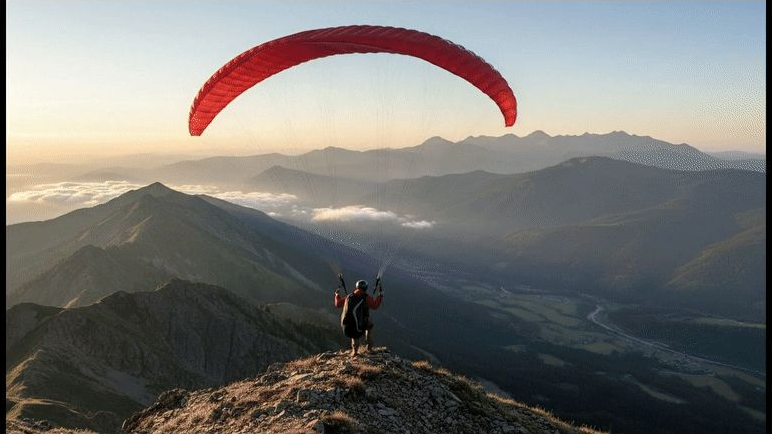}{\textit{Example image unavailable.}}
\end{minipage}\hfill
\begin{minipage}[t]{0.485\linewidth}
\raggedright\sffamily\small
\textbf{(b) Code-rendered music video}\par
{\footnotesize Claude / creator-reported workflow}\par\vspace{3pt}
{\color{black}\hrule height 0.35pt}\vspace{2pt}
\begin{tikzpicture}
\node[inner sep=0,anchor=south west] (mvframe) at (0,0)
  {\conceptasset{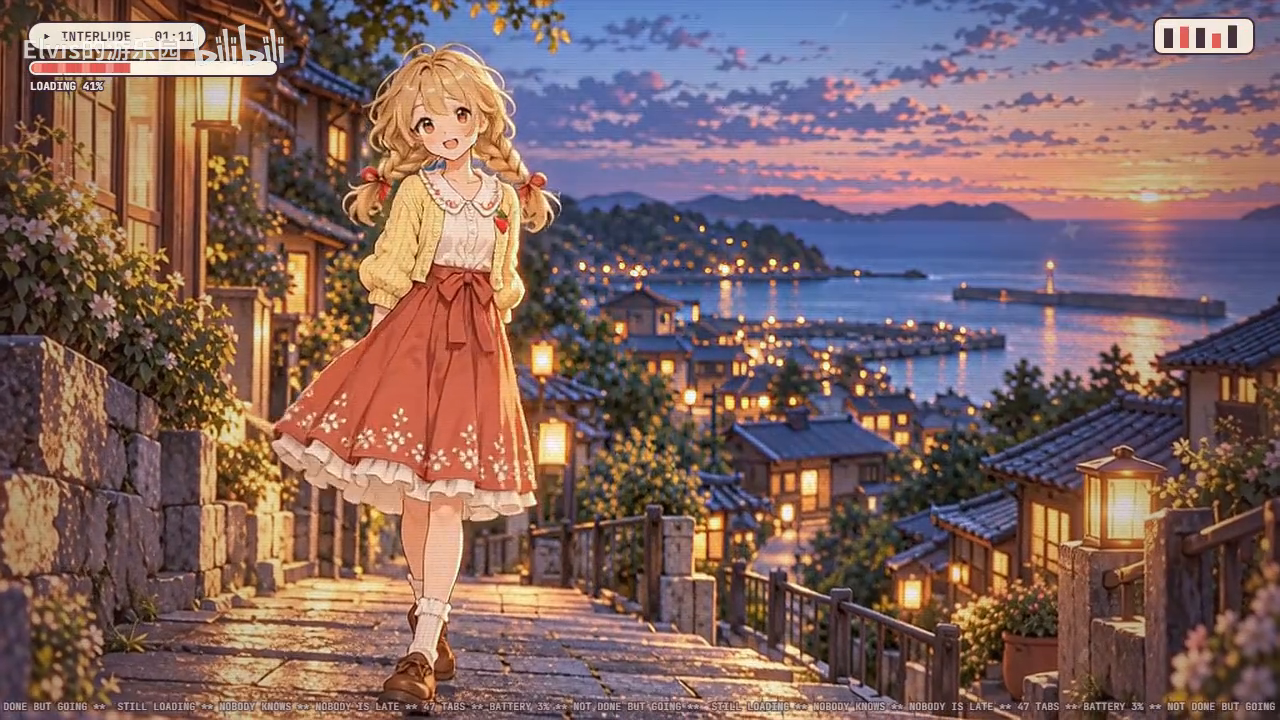}{\textit{Example image unavailable.}}};
\begin{scope}[x={(mvframe.south east)},y={(mvframe.north west)}]
\clip (0.014,0.870) rectangle (0.229,0.975);
\node[inner sep=0,anchor=south west] at (0,0)
  {\conceptasset{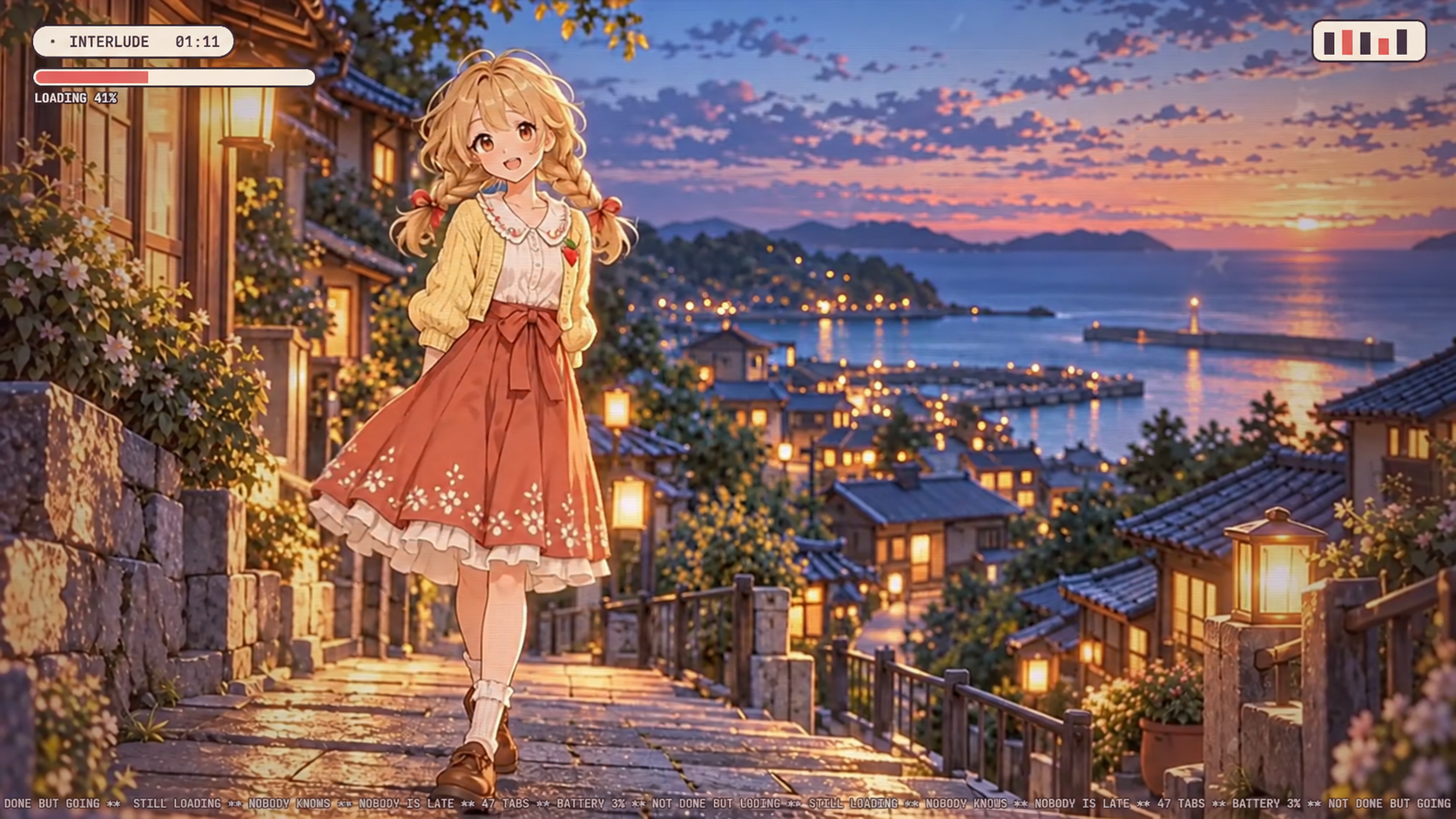}{}};
\end{scope}
\end{tikzpicture}
\end{minipage}\par\vspace{10pt}
\begin{minipage}[t]{0.485\linewidth}
\raggedright\sffamily\small
\textbf{(c) Code-rendered map animation}\par
{\footnotesize Claude Code + Remotion}\par\vspace{3pt}
{\color{black}\hrule height 0.35pt}\vspace{2pt}
\conceptasset{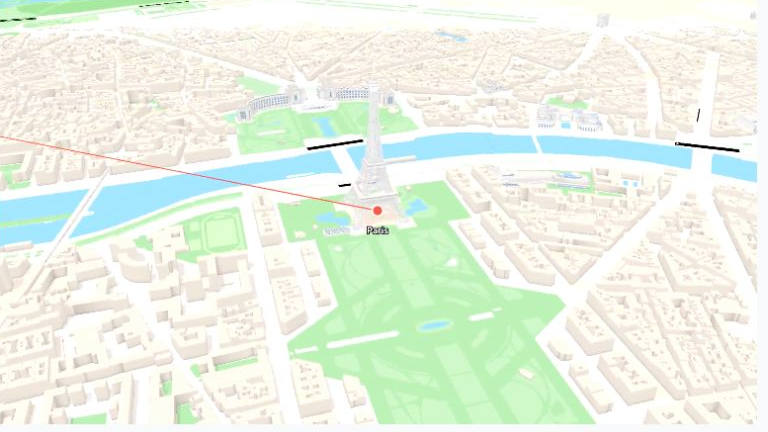}{\textit{Example image unavailable.}}
\end{minipage}\hfill
\begin{minipage}[t]{0.485\linewidth}
\raggedright\sffamily\small
\textbf{(d) Code-rendered 3D visualization}\par
{\footnotesize Claude Code + Remotion / X example}\par\vspace{3pt}
{\color{black}\hrule height 0.35pt}\vspace{2pt}
\conceptasset{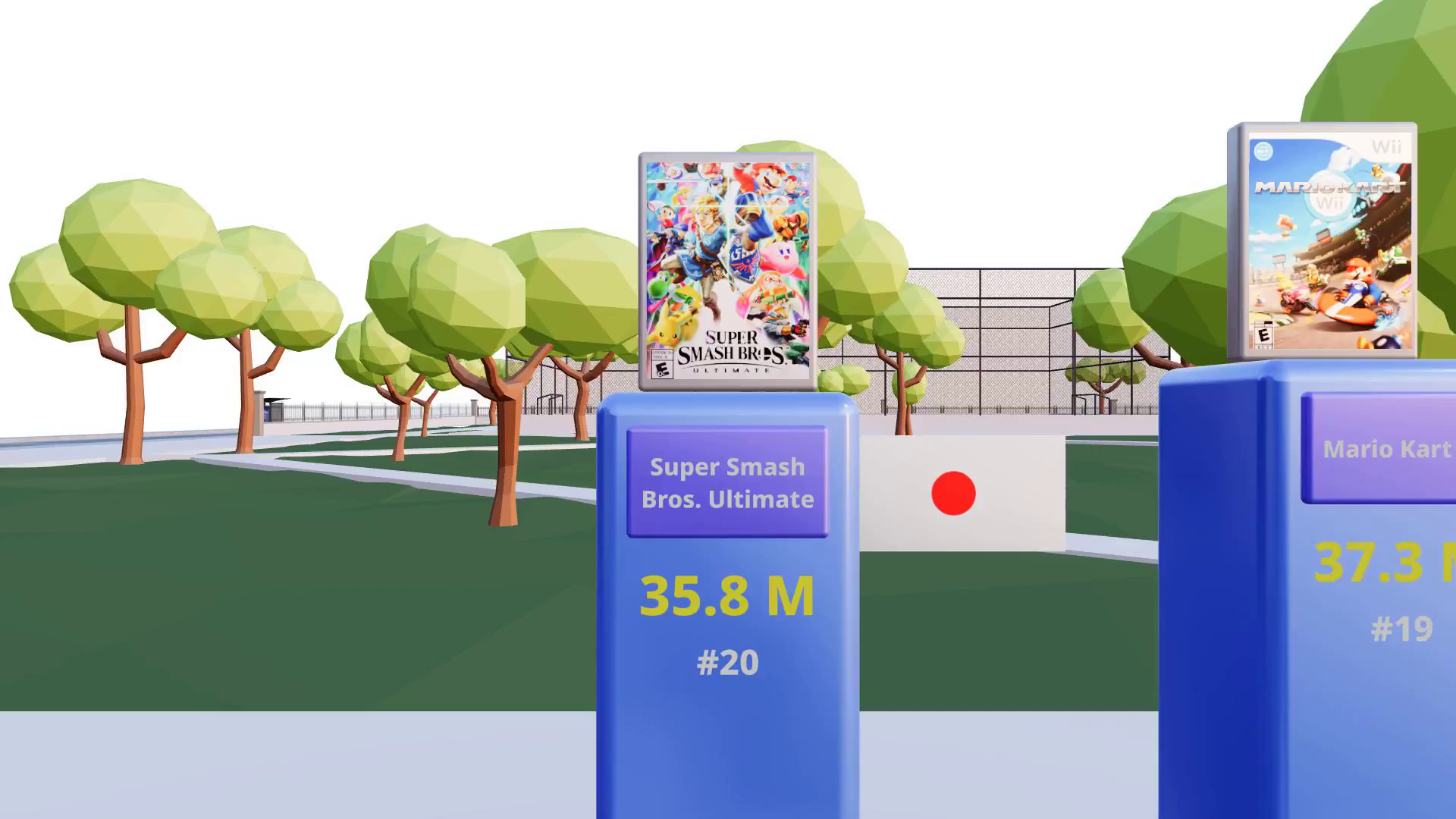}{\textit{Example image unavailable.}}
\end{minipage}\par\vspace{10pt}
\begin{minipage}[t]{0.485\linewidth}
\raggedright\sffamily\small
\textbf{(e) Image-model output}\par
{\footnotesize Diagram generated for this article}\par\vspace{3pt}
{\color{black}\hrule height 0.35pt}\vspace{2pt}
\conceptasset{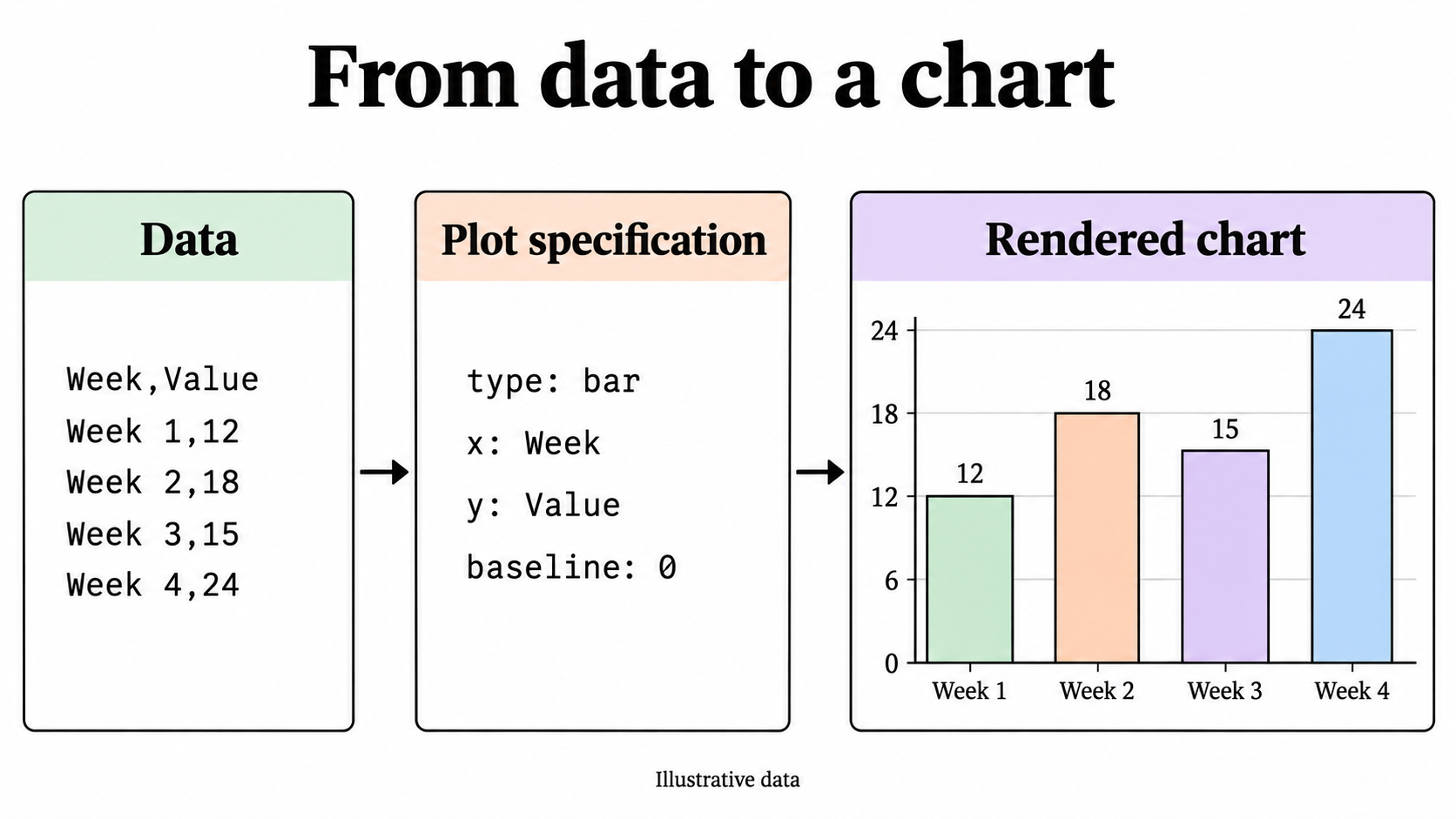}{\textit{Example image unavailable.}}
\end{minipage}\hfill
\begin{minipage}[t]{0.485\linewidth}
\raggedright\sffamily\small
\textbf{(f) LLM-authored SVG}\par
{\footnotesize Editable source, then local rasterization}\par\vspace{3pt}
{\color{black}\hrule height 0.35pt}\vspace{2pt}
\conceptasset{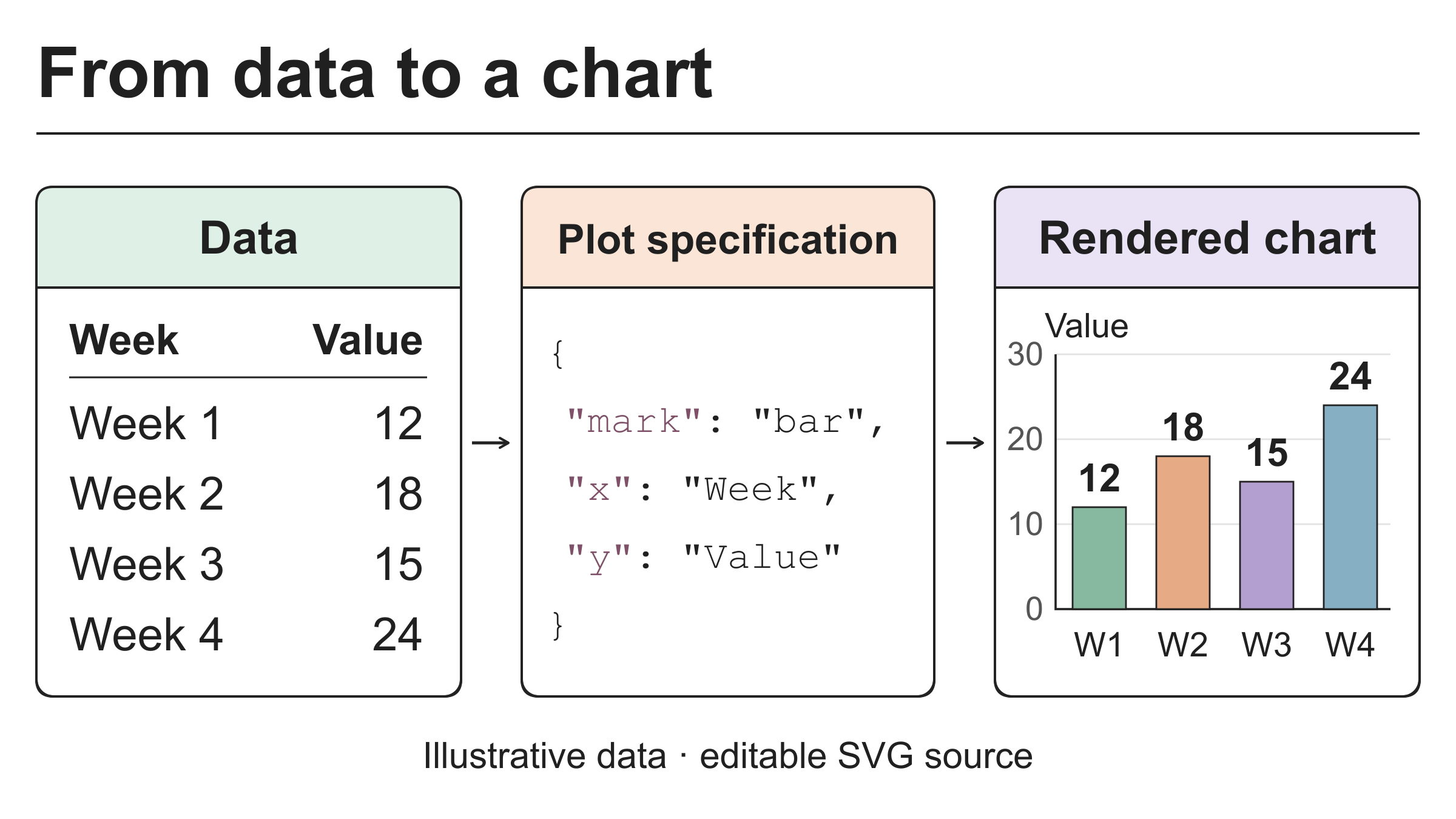}{\textit{Example image unavailable.}}
\end{minipage}\par\vspace{10pt}
\caption{\textbf{Selected visual outputs and their production workflows.} The panels show workflow variety, not a controlled quality or watermark comparison. Labels follow source records, not visual appearance. (a) A frame from Google\textquotesingle s Veo 3.1 preview \cite{googleVeo31Example}. (b) A frame of \emph{still loading} at 01:11.44, reported by its creator as Claude code-rendered \cite{gao2026StillLoading}. The upper-left platform mark is retouched for display. Incorporated asset origins are not established here. (c) A travel-map animation credited to @JNYBGR \cite{remotionTravelGallery}. (d) Dilum Sanjaya\textquotesingle s X video, with the frame obtained from Remotion\textquotesingle s gallery \cite{sanjaya2026GamesVideo,remotionGamesGallery}. (e) An illustrative diagram produced for this article with an image-generation tool. (f) An illustrative diagram produced for this article as LLM-authored SVG; the SVG source is retained.}
\label{fig:real-examples}
\end{figure}

\subsection{Image and video generation}

\textbf{Four conditions shape a study of generated images and video: intended use, audio (for narrated video), verifier access, and the assertion under test.}
Evaluate visual fidelity and provenance at the intended use. For a generated illustration, a task might allow small texture changes while requiring the depicted objects and their relations to remain intact. Video adds continuity, event order, and timing.

These requirements explain why image-watermark extraction alone cannot establish video suitability. A signal that is unobtrusive in individual frames may vary distractingly over time, and a detector may rely on frames removed during editing. The spatial and temporal designs of Video Seal and VideoMark illustrate different responses to these concerns \cite{videoseal2024,videomark2025}.

Narrated video adds an audio observation channel. LAVA studies alignment, channel-reliability gating, and fusion of semi-fragile audio--visual watermark scores \cite{zeng2026lava}. Its evaluated offset selection uses an oracle; label-free alignment and calibrated fusion under joint channel failure remain deployment questions.

The verifier's access also changes the problem. A recipient with a downloaded clip can aggregate frame evidence, but someone with one screenshot cannot use temporal redundancy. A model-owning verifier may attempt inversion, but a third party may possess only a learned decoder and the received pixels.

Passive inference of synthetic footage is distinct from message recovery, the estimation of an embedded message. GenVidBench and UNITE study declared synthetic-video detection populations \cite{genvidbench2025,unite2025}. Their labels do not automatically resolve whether an LLM authored a program behind another clip. Declare the assertion under test before choosing the detector: a visible generated contribution, a generation event, or an issuer association.

\subsection{Plots, educational visuals, and SVG}

\textbf{A visually small change can be a substantive error in an explanatory graphic.}
Consider a plot teaching the relation between position and velocity. Moving a data point, changing an axis label, or shifting a highlighted event may alter the lesson even when the image remains attractive. A task specification could instead allow changes to a background pattern or line decoration while fixing values, units, labels, and the intended correspondence. Whether those freedoms support reliable marking is an application-specific question. Perceptual similarity alone does not define success.

Graphics tools expose several observation interfaces. Matplotlib documents both raster and vector export formats \cite{matplotlibSavefig}. Claude's documented custom-visuals feature constructs HTML-based visuals and supports SVG or HTML downloads \cite{claudeCustomVisuals}; Section~\ref{sec:systems} gives the documentation details. A verifier that receives such source material can inspect its structure and text, but one that receives a raster export sees the rendered effect. SVG's document and rendering structures are distinct. Its supported features make the grammar, resources, and rendering environment relevant to any precise equivalence claim \cite{w3cSVGRender}.

For example, a marker in an ignored SVG comment could remain in the downloaded source and disappear from a screenshot. A marker encoded in rendered color choices has a path to that screenshot, subject to the capture and processing channel. Thus re-exporting the same marked description differs from recreating the graphic from its numerical data. A task check of the plot's truth and a watermark detector's check of a signal need explicit, potentially different acceptance rules.

A rendered plot can thus receive a raster watermark even when its source is ordinary SVG or plotting code. That choice protects the marked export under a specified image channel. A fresh render from an unmarked source bypasses that postprocessing stage (Section~\ref{sec:stages}). Conversely, a source-level mark that changes no observed pixels gives a pixel-only verifier no new signal. Designing a mark that crosses this boundary requires selecting allowed visible changes and the task they must preserve.

\subsection{Programmable video}

\textbf{Exported video may conceal its construction history.} Figure~\ref{fig:rerendering} separates pixel processing, asset reuse, and clean-source export. In Claude Code--Remotion workflows, the agent edits a project and the renderer produces media \cite{claudeCodeOverview,remotionCodingAgents}. Section~\ref{sec:systems} gives the details. Project access exposes code, assets, and configuration that a clip-only viewer lacks. Different histories can produce identical depicted events, so source-mark recovery and animation attribution are separate questions.

In a teaching animation, the ball must reach its target on time and the caption must describe its motion correctly. Both frame-indexed rendering and video generation can satisfy this task. A study must declare the allowed changes to frame rate, duration, trajectories, captions, and synchronization. It must distinguish retiming of the exported clip from source editing followed by rerendering.

Exporting clean source and assets can bypass output-only marking. Rerendering a marked texture may retain its signal. Neither establishes a general robustness advantage for programs. A guarantee must identify the constrained components and the production event that a retained signal supports.

\begin{figure}[!htbp]
\centering
\conceptasset{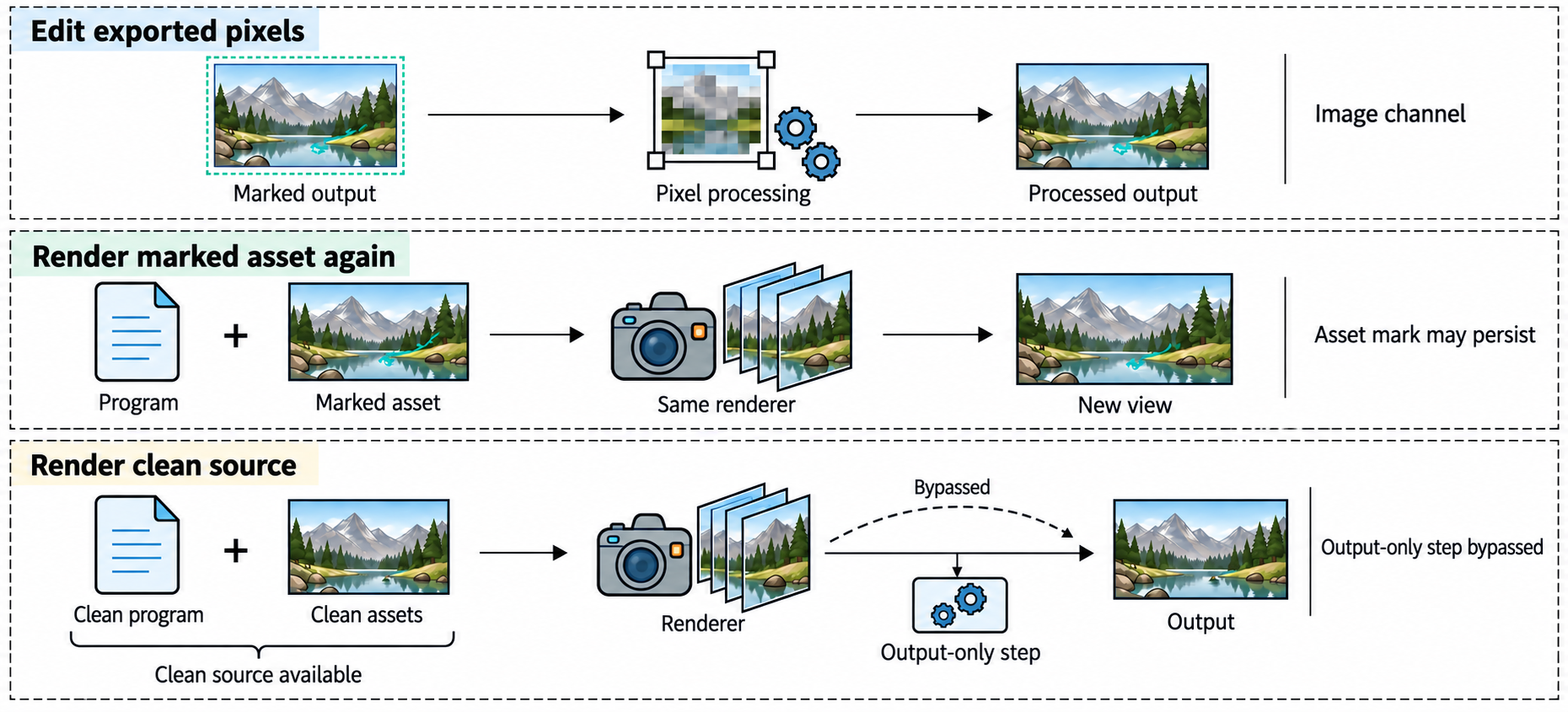}{\begin{tabularx}{\linewidth}{@{}lX@{}}\toprule
Edit pixels & Received marked output passes through an image-processing channel.\\
Reuse asset & A fresh view of a marked representation may retain its carrier.\\
Use clean source & With clean source and assets available, a fresh export can bypass an optional output-only mark.\\
\bottomrule\end{tabularx}}
\caption{\textbf{Three different transformation interfaces.} Editing received pixels, rendering a marked asset again, and exporting available clean source grant different access. Asset-mark persistence depends on the channel. The bottom path bypasses an optional output-only marking step under the stated clean-source access; it does not establish general erasure of upstream marks.}
\label{fig:rerendering}
\end{figure}

\subsection{Agent-composed hybrid media}

\textbf{An agent can coordinate both routes within one artifact.}
Official tool documentation separately describes code execution that produces graph images and specialized image generation \cite{openaiCodeInterpreter,openaiImageGeneration}. Remotion's Prompt to Video template supplies a concrete composition workflow involving generated scripts, images, and voiceover \cite{remotionPromptToVideo}. These examples motivate treating an LLM agent as a controller choosing operations, rather than as a third, mutually exclusive visual-production route.

Instead of pixels, agent attribution can target an execution record. TRACE combines content-keyed action selection with a trajectory-group record-count channel under a specified log-editing threat model \cite{gao2026trace}. Its protected executed-action and no-re-execution assumptions distinguish this trajectory evidence from a viewer's pixel-only observation.

Figure~\ref{fig:hybrid-evidence} illustrates an agent assembling a model-generated background, a plotted curve, and human annotations into a lesson. The plot's data must remain correct, the annotations must align with the explanation, and the background must retain its intended role. A recovered background mark concerns that contribution and does not establish the authorship of the curve, the annotations, or the entire lesson. Cropping, occlusion, and transitions can also change which contribution is observable. Localized image-message extraction in Watermark Anything supplies a relevant precedent, while leaving video synchronization and complete composition history as separate obligations \cite{wam2024}.

Recipients may receive pixels, assets, or production records. The interface must identify the covered contributions and the unobserved production events. Signal survival follows the media path. Task preservation follows each component's purpose. Authentication identifies who vouches for the history.

\begin{figure}[!htbp]
\centering
\conceptasset{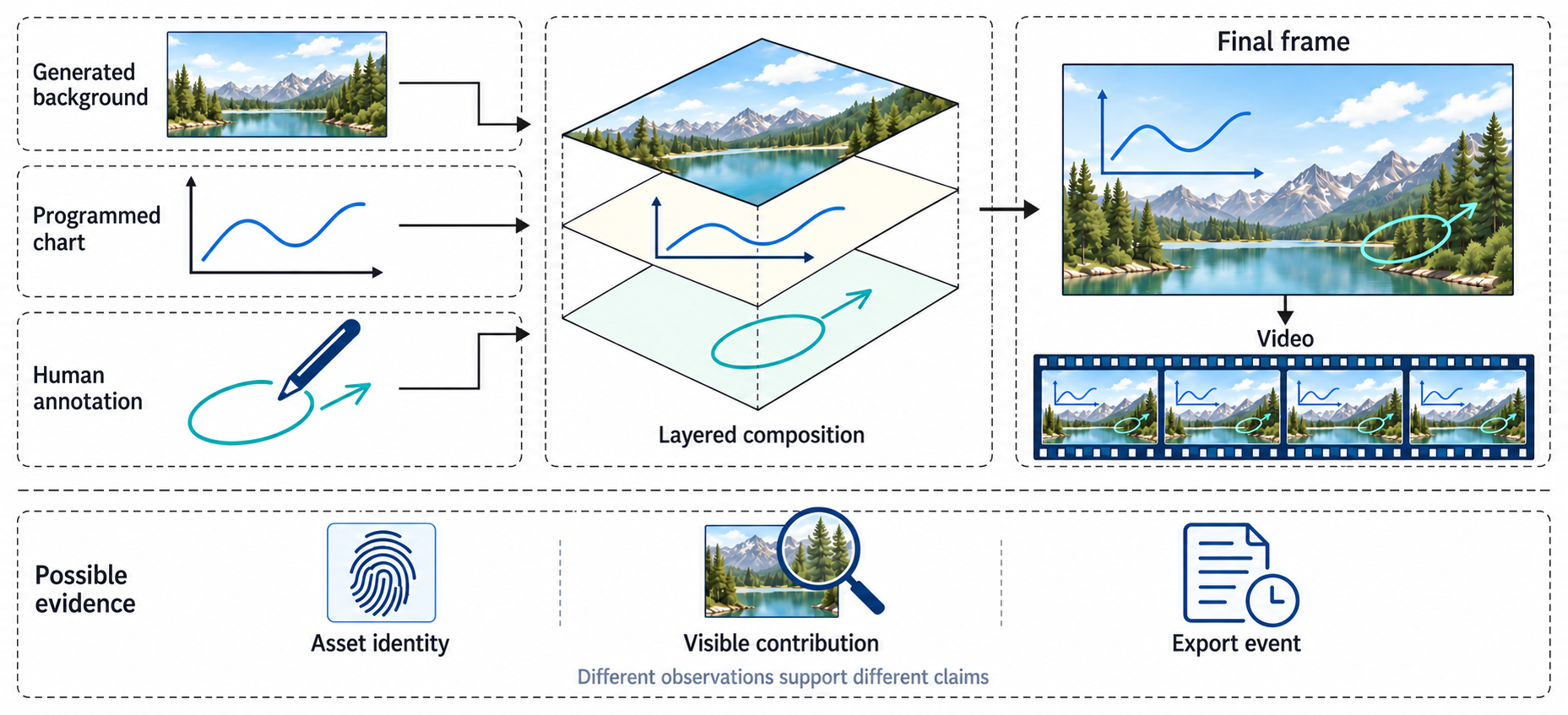}{\begin{tabularx}{\linewidth}{@{}lX@{}}\toprule
Generated background & Evidence may identify a particular contributing asset.\\
Programmed chart & Data values and visual correspondence impose task constraints.\\
Human annotation & A separate contribution can have a different history.\\
Export event & A record can concern one invocation without authenticating every ancestor.\\
\bottomrule\end{tabularx}}
\caption{\textbf{One composition supports several evidence questions.} A generated background, programmed chart, and human annotation enter a shared visual artifact. Asset identity, retained visible contribution, and an export event are different possible targets. Their icons indicate evidence categories, not verified facts about the illustrated example or complete authorship.}
\label{fig:hybrid-evidence}
\end{figure}

\section{Documented Systems and Deployment Boundaries}
\label{sec:systems}

Official documentation and selected repository sources, accessed on 4 October 2026, identify model outputs, rendering components, and inspectable artifacts. They establish interfaces and implementation boundaries, not adoption frequency or watermark performance. Each case links an operation and its accessible state to the assertion it could support and the evidence still needed.

\subsection{Claude Code and server-side Remotion rendering}

\textbf{The assistant supplies code; a separate pipeline renders it.} Remotion's tutorial covers Claude Code project creation, skills, and preview setup, and Claude Code documents file editing and command execution \cite{remotionCodingAgents,claudeCodeOverview}. Remotion separately documents component-code generation, compilation, and rendering \cite{remotionCodeGeneration}. That provider-specific example is not a Claude API example.

In a public creator walkthrough, Sanjaya describes opening a Remotion project in Claude Code and iterating on the requested animation \cite{sanjaya2026GamesWorkflow}. Sanjaya then describes adding images and a Sketchfab background model \cite{sanjaya2026GamesAssets}. This creator report illustrates why code authorship, asset origin, and final rendering are separate provenance questions.

The server API separates bundling, composition selection, and rendering \cite{remotionSSR}. Source files, dependencies, resources, inputs, and configuration affect execution, and media controls include input properties, codec, metadata, and frame range \cite{remotionRenderMediaAPI}. Operators may inspect these objects, but viewers receive the export. A stored source version alone does not establish the configuration used for that file.

In the inspected repository snapshot, media export renders and stitches frames, and related sources manage the browser, seek, and capture each frame \cite{remotionRendererMedia,remotionRendererFrames,remotionRendererPerFrame,remotionRendererCapture}. These snapshot-specific observations locate potential evidence at source preparation, frame production, and export. They do not establish which invocation produced a delivered file or describe every released package.

\subsection{Browser and cloud rendering change the trust boundary}

\textbf{Rendering backends expose different boundaries.} Remotion's browser backend takes a component, configuration, and input properties through \texttt{@remotion/web-renderer}. It documents WebCodecs and Mediabunny, skips server bundling, and restricts supported HTML \cite{remotionClientRender}. A successful preview thus does not establish exportability through this backend. The component, runtime, and supported operations remain part of the production specification.

Remotion Lambda documents S3 deployment, parallel chunk rendering, stitching, and S3 output \cite{remotionLambda}. Evidence from one cloud component concerns that component's work. Binding it to the final file requires relationships to the other chunks and export. Changing the rendering backend changes execution, observation, and deployment control without itself changing program authorship.

\subsection{Static graphics expose multiple kinds of artifact}

\textbf{The same distinction appears before time is introduced: static images already involve distinct operations and distinct observable artifacts.} OpenAI documents Code Interpreter as executing Python and producing files such as graph images \cite{openaiCodeInterpreter}. OpenAI separately documents an image-generation tool based on GPT Image models \cite{openaiImageGeneration}. These are distinct available operations. Identifying the assistant that returned a picture does not specify which operation produced it, and a proposed tool invocation does not establish successful execution.

Claude's custom-visuals documentation describes HTML-based diagrams and charts with SVG or HTML downloads \cite{claudeCustomVisuals}. A downloaded description, an interactive display, and a screenshot expose different state. The SVG specification distinguishes the document tree from the rendering tree \cite{w3cSVGRender}. Matplotlib's export documentation likewise distinguishes formats whose availability depends on the backend \cite{matplotlibSavefig}. Source-visible evidence can therefore concern an editable description even when a pixel-only verifier receives only its appearance.

The consulted Artifacts help page lists particular export formats but does not supply a general MP4-export promise \cite{claudeArtifactsHelp}. Broader capability claims would require additional evidence.

\subsection{Generated assets inside authored compositions}

\textbf{Mixed workflows need contribution and handoff evidence.} Remotion's Prompt to Video combines generated script, images, and voiceover using OpenAI and ElevenLabs \cite{remotionPromptToVideo}. Its markup guidance permits media resources \cite{remotionSkillsMarkup}. Generated assets thus enter an authored composition before rendering and encoding.

Asset producers, composition authors, controllers, and exporters observe different stages. Producers may not see later cropping, overlays, or reuse. Exporters may lack authenticated input histories. Evidence must connect asset identity, retained contribution, and export scope, because recorded ancestry alone does not establish visibility.

\subsection{Upstream marking and downstream credentials}

\textbf{Existing upstream evidence does not settle its downstream interpretation.} Anthropic documents text watermarking for supported models and C2PA credentials for supported generated files, with coverage and capabilities depending on the model or platform \cite{claudeMarking}. The documentation also qualifies what detection establishes and notes that transformations such as conversion or screenshots can remove metadata. These statements make upstream marking part of the production analysis. Source-to-render inheritance still needs examination for each carrier and workflow.

C2PA also supports watermark or fingerprint \emph{soft bindings} (Appendix~\ref{app:glossary}): a retained manifest may be located after metadata loss, with region or time scopes available \cite[Sections~9.3 and~18.10]{c2pa}. Soft matching complements signed assertions and does not replace hard cryptographic binding. Thus metadata loss alone is not a sufficient motivation for a new credential mechanism. Questions~\ref{q:hybrid} and~\ref{q:credentials} must compare with signed records plus soft-binding lookup, while specifying matching accuracy, authorized production-event relations, repository access, and privacy. These capabilities do not imply that a source-code watermark survives external video rendering.

\subsection{Permissions are part of the verification specification}

\textbf{Deployment determines who can inspect source, invoke rendering, retain logs, and issue evidence.} Public verification and secret issuance are distinct arrangements. If a proposed design places an issuance secret inside client code available to its attacker, its security argument must account for that access. A managed service can impose different access conditions. The specification must state these conditions alongside the trust placed in its execution records.

\begin{takeawaybox}
\textbf{Takeaway:} These are modeling requirements, not claims that the documented products implement a particular watermark protocol. To compare image/video generation services, editable programs, and compositions, we need a verification specification (Section~\ref{sec:formal-contract}) in which source visibility, clean assets, issuance queries, and verifier access each have a separate statement.
\end{takeawaybox}

\section{Challenges and Research Agenda}
\label{sec:agenda}

Intervention, decoding, and provenance assertions can concern different production stages, so the ten Questions below match these objects before optimizing detection:
\begin{itemize}
\item[\ref{q:identifiable}.] Identifiable provenance labels.
\item[\ref{q:downstream}.] Downstream observability.
\item[\ref{q:vector-raster}.] Vector descriptions and raster observations.
\item[\ref{q:fair-comparison}.] Fair comparison across intervention stages.
\item[\ref{q:capacity}.] Recoverable payload under task constraints.
\item[\ref{q:reconstruction}.] Reconstruction with access to production state.
\item[\ref{q:synchronization}.] Synchronization with calibrated search.
\item[\ref{q:composed-edits}.] Composed edits with a finite budget.
\item[\ref{q:hybrid}.] Hybrid local contribution.
\item[\ref{q:credentials}.] Content-bound, privacy-aware credentials.
\end{itemize}

A Question states its Setting, the Success criterion, a Baseline, and the Deliverable sought, and, where useful, a First case and what is Not asked, meaning the variants it excludes. In Question~\ref{q:reconstruction}, Success describes the attack.

Finite instances are starting checks. The intended deliverable is:
\begin{itemize}
\item a structural criterion;
\item a bound that varies with declared parameters;
\item an efficient construction; or
\item a separation from an explicit baseline.
\end{itemize}
Enumerating one instance alone supplies none of these. The formulation of a Question does not certify that every variant is globally open. A natural first project is Question~\ref{q:downstream}'s carrier recoverable across a population without clean reference frames. Question~\ref{q:capacity} then asks how its payload scales.

\subsection{Identifiability and downstream observability}

\textbf{A provenance distinction must survive the declared observation.}
Passive inference starts with production distributions and a label, not merely a contrast between product names. Appendix~\ref{app:observation} shows both extremes: identical label-conditional observation laws give no classification advantage over the prior alone (Corollary~\ref{cor:label-no-advantage}), but different mixtures of the same renderer can remain distinguishable (Example~\ref{ex:mixture}). Side information may change the problem. These distinctions motivate a precise label and an explicit abstention policy.

\begin{question}[Identifiable provenance labels]
\label{q:identifiable}
\qlabrun{Ask} How accurately can a verifier infer the production route, image/video generation or code-based rendering, from its observation if it may abstain?

\qlab{Setting} Both routes draw intended content from one parameterized scene family by a common rule. Declare the nuisance-policy families, including whether they may differ by route and admit randomized mixtures. A nuisance policy is a rule by which a route chooses implementation details, such as which program or seed realizes a scene, that do not change the intended content. The verifier observes the received artifact plus any declared side observations. Compare a fixed menu of such side observations with stated trust and disclosure costs; grant no free authenticated route-label bit.

\qlab{Success} Minimize the worst-case, equal-prior probability of a wrong, non-abstaining report. Take the worst case over the declared nuisance-policy families, with scene content drawn by the common rule in both routes. Require non-abstention coverage at least $\kappa\in(0,1]$ under each route and allowed nuisance policy.

\qlab{Baseline} For fully known finite laws without abstention, Proposition~\ref{prop:bayes} gives the Bayes benchmark.

\qlab{Deliverable} A uniform criterion or bound as scene and nuisance budgets vary. Exhibit both separated and indistinguishable families.

\qlab{Not asked} Hybrid histories, a separate declared class.
\end{question}

Active marking, unlike passive inference, asks which allowed changes to a representation reach the decoder after the remaining pipeline. The answer can depend on the rendering environment, quantization, and what the verifier receives. Analyze separately a source-only distinction that disappears in rasterization and an asset intervention designed to affect pixels.

\paragraph{Running task.}
The running task is a bar-chart animation with prescribed values, labels, and monotone growth, $T$ sampled frames, and $q$ quantized height levels under a fixed rasterizer. Task-preserving timing variation supplies a candidate carrier that quantization and edits can remove. Example~\ref{ex:monotone-bars} in Appendix~\ref{app:rendering} gives an explicit task-allowed trajectory. A program family can expose its timing coefficients directly. We do not assume that an image/video model realizes the same trajectories. A comparison using an image/video model must declare its decoding map and reachable interventions. The task connects observability, payload, synchronization, and reconstruction without identifying them with one another.

\begin{question}[Downstream observability]
\label{q:downstream}
\qlabrun{Ask} For a declared population of finite scenes, which task-preserving interventions remain distinguishable after a specified downstream map and edit class?

\qlab{Setting} Use one encoder/decoder family across that population. The decoder receives pixels, a key, and public population parameters, but no instance-specific clean render or coordinates. A comparison using an image/video model must declare its decoding map, reachable interventions, and any inversion access.

\qlab{Success} Different messages must remain distinguishable even when the unknown scene states differ.

\qlab{Baseline} Source changes that leave the observation unchanged give the zero-information baseline (Corollary~\ref{cor:exact-equality}).

\qlab{First case} In the bar-animation case, vary allowed timing coefficients while fixing values and labels.

\qlab{Deliverable} A carrier with a uniform recovery condition as sampling precision and edit budget vary, and a matching ambiguity example.
\end{question}

\subsection{Vector descriptions and raster exports}

\textbf{A source-format guarantee is not automatically an exported-image guarantee.}
SVG makes the observation change concrete. Source editing, canonicalization, object transformation, and raster export act on different representations. Identifiers may affect references or styling, but a comment is invisible only to an interpreter that ignores it. A bounded grammar, resource policy, and renderer give a useful first case.

\begin{question}[Vector descriptions and raster observations]
\label{q:vector-raster}
\qlabrun{Ask} Which allowed line, color, or timing parameters retain distinguishable messages in the raster observation?

\qlab{Setting} Fix a bounded script-free vector grammar. Suppose raster observation is a known, computationally simulable channel of the received marked source and other available inputs, with noise jointly independent of the message and of all inputs of the source verifier. Canonicalization (rewriting the source into a standard form) must preserve rendering in every allowed environment; values and labels stay fixed.

\qlab{Success} Retain distinguishable messages uniformly across the declared family of renderer environments and quantizers.

\qlab{Baseline} A source verifier can compose that channel with the raster decoder, matching its conditional performance (Lemma~\ref{lem:simulation}).

\qlab{First case} The first nontrivial case asks for one encoder and raster decoder across a declared family of renderer environments and quantizers, without revealing the chosen environment or the clean instance to the decoder.

\qlab{Deliverable} A structural condition or payload bound as environment variation grows, with a counterexample at failure.

\qlab{Not asked} Source recovery: report it separately when its attack channel differs.
\end{question}

\subsection{Fair comparison, payload limits, and reconstruction}

\textbf{Compare interventions with matched information and permissions.}
Classical information embedding distinguishes decoding without the unmarked reference (host-blind) from known-host decoding and studies rate--distortion--robustness tradeoffs \cite{chen2001qim}. The comparisons here must instantiate that baseline with task predicates, reachable interventions, and renderer environments. An upstream encoder may know state that an output-only encoder lacks, and both routes share a pixel-stage baseline. In Example~\ref{ex:hidden-anchor}, two scene descriptions render the same solid square but store different hidden anchor values that the baseline renderer ignores. This shows non-emulation, not a performance separation.

\paragraph{A comparison with fixed assumptions.}
Fix one visible bar, five frames, height levels $0,\ldots,8$, and fixed labels and position. The task allows monotone heights from $0$ to $8$, and the rasterizer maps distinct heights to distinct pixels. Compare
\[
 h^{(0)}=(0,2,4,6,8),\qquad h^{(1)}=(0,3,4,5,8).
\]

\noindent Under the identity attack, a pixel-only verifier distinguishes these two messages. Conditional on either route supplying the same baseline $h^{(0)}$, an output encoder allowed these pixel changes can realize this code. An ignored source comment cannot realize it, but a program exposing sampled heights can. A neural-state intervention needs a separate reachability argument.

This is message recovery, not calibrated detection: $h^{(0)}$ is also the unmarked baseline, so perfect recovery accepts that null with probability one. If a different attack specification grants clean-baseline replacement, both messages collapse to $h^{(0)}$. The example separates carrier visibility, stage access, and false-alarm control before comparing algorithms.

\begin{question}[Fair comparison across intervention stages]
\label{q:fair-comparison}
\qlabrun{Ask} When can one intervention stage simulate another stage's performance, and when does extra state or control yield a strict separation?

\qlab{Setting} Declare a common production-state family and, for each intervention stage of Section~\ref{sec:stages}, the encoder's view and allowed controls. All stages face the same task, attack, verifier, and key assumptions. Hold distortion and auxiliary-channel budgets fixed, and state computational limits. An auxiliary channel is a channel other than the video pixels, such as audio, metadata, duration, or additional precision (Definition~\ref{def:payload} and the paragraph after it).

\qlab{Success} Compare achievable triples of payload, worst-case message error, and false-alarm probability. Distinguish key-averaged criteria from fixed-key criteria.

\qlab{Baseline} The hidden-anchor example (Example~\ref{ex:hidden-anchor}, Proposition~\ref{prop:non-emulation}) supplies only non-emulation of a prescribed output. A video-only deterministic encoder cannot reproduce that selected upstream rule.

\qlab{Deliverable} A simulation result or a strict separation, with performance compared through the triples of Success. A performance separation must also rule out alternative downstream encodings.
\end{question}

\begin{question}[Recoverable payload under task constraints]
\label{q:capacity}
\qlabrun{Ask} For the bar-animation family, how does maximal recoverable payload vary with frame count, quantization, task-allowed timing variation, and edit budget?

\qlab{Setting} Let $\mathcal V_U(s)$ denote the task-allowed outputs, and restrict encoding further to outputs reachable through the chosen program or neural interventions. Use one blind decoder over the declared state family, with finite auxiliary-channel budgets.

\qlab{Success} In the deterministic zero-error starting case, take the union of observation sets over unknown states and attacks. These unions must be disjoint for different messages (Definition~\ref{def:payload}; Proposition~\ref{prop:zero-error-count}(i)).

\qlab{Baseline} Counting all quantized videos is only a loose upper bound (Proposition~\ref{prop:zero-error-count}).

\qlab{Deliverable} A parameter-dependent bound and a construction within a named reachable family.

\qlab{Not asked} Nonzero message error and false alarm under a declared null law, treated separately after the zero-error case.
\end{question}

\begin{question}[Reconstruction with access to production state]
\label{q:reconstruction}
\qlabrun{Ask} What access and budget let a reconstruction attack reduce average message recovery to a stated level while preserving the task?

\qlab{Setting} Fix a watermark scheme $(\mathrm{KeyGen},E,D)$, a declared source law, one task, and a reconstruction class with a stated cost. The message is independent and uniform in the key-averaged experiment of Section~\ref{sec:recovery}. Rerendering a marked asset, replacing it with clean material, and reconstructing a task-equivalent program grant different access. Specify which access the attack may use. Fix the attack policy before the draws, using only its declared view. The substantive case exposes only upstream state that carries the mark, such as a marked asset or latent.

\qlab{Success} Average recovery is at most $\beta\in[0,1)$, and the attack preserves the task.

\qlab{Baseline} Clean-asset replacement is a separate baseline. SHIFT instead starts from the received image, applying partial latent noising and stochastic reverse resampling with a pretrained diffusion model \cite{bao2026shift}. This image-access attack differs from rerendering known source or reusing a marked latent; comparisons must preserve those access and task-quality distinctions.

\qlab{First case} At most $b$ coefficient substitutions in a finite timing grammar, or a declared finite neural-intervention set.

\qlab{Deliverable} An access-dependent bound or separation, stating the mandatory trusted stages.

\qlab{Not asked} An ordered hierarchy of access types, which need not exist.
\end{question}
\subsection{Temporal synchronization and composed edits}
\label{sec:temporal}

\textbf{A video verifier must account for the clock it searches.}
Deletion, duplication, rate conversion, and interpolation change different relations between original time and observed frames. Periodic carriers can have phase ambiguities, although length, boundaries, or timestamps may supply origin information. For example, in an ideal scalar model, a cosine carrier of period eight frames cannot separate offsets that differ by eight frames (Example~\ref{ex:periodic} in Appendix~\ref{app:temporal}). Searching over candidate alignments creates a multiple-testing problem: validity for a fixed alignment need not imply validity for the selected report (Example~\ref{ex:adaptive}).

\begin{question}[Synchronization with calibrated search]
\label{q:synchronization}
\qlabrun{Ask}
What span $r$ of the retained original-time window and what redundancy suffice to meet the three error bounds under Success? Measure the span in original-time frame intervals.

\qlab{Setting}
Work under a stated carrier-measurement model. Redundancy is the factor by which the carrier repeats or encodes the message.

\qlab{Success}
Count abstention on the message or on the alignment as failure. Specify whether alignment means absolute origin or relative order. Require:
\begin{enumerate}[label=(\alph*),leftmargin=1.8em]
\item message error at most $\varepsilon$;
\item alignment error at most $\eta$;
\item whole-search false-alarm probability at most $\alpha$.
\end{enumerate}

\qlab{Baseline}
Repeated-message recovery and finite union-bound calibration (Proposition~\ref{prop:union}) are baselines.

\qlab{First case}
Begin with finite offsets and sampling ratios, and with at most one deletion (Definition~\ref{def:deletion-channel}). Include periodic schedules for which absolute alignment is impossible (Example~\ref{ex:periodic}).

\qlab{Deliverable}
Seek a quantitative tradeoff beyond the baselines. Bound the retained-sample count and the maximum gap separately; the maximum gap is the largest number of original-time frame intervals between consecutive retained samples. Compare program-time and latent-frame carriers with matched observations.

\qlab{Not asked}
Obtaining reliable carrier measurements from pixels is an additional obligation, not an assumed consequence of either route.
\end{question}

Robustness to individual edits is not a guarantee for their composition. Two crops expressed as fractions of the current clip consume a different budget from fractions of the original clip. For example, two maximal prefix crops that each delete a quarter of the current frames of a 64-frame clip leave 48 and then 36 frames, so together they delete seven sixteenths of the original clip (Example~\ref{ex:crop-budget}). Spatial and temporal transformations need not commute.

\begin{question}[Composed edits with a finite budget]
\label{q:composed-edits}
\qlabrun{Ask}
Can a compact statement of a stage's recovery and calibration conditions stay valid after another allowed editing stage, without enumerating every edit sequence?

\qlab{Setting}
Fix the decoder, a total cost relative to the original clip, the allowed intermediate states, and the transformed null laws.

\qlab{Success}
A probabilistic composition argument must state bounds conditional on relevant prior histories or intermediate states.

\qlab{Baseline}
The edit-distance triangle inequality tracks cost (Lemma~\ref{lem:edit-triangle}). It does not establish recovery.

\qlab{First case}
Begin with bounded frame insertions, deletions, and substitutions, then add one spatial transform.

\qlab{Deliverable}
Seek checkable sufficient conditions, or a counterexample exposing why original-input-only guarantees fail.

\qlab{Not asked}
Alignment-carrier design, the target of Question~\ref{q:synchronization}.
\end{question}

\subsection{Local contributions in hybrid production graphs}

\textbf{Evidence for one asset should retain its local meaning.}
A generated texture, a programmed trajectory, and a human annotation can all contribute to one frame. A production can also load an ancestor that stays hidden or almost entirely occluded. Watermark Anything and LoT-Pass motivate local or cross-stage questions, while leaving the meaning of a complete video contribution to be specified~\cite{wam2024,lotpass2025}.

\begin{question}[Hybrid local contribution]
\label{q:hybrid}
\qlabrun{Ask}
Which visibility and asset-separation conditions permit the acceptance tradeoff under Success, together with localized recovery?

\qlab{Setting}
Fix a retained-asset relation $R_{\mathrm{asset}}(g,v;\sigma)$ between a video $v$ and a reference asset $g$, with $\sigma$ the object that authenticates the composition masks. The verifier receives $g$ and $\sigma$.

\qlab{Success}
Require acceptance at least $\kappa\in(0,1]$ uniformly over a declared positive family, and at most $\alpha<\kappa$ over a declared relation-violating negative family.

\qlab{Baseline}
In a lossless opaque two-layer baseline, the relation holds when $v$ agrees with $g$ at source coordinates and masks authenticated in $\sigma$, for a required area and duration.

\qlab{First case}
Start with that baseline, then extend the relation to a specified bounded edit family and observation channel.

\qlab{Deliverable}
A statement of those conditions. Identical patches can correspond to several assets, so unique attribution may require abstention.

\qlab{Not asked}
Historical use: allowed reuse includes copying, and correspondence does not establish it. Production-event association is separate (Question~\ref{q:credentials}).
\end{question}

\subsection{Authentication, privacy, and deployment trust}
\label{sec:security-agenda}
\label{sec:cryptography}

\textbf{Recovering a message is only one component of an authenticated assertion.}
Issuer binding, content binding, and event binding answer different questions (Appendix~\ref{app:glossary}). Reuse of an authorized asset can be legitimate. A genuine asset credential does not automatically authenticate a new export event. Hiding a payload does not establish privacy of prompts, source code, or production history. A deployment must identify its issuance and verification keys, query access, accepted artifact relation, and allowed disclosure.

CSI (Coherence-Preserving Semantic Injection) illustrates why content binding differs from signal survival. LLM-guided semantic edits and regeneration seek altered images that a watermark detector still accepts~\cite{gao2026csi}. This tamper-acceptance objective differs from removal. Surviving evidence need not authenticate the issuer's intended content.

Language-model watermark results help separate secret-key detection, computational undetectability, and unforgeability~\cite{christ2023,lin2026signatures}. Their text-specific assumptions do not directly supply video-edit or program-reconstruction guarantees. Likewise, the quality-oracle and efficient-mixing assumptions of Watermarks in the Sand need attention before applying its limits to a particular production task~\cite{zhang2023sand}.

\begin{question}[Content-bound, privacy-aware credentials]
\label{q:credentials}
\qlabrun{Ask}
Can evidence combine honest acceptance, bounded false acceptance outside an authorized relation, and privacy relative to the received artifact?

\qlab{Setting}
For one trusted issuer, fix an efficiently checkable authorized issuance-and-derivative relation, a security parameter, and issuance-query access. The relation states which artifacts count as the issuer's outputs or their allowed derivatives. Assume the service and issuance keys stay uncompromised.

\qlab{Success}
Require three properties:
\begin{enumerate}[label=(\alph*),leftmargin=1.8em]
\item Honest acceptance at least $1-\varepsilon$, $\varepsilon<1$, for every allowed derivative, including the unedited output, using the key-averaged convention of Section~\ref{sec:recovery}. Rejection and abstention count as failures.
\item False acceptance at most $\delta\in[0,1-\varepsilon)$ outside that relation.
\item Privacy relative to the received artifact. Specify a privacy game and public leakage, including issuer, production-event type, and verification result. In a privacy game, an attacker tries to learn a private attribute from the evidence.
\end{enumerate}

\qlab{Baseline}
Compare exact-output signatures, detached credentials, and signed manifests recovered through soft bindings (Appendix~\ref{app:glossary}), with explicit repository access and trust.

\qlab{Deliverable}
Identify when an in-band watermark, embedded in the content rather than in separate metadata, improves this three-way tradeoff.

\qlab{Not asked}
An improvement claimed from message recovery alone, which does not establish it.
\end{question}

In Question~\ref{q:credentials}, verification targets association between an output and an authorized production event, under the declared relation and error guarantee. Copied evidence can preserve association to the old production event. It establishes neither a new export nor the copy's causal history.

\section{Conclusion}
\label{sec:conclusion}

AI visual media can combine image/video generation, code-based rendering, and human contributions. Their provenance needs a verification specification (Section~\ref{sec:spec-list}) linking intervention, observation, allowed changes, and the assertion under test. Source recovery need not survive rasterization, rerendering need not remove an asset mark, and contribution evidence need not establish a complete history. The resulting agenda prioritizes observable carriers, recoverable payload under task constraints, and bounded authentication.

\clearpage
\appendix
\addtocontents{toc}{\protect\setcounter{tocdepth}{1}}

\enlargethispage{1.4\baselineskip}
\paragraph{How to read the proofs.}
The appendices (Table~\ref{tab:roadmap}) state each elementary result with its hypotheses. They are standard facts or direct calculations, included so that readers can check the boundary examples; none is presented as a new theorem. The numerical examples are not empirical results.

Some proofs use numbered steps $\langle1\rangle1$, $\langle1\rangle2$, \dots{} in the style of Lamport~\cite{lamport1995write,lamport2012write}. Each step gives its justification after the marker Proof, and Let introduces a symbol. The last step, Q.E.D., proves the goal from earlier steps and names them.

\begin{table}[ht]
\centering\small
\begingroup
\renewcommand{\arraystretch}{1.12}
\setlength{\extrarowheight}{2pt}
\arrayrulecolor{black}
\begin{tabularx}{\linewidth}{>{\raggedright\arraybackslash}p{0.255\linewidth}Y>{\raggedright\arraybackslash}p{0.2\linewidth}}
\toprule
\rowcolor{TableHead}
\textbf{\sffamily Part} & \textbf{\sffamily Results} & \textbf{\sffamily Serves} \\
\midrule
Appendix~\ref{app:observation}: observation boundaries
 & Definition~\ref{def:finite-experiment}; Propositions~\ref{prop:bayes}, \ref{prop:contraction} and~\ref{prop:identical}; Lemma~\ref{lem:decoder-mass}; Remark~\ref{rem:invisible}; Corollary~\ref{cor:label-no-advantage}; Example~\ref{ex:mixture}.
 & Sections~\ref{sec:background} and~\ref{sec:agenda}; Question~\ref{q:identifiable} \\
\rowcolor{TableBand}
Appendix~\ref{app:rendering}: rendering, task preservation, payload
 & Definitions~\ref{def:render-fiber} and~\ref{def:payload}; Remarks~\ref{rem:preservation}, \ref{rem:local-linear} and~\ref{rem:worst-average}; Corollary~\ref{cor:exact-equality}; Examples~\ref{ex:monotone-bars} and~\ref{ex:hidden-anchor}; Propositions~\ref{prop:non-emulation} and~\ref{prop:zero-error-count}; Lemma~\ref{lem:simulation}.
 & Sections~\ref{sec:background} and~\ref{sec:agenda}; Questions~\ref{q:downstream} to~\ref{q:capacity} \\
Appendix~\ref{app:temporal}: temporal search and editing
 & Definitions~\ref{def:deletion-channel}, \ref{def:valid-pvalue} and~\ref{def:edit-distance}; Examples~\ref{ex:periodic}, \ref{ex:aliasing}, \ref{ex:adaptive}, \ref{ex:search-count} and~\ref{ex:crop-budget}; Proposition~\ref{prop:union}; Lemma~\ref{lem:edit-triangle}; Remarks~\ref{rem:key-average} and~\ref{rem:permutation}.
 & Sections~\ref{sec:background} and~\ref{sec:temporal}; Questions~\ref{q:synchronization} and~\ref{q:composed-edits} \\
\bottomrule
\end{tabularx}
\endgroup
\caption{Roadmap of the appendices.}
\label{tab:roadmap}
\end{table}

\section{Elementary Observation Boundaries}\label{app:observation}

\subsection*{The finite classification experiment}

\begin{definition}[Finite classification experiment]\label{def:finite-experiment}
Fix a finite observation alphabet $\Omega$, for example videos with bounded length, a fixed resolution, and quantized pixels. Let $f_0$ and $f_1$ be probability mass functions on $\Omega$ for two fully specified production experiments. Give the labels $0$ and $1$ equal prior probability, and draw the observation $o$ from $f_0$ or $f_1$ according to the label. A randomized classifier sees $o$ and chooses label $0$ with probability $d(o)\in[0,1]$, and label $1$ otherwise. Its private randomness has no further information about the label. Its error is
\begin{equation}
 \operatorname{err}(d)=\frac12\sum_{o\in\Omega}
 \bigl[f_0(o)(1-d(o))+f_1(o)d(o)\bigr].
 \label{eq:information-risk}
\end{equation}
The optimal equal-prior error is the least $\operatorname{err}(d)$ over all $d:\Omega\to[0,1]$. The total variation distance is $\TV(f_0,f_1)=\frac12\sum_{o\in\Omega}|f_0(o)-f_1(o)|$.
\end{definition}

\begin{proposition}[Equal-prior Bayes benchmark]\label{prop:bayes}
In the setting of Definition~\ref{def:finite-experiment}, the minimum of $\operatorname{err}(d)$ over all functions $d:\Omega\to[0,1]$ is
\begin{equation}
 \min_d \operatorname{err}(d)=\frac{1-\TV(f_0,f_1)}2.
 \label{eq:information-bayes}
\end{equation}
A classifier attains the minimum if $d(o)=1$ when $f_0(o)>f_1(o)$ and $d(o)=0$ when $f_0(o)<f_1(o)$, with any value at a tie. Moreover $0\leq\TV(f_0,f_1)\leq1$. The minimum is $1/2$ if $f_0=f_1$, and $0$ if $f_0$ and $f_1$ have disjoint supports.
\end{proposition}

\begin{proof}
For each $o$, twice the summand of \eqref{eq:information-risk} is $f_0(o)+(f_1(o)-f_0(o))d(o)$. This is affine in $d(o)\in[0,1]$, so it is at least $\min\{f_0(o),f_1(o)\}$. Equality holds at $d(o)=0$ if $f_1(o)>f_0(o)$, at $d(o)=1$ if $f_1(o)<f_0(o)$, and anywhere at a tie. The values $d(o)$ are free, so the classifier of the statement attains $\min_d \operatorname{err}(d)=\frac12\sum_o\min\{f_0(o),f_1(o)\}$.

For $a,b\geq0$, $\min\{a,b\}=(a+b-|a-b|)/2$. Sum this identity and use $\sum_o f_0(o)=\sum_o f_1(o)=1$ to get $\sum_o\min\{f_0(o),f_1(o)\}=1-\TV(f_0,f_1)$. Because $|a-b|\leq a+b$, $\sum_o|f_0(o)-f_1(o)|\leq2$, so $0\leq\TV\leq1$. If $f_0=f_1$, then $\TV=0$. If the supports are disjoint, then $|f_0(o)-f_1(o)|=f_0(o)+f_1(o)$ for every $o$, so $\TV=1$.
\end{proof}

The formula does not estimate these distributions or guarantee that a learned detector finds the optimal rule. Different mechanisms alone supply no positive lower bound on their separation. Content, compression, and renderer choices can separate a particular dataset while saying little about authorship.

Common downstream processing gives another boundary. A resize or quantizer can be a common channel once the specification declares its randomness and inputs.

\begin{proposition}[Contraction under a common channel]\label{prop:contraction}
Let $\Omega$ and $\Gamma$ be finite alphabets. Let $K(y\mid o)\geq0$ for $o\in\Omega$ and $y\in\Gamma$, with $\sum_{y}K(y\mid o)=1$ for every $o$. Let $f_0$ and $f_1$ be probability mass functions on $\Omega$. Apply the same $K$ to both, and define $f'_i(y)=\sum_{o}K(y\mid o)f_i(o)$ for $i\in\{0,1\}$. Then:
\begin{enumerate}[label=(\alph*)]
\item $\TV(f'_0,f'_1)\leq\TV(f_0,f_1)$;
\item the optimal equal-prior error of $(f'_0,f'_1)$ on $\Gamma$ is at least the optimal equal-prior error of $(f_0,f_1)$ on $\Omega$.
\end{enumerate}
\end{proposition}

\begin{proof}[Proof of Proposition~\ref{prop:contraction}]
\begin{steps}
\step{1}{1} Part (a) holds.
\pf \Let $\Delta(o)=f_0(o)-f_1(o)$, so $f'_0(y)-f'_1(y)=\sum_oK(y\mid o)\Delta(o)$. Then $2\TV(f'_0,f'_1)=\sum_y\bigl|\sum_oK(y\mid o)\Delta(o)\bigr|\leq\sum_y\sum_oK(y\mid o)\,|\Delta(o)|$, by the triangle inequality and $K\geq0$. Interchange the finite sums. Because $\sum_yK(y\mid o)=1$, the right side is $\sum_o|\Delta(o)|=2\TV(f_0,f_1)$.
\step{1}{2} Part (b) holds.
\pf Because $K\geq0$, $\sum_yK(y\mid o)=1$ and $\sum_o f_i(o)=1$, $f'_0$ and $f'_1$ are probability mass functions on $\Gamma$. By \eqref{eq:information-bayes} on $\Omega$ and on $\Gamma$, the optimal errors are $(1-\TV(f_0,f_1))/2$ and $(1-\TV(f'_0,f'_1))/2$. The map $t\mapsto(1-t)/2$ is decreasing, so \sref{1}{1} gives the claim.
\qedstep{1}{3}
\pf By \sref{1}{1} and \sref{1}{2}.
\end{steps}
\end{proof}

The shared-channel assumption matters. Label-dependent processing could introduce a label signal of its own. Part (b) concerns optimal classification, not the accuracy of a particular implementation. A practical classifier can benefit from preprocessing even when the observation contains no additional information. This is a reason to separate information limits from engineering comparisons.

\subsection*{A source mark must pass through execution}

\begin{lemma}[Output probabilities of a decoder]\label{lem:decoder-mass}
Let $\mathcal M$ be a finite message set. Let $D$ be a randomized decoder with values in $\mathcal M\cup\{\bot\}$, where $\bot$ means abstention. For every fixed input $x$ of $D$, $\sum_{m\in\mathcal M}\Prb\{D(x)=m\}\leq1$. Equality holds if and only if $\Prb\{D(x)=\bot\}=0$.
\end{lemma}

\begin{proof}
For fixed $x$, the outputs in $\mathcal M$ and $\bot$ are exclusive events of the private randomness of $D$. Their probabilities sum to one.
\end{proof}

\begin{proposition}[Identical observation]\label{prop:identical}
Let $\mathcal M$ be a finite message set with $N\geq1$ elements. Let the random message $M$ have any prior $\rho$ on $\mathcal M$. Let $X$, with values in a finite set, collect the pixels, the key, and every other input of the verifier. Let $D$ be a randomized decoder as in Lemma~\ref{lem:decoder-mass} that maps each value of $X$ to an element of $\mathcal M\cup\{\bot\}$. Assume:
\begin{enumerate}[label=(\roman*)]
\item $X$ is independent of $M$;
\item the private randomness of $D$ is independent of $(M,X)$.
\end{enumerate}
Let $\theta_m=\Prb\{D(X)=m\}$. Then $\sum_{m\in\mathcal M}\theta_m\leq1$, with equality if $\Prb\{D(X)=\bot\}=0$, and
\begin{equation}
 \Prb\{D(X)=M\}=\sum_{m\in\mathcal M}\rho(m)\,\theta_m\leq\max_{m\in\mathcal M}\rho(m).
 \label{eq:information-invisible}
\end{equation}
A constant decoder that outputs a most probable message attains the bound, which is $1/N$ for the uniform prior.
\end{proposition}

\begin{proof}[Proof of Proposition~\ref{prop:identical}]
\begin{steps}
\step{1}{1} $\sum_m\theta_m\leq1$, with equality if $\Prb\{D(X)=\bot\}=0$.
\pf By (ii), $\Prb\{D(X)=m\mid X=x\}=\Prb\{D(x)=m\}$ for every value $x$ of $X$ with $\Prb\{X=x\}>0$, so $\theta_m=\sum_x\Prb\{X=x\}\Prb\{D(x)=m\}$. Interchange the finite sums, apply Lemma~\ref{lem:decoder-mass} to each $x$, and use $\sum_x\Prb\{X=x\}=1$. If $\Prb\{D(X)=\bot\}=0$, then $\Prb\{D(x)=\bot\}=0$ for every $x$ with $\Prb\{X=x\}>0$, and the lemma gives equality.
\step{1}{2} $\Prb\{D(X)=M\}=\sum_m\rho(m)\theta_m\leq\max_m\rho(m)$.
\pf By (i) and (ii), $M$, $X$ and the private randomness of $D$ are mutually independent. The output $D(X)$ is a function of $X$ and that private randomness, so $\Prb\{M=m,\,D(X)=m\}=\rho(m)\,\theta_m$. These events are disjoint for distinct $m$, with union $\{D(X)=M\}$. Because $\theta_m\geq0$, the sum over $m$ is at most $(\max_m\rho(m))\sum_m\theta_m$, and \sref{1}{1} finishes.
\step{1}{3} A constant decoder attains the bound.
\pf Let $m^\ast$ satisfy $\rho(m^\ast)=\max_m\rho(m)$, and let $D$ always output $m^\ast$. This $D$ satisfies (ii), and $\theta_{m^\ast}=1$ and $\theta_m=0$ otherwise. So $\Prb\{D(X)=M\}=\rho(m^\ast)$.
\qedstep{1}{4}
\pf \sref{1}{1}, \sref{1}{2} and \sref{1}{3} give the claims.
\end{steps}
\end{proof}

For a source mark, fix a program $P$, an environment $z$, and a key $k$. Let the renderer $R$ be deterministic, let $W_k(P,m)$ encode the message $m\in\mathcal M$ of Proposition~\ref{prop:identical}, and suppose that $R(W_k(P,m),z)=o_0$ for every $m$. Then $X=(o_0,k)$ is constant, so independent of $M$. Let a decoder receive only these pixels and the key, with private randomness independent of $M$. By Proposition~\ref{prop:identical}, it succeeds with probability at most $\max_m\rho(m)$. For a uniform message this is $1/N$.

A stochastic downstream channel gives the same conclusion if its finite-valued output $Y$ has the same law given every message of positive prior probability, for the fixed key and other fixed verifier inputs. Then $X$, which collects $Y$ and the fixed inputs, is independent of $M$, and Proposition~\ref{prop:identical} applies to a decoder whose private randomness is independent of $(M,X)$. If the key or side information reveals the selected message, then $X$ depends on $M$, and the proposition does not apply.

\begin{remark}[Which edits can be invisible]\label{rem:invisible}
Whether a comment, an identifier change, or a latent perturbation leaves the observation unchanged depends on the specified renderer or decoder. A comment illustrates the identical-pixel premise (Corollary~\ref{cor:exact-equality}) only if the renderer neither executes nor inspects it. An identifier change can affect a reflective program, and an executed numerical literal can change visible motion. For a deterministic renderer, the relevant object is a render fiber (Definition~\ref{def:render-fiber}): the set of source descriptions that give the same observation. It is not a syntactic category called ``code watermark.'' A neural latent intervention can likewise disappear under its decoder, depending on the specified map.
\end{remark}

The nonconstant case needs a second specification: which changes the animation may tolerate (Remark~\ref{rem:preservation}). Exact equality of output pixels is stronger and can eliminate a proposed carrier. Neither the encoder's access to the source code nor its knowledge of a desired trajectory grants the verifier that same knowledge. A blind detection problem should explicitly withhold clean frames and scene coordinates. Otherwise the decoder may recover differences that privileged side information supplies, not information that the observation carries.

\subsection*{Conditional equality is not marginal equality}

\begin{corollary}[Identical label-conditional laws]\label{cor:label-no-advantage}
Let the random label $I$ have a finite range and law $\rho$. Let $X$ be a random variable with values in a finite set, and assume that $X$ has the same conditional law given every label of positive prior probability. Let $D$ be a randomized decoder of the label, as in Lemma~\ref{lem:decoder-mass}, with values in the labels and $\bot$, whose private randomness is independent of $(I,X)$. Then $\Prb\{D(X)=I\}\leq\max_i\rho(i)$, the success probability of the best constant guess.
\end{corollary}

\begin{proof}
Let $c(x)$ be the common value of $\Prb\{X=x\mid I=i\}$ over the labels $i$ with $\rho(i)>0$. Summing $\Prb\{I=i,X=x\}=\rho(i)\,c(x)$ over these $i$ gives $\Prb\{X=x\}=c(x)$. Thus $\Prb\{I=i,X=x\}=\rho(i)\,\Prb\{X=x\}$ for every $i$, also if $\rho(i)=0$. So $X$ is independent of $I$, and Proposition~\ref{prop:identical} with $M=I$ applies.
\end{proof}

\begin{example}[Two author classes, one renderer]\label{ex:mixture}
Let $\Pi$ be a finite set of pairs $(P,z)$ of a program and an environment. Let $\Omega$ be a finite set of videos, the observation alphabet of Definition~\ref{def:finite-experiment}. Let $K(v\mid P,z)$ be the conditional output law of the renderer on $\Omega$. It does not depend on who supplied $P$, so a human and an LLM who supply the same $(P,z)$ give equal conditional video laws. Let $\mu_{\mathrm{hum}}$ and $\mu_{\mathrm{LLM}}$ be the author laws of a human and of an LLM on $\Pi$; they can differ. For an author law $\mu$, the marginal video law is $\bar\mu(v)=\sum_{(P,z)\in\Pi}K(v\mid P,z)\,\mu(P,z)$.

Fix $z$ and two programs $P_b$ and $P_w$, and let each author law select programs paired with this $z$. The renderer outputs an all-black clip with probability one for $P_b$, and an all-white clip with probability one for $P_w$. Let $\mu_{\mathrm{hum}}$ select $P_b$ with probability $3/4$, and $\mu_{\mathrm{LLM}}$ with probability $1/4$; each class selects $P_w$ otherwise. The black clip then has probabilities $3/4$ and $1/4$, and the white clip $1/4$ and $3/4$. Each pair differs by $1/2$, so $\TV(\bar\mu_{\mathrm{hum}},\bar\mu_{\mathrm{LLM}})=(1/2+1/2)/2=1/2$. By Proposition~\ref{prop:bayes} with $f_0=\bar\mu_{\mathrm{hum}}$ and $f_1=\bar\mu_{\mathrm{LLM}}$ on $\Omega$, the equal-prior optimal authorship error is $(1-1/2)/2=1/4$.
\end{example}

Equal conditional kernels, without information about the author laws, determine neither equality nor inequality of the marginal laws. Different author laws can remain distinguishable, as in Example~\ref{ex:mixture}. They can also collapse under the same kernel, for instance one that renders $P_b$ and $P_w$ to the same clip. A first analysis should keep author sampling separate from rendering. It should also distinguish recovery of a deliberately inserted mark from passive inference about the author's identity. This artificial sampling example establishes no real authorship detector; it shows exactly why a conditional equality cannot justify a marginal impossibility claim.
\Needspace{9\baselineskip}
\section{Rendering, Task Preservation, and Payload Bounds}\label{app:rendering}

\subsection*{Render fibers and three preservation requirements}

\begin{definition}[Render fiber]\label{def:render-fiber}
Let $R$ be a deterministic renderer that maps a program $P$ and an environment $z$ to an output $R(P,z)$ in a finite set $\Omega$ of declared frame representations. Equality in $\Omega$ includes quantization.
\begin{enumerate}[label=(\alph*),leftmargin=2em]
\item The \emph{render fiber} of $P$ at $z$ is
\[
 [P]_{R,z}=\{P':R(P',z)=R(P,z)\}.
\]
\item For a nonempty set $\mathcal Z$ of allowed environments, the \emph{uniform render fiber} is $[P]_{R,\mathcal Z}=\bigcap_{z\in\mathcal Z}[P]_{R,z}$.
\item For a neural state $\Lambda$, a deterministic decoder $F$, and an environment $z$, the analogous fiber is $[\Lambda]_{F,z}=\{\Lambda':F(\Lambda',z)=F(\Lambda,z)\}$.
\end{enumerate}
\end{definition}

Equality at every allowed environment is stronger than membership in the fiber at one environment.\footnote{A fiber of a map is the set of inputs that produce one given output.} Two programs can agree at one environment and differ at another. A finite suite of functional tests does not establish either form of render equivalence unless the tested environments exhaust the allowed set or a proof shows the suite complete for the declared semantics.

\begin{remark}[Three preservation requirements]\label{rem:preservation}
A verification specification should state three requirements separately.
\begin{enumerate}[label=(\alph*),leftmargin=2em]
\item \emph{Program semantic equivalence} under a declared execution model, which concerns declared outputs and interactions.
\item \emph{Exact render equivalence}: membership in a render fiber or a uniform render fiber (Definition~\ref{def:render-fiber}).
\item \emph{Task preservation}: the declared predicate $U_x(s_0,s')=1$ for a baseline state $s_0$ and a candidate state $s'$.
\end{enumerate}
If rendered frames are the outputs, a movement that changes them changes the semantics, yet a video task might allow that motion to vary. Conversely, a program can produce identical frames and differ in logs or network requests, so a frame-only verifier sees an equivalence that a verifier observing execution does not.

Task preservation need not define a transitive equivalence relation. For example, let the states be bar heights, and let $U_x(s_0,s')=1$ exactly when $|s'-s_0|\leq 1/2$. Then $U_x(0,1/2)=U_x(1/2,1)=1$, but $U_x(0,1)=0$.
\end{remark}

The predicate $U_x(s_0,s')$ states the freedom that the task allows. An explanatory chart might require exact data values, readable labels, and event order, and tolerate small changes in motion. A simulation might require conservation of specified quantities and forbid any coordinate perturbation. A cinematic scene could allow a much broader appearance range. A generic perceptual similarity score alone cannot reveal these choices, yet they determine the visible freedom for encoding information.

For example, changing a circle's path $c(t)$ to $c(t)+u_k(t)$, with a small displacement $u_k(t)$ indexed by $k$, can preserve object identity and narrated event order while it changes every affected frame. It is task-preserving only under a task that allows it, and it is not exact preservation of the original video. The same holds for a neural latent perturbation judged visually acceptable. Quality tolerance creates a different constraint from exact decoded equality.

\subsection*{Applying the exact-equality obstruction}

The simplest obstruction needs no hypothesis about the architecture of a watermark algorithm.

\begin{corollary}[Exact-equality obstruction]\label{cor:exact-equality}
Fix a program $P$, an environment $z$, a key $k$, and a deterministic renderer $R$. Let $M$ be a random message with any prior in a finite message set $\mathcal M$ of size $N\geq 1$. Let an encoder satisfy $W_k(P,m)\in[P]_{R,z}$ for every $m\in\mathcal M$, so that
\[
 R(W_k(P,m),z)=R(P,z)=o_0
 \qquad\text{for every }m\in\mathcal M.
\]
Let a pixel-only decoder $D_k$ receive $o_0$ and $k$, use private randomness independent of $M$, and possibly abstain. Then
\[
 \Prb\{D_k(o_0)=M\}\leq\max_{m\in\mathcal M}\Prb\{M=m\},
\]
and this bound equals $1/N$ for the uniform prior. In the neural case, replace $P$, $R$, and $W_k$ by a neural state $\Lambda$, a deterministic decoder $F$, and an encoder $E_k$ with $E_k(\Lambda,m)\in[\Lambda]_{F,z}$ for every $m\in\mathcal M$. Then $o_0=F(\Lambda,z)$, and the same bound holds for a pixel-only decoder with the same properties.
\end{corollary}

\begin{proof}
The decoder input $X=(o_0,k)$ takes the same value for every message, so $X$ is independent of $M$. Private randomness independent of $M$ is then independent of $(M,X)$. Proposition~\ref{prop:identical} applied to $X$ gives inequality~\eqref{eq:information-invisible}, which is the bound. In the neural case $F(E_k(\Lambda,m),z)=F(\Lambda,z)$ for every $m$, so the same argument applies.
\end{proof}

The corollary does not apply when:
\begin{itemize}[leftmargin=1.6em]
\item the environment $z$ carries a message;
\item the decoder receives message-correlated side information;
\item the renderer emits additional marked pixels, so that $R(W_k(P,m),z)$ depends on $m$;
\item the verifier observes container metadata or execution traces, which are different observations;
\item the renderer is random; the argument then needs equal conditional observation laws for all messages, not equality at one seed.
\end{itemize}
The corollary does not say that a watermark cannot identify source code: marked sources can be different elements of one render fiber. In the neural case the obstruction likewise lies in the observation map, without a claim that learned decoding and program execution have identical geometry, computational costs, or available interventions.

\begin{example}[Monotone bars]\label{ex:monotone-bars}
Let $d>0$ be the duration of the clip, and let bar $i$ in a finite family have specified final height $b_i>0$. A task allows any monotone (nondecreasing) growth of bar $i$ from height $0$ to $b_i$ during $[0,d]$, with unchanged labels. The baseline trajectory is $h_i(t)=b_it/d$. For $0<\nu\leq(2\pi)^{-1}$, let
\[
 \widetilde h_i(t)=b_i\bigl[t/d+\nu\sin(2\pi t/d)\bigr],
 \qquad 0\leq t\leq d.
\]
Then (a)--(c) hold.
\begin{enumerate}[label=(\alph*),leftmargin=2em]
\item $\widetilde h_i(0)=0$ and $\widetilde h_i(d)=b_i$.
\item $\widetilde h_i'(t)=(b_i/d)\bigl(1+2\pi\nu\cos(2\pi t/d)\bigr)\geq0$, so $\widetilde h_i$ is nondecreasing with values in $[0,b_i]$ and satisfies the task. Also $\widetilde h_i(t)-h_i(t)=b_i\nu\sin(2\pi t/d)$ is nonzero except at $t\in\{0,d/2,d\}$, so the alternative changes intermediate heights.
\item Let $R$ be a renderer that draws the bars deterministically from the heights at $T$ sample times, using a fixed deterministic quantizer with $q$ levels whose cells are intervals, with all other state and the environment $z$ fixed. Suppose that the quantized sampled heights of $h_i$ and $\widetilde h_i$ agree at every sample time for every bar. Then the two videos are identical, so the programs that implement the baseline and alternative trajectories lie in one render fiber. Corollary~\ref{cor:exact-equality} then applies to an encoder that selects between them by message. The hypothesis holds if every sampled baseline height of bar $i$ lies farther than $b_i\nu$ from each quantization boundary, because $|\widetilde h_i(t)-h_i(t)|\leq b_i\nu$. Equivalently, the video changes only if some sampled height crosses a quantization boundary.
\end{enumerate}
\end{example}

\begin{proof}[Verification of Example~\ref{ex:monotone-bars}]
\begin{steps}
\step{1}{1} Parts (a) and (b) hold.
\pf Because $\sin0=\sin2\pi=0$, (a) holds. In (b), $\cos\geq-1$ and $\nu\leq(2\pi)^{-1}$ make the derivative nonnegative, so $\widetilde h_i$ is nondecreasing and by (a) has values in $[0,b_i]$. The sine vanishes on $[0,d]$ exactly at $0$, $d/2$, $d$.
\step{1}{2} Part (c) holds.
\pf The video is a deterministic function of the quantized sampled heights and the other fixed state, so equal quantized sampled heights give equal videos. Thus the two programs have the same output at $z$ and lie in one render fiber, as Corollary~\ref{cor:exact-equality} requires. The last sentence of (c) is the contrapositive. If, for every bar $i$, every sampled baseline height lies farther than $b_i\nu$ from each quantization boundary, then since $|\widetilde h_i(t)-h_i(t)|=b_i\nu|\sin(2\pi t/d)|\leq b_i\nu$, both heights lie in the same quantization cell, so the hypothesis of (c) holds.
\qedstep{1}{3}
\pf By \sref{1}{1} and \sref{1}{2}.
\end{steps}
\end{proof}

A program can implement the alternative trajectory directly, but whether a specified neural encoder can reach it is a separate controllability question, not a consequence of the task predicate. Conflating acceptable continuous trajectories, reachable production states, and distinct quantized observations would overstate the payload.

\subsection*{Comparing intervention stages without hiding access}

\begin{proposition}[Non-emulation]\label{prop:non-emulation}
Let $\mathcal S$ be a set of production states, let $V(s)\in\Omega$ be the ordinary output video of $s\in\mathcal S$, and let $J_k(s,m)\in\Omega$ be the desired marked output for a key $k$ and a message $m$. A video-only encoder is a deterministic function $f(v,k,m)$ of a video, the key, and the message. Suppose that $V(s)=V(s')$ and $J_k(s,m)\neq J_k(s',m)$ for some $k$, $m$, and $s,s'\in\mathcal S$. Then, for these $k$ and $m$, no video-only encoder satisfies $f(V(s''),k,m)=J_k(s'',m)$ for all $s''\in\mathcal S$.
\end{proposition}

\begin{proof}
Because $V(s)=V(s')$, the encoder returns one value for $s$ and $s'$. That value cannot equal both $J_k(s,m)$ and $J_k(s',m)$.
\end{proof}

An upstream intervention can depend on information that the final video lacks. The proposition states this as a direct equality conflict, not as a claim about any particular neural method.

\begin{example}[Hidden anchor]\label{ex:hidden-anchor}
Let $s_0$ and $s_1$ be scene descriptions of one solid square with disjoint nonempty left and right halves. Each $s_i$ also stores a hidden anchor $i\in\{0,1\}$ that the baseline renderer ignores, so $V(s_0)=V(s_1)$. The upstream rule draws a nonempty mark in a color different from that of the square. The mark may depend on the key and the message, and it lies inside the left half for anchor $0$ and inside the right half for anchor $1$. So $J_k(s_0,m)$ shows the mark color in the left half, where $J_k(s_1,m)$ shows the color of the square. Thus $J_k(s_0,m)\neq J_k(s_1,m)$ for every $k$ and $m$. By Proposition~\ref{prop:non-emulation}, no video-only deterministic encoder reproduces this rule.
\end{example}

Likewise, neural states can differ while they produce the same baseline observation, so a lost distinction can prevent a downstream encoder from recreating a state-dependent intervention. This does not establish a payload or robustness advantage: another downstream rule might do equally well, and an advantage requires comparing optimized feasible schemes.

A narrower comparison is valid. Suppose a downstream encoder
\begin{itemize}[leftmargin=1.6em]
\item has the same production state and key,
\item can compute the intended marked output, and
\item may apply a transformation that carries the baseline output to the intended marked output.
\end{itemize}
Then it can match that output. These assumptions are substantive: a rule may prohibit output replacement, computational access may differ, or the allowed distortion budget may depend on an inaccessible reference. A common pixel-stage baseline is valuable without assuming that it can emulate every earlier intervention.

\subsection*{Source and raster verifiers}

\begin{lemma}[Simulating a known channel]\label{lem:simulation}
Let $M$ be a random message with values in a finite set $\mathcal M$. Let $X$ be the received marked source together with all other inputs of a source verifier, where $X$ takes finitely many values. Let $O=\Phi(X,\omega)$ be the raster observation (all that a raster decoder receives), where $\Phi$ is a deterministic map and $\omega$ is channel noise. Assume the following.
\begin{enumerate}[label=(H\arabic*),leftmargin=2.6em]
\item The verifier knows $\Phi$ and can evaluate it, and $X$ contains every input of $\Phi$ except $\omega$.
\item The noise $\omega$ is independent of $(M,X)$.
\item The verifier knows the law of $\omega$ and can sample from it.
\end{enumerate}
Let $D'$ be a randomized raster decoder with values in $\mathcal M\cup\{\bot\}$ and with private randomness independent of $(M,X,\omega)$. Let $D$ be the source decoder. On a value $x$ of $X$, it draws a fresh sample $\omega'$ from the law of $\omega$, applies $D'$ to $\Phi(x,\omega')$ with fresh private randomness, and outputs the result. Then, for every $(m,x)$ with $\Prb\{M=m,X=x\}>0$, the laws of $D(X)$ and $D'(O)$ given $(M,X)=(m,x)$ are equal. So $D$ and $D'$ have the same conditional success and abstention probabilities for every such message $m$ and received source $x$.
\end{lemma}

\begin{proof}
Fix $(m,x)$ with $\Prb\{M=m,X=x\}>0$. By (H2), given $(M,X)=(m,x)$ the observation is $\Phi(x,\omega)$, where $\omega$ has its unconditional law. By (H1) and (H3), $D$ can evaluate $\Phi(x,\omega')$ for a fresh sample $\omega'$ with that law. Both outputs apply $D'$ to such a value, with private randomness independent of the noise sample. So the output laws agree.
\end{proof}

Thus a source verifier that satisfies (H1)--(H3) matches the conditional performance of any raster decoder, the baseline of Question~\ref{q:vector-raster}.

\subsection*{Visible freedom and nuisance variation}

\begin{remark}[Local linear model]\label{rem:local-linear}
Let $A$ and $B$ be real matrices of compatible sizes. Let $\xi$, an intermediate adjustment, and $\vartheta$, a nuisance state, be real vectors. At a fixed operating point, approximate a chosen video feature by
\[
 y=A\xi+B\vartheta.
\]
A watermark adjustment $w$ then has first-order effect $Aw$. In this model:
\begin{enumerate}[label=(\alph*),leftmargin=2em]
\item If $Aw=0$, that is, $w\in\ker A$, then $y(\xi+w,\vartheta)=y(\xi,\vartheta)$.
\item If $Aw=B\zeta$ and the nuisance model allows replacing $\vartheta$ by $\vartheta-\zeta$, then $y(\xi+w,\vartheta-\zeta)=y(\xi,\vartheta)$. Without that permission, the cancellation is not a valid attack.
\item Directions outside the set of adjustments that the task allows are not available, even if they are observable.
\end{enumerate}
Items (a) and (b) are exact identities for the linear model, because $A(\xi+w)+B(\vartheta-\zeta)=A\xi+B\vartheta+Aw-B\zeta$. They say nothing about a nonlinear renderer or decoder.
\end{remark}

Conditions (a) and (c) differ: visibility alone does not make an intervention acceptable. Bounded camera variation, fixed inputs, or trusted synchronization can remove otherwise available compensations, so the model requires a declared range for the nuisance state, not every mathematical direction as an attacker capability.

Preserving narrated event order may permit timing variation but forbid changing the measured trajectory. A different lesson may require timing itself to be exact. A formulation must decide which variations remain legitimate, such as timing in an instructional animation, before it maximizes recoverable bits. The task predicate defines that freedom, the renderer or decoder maps it to observations, and allowed nuisance transformations determine which distinctions persist.

For code-based rendering, $\xi$ might describe geometry, color, or timing parameters. For an image/video model, it might describe a local latent adjustment, with $A$ a linearization of the decoder. The local picture can fail under visibility changes, nonlinear decoding, quantization, or finite perturbations. The model clarifies the competing constraints as a conceptual guide for locating candidate carriers. It is not a security theorem or an empirical approximation established here.

\subsection*{Payload bounds need a finite observation model}

Unbounded precision can give artificial payload infinities, so the next definition fixes a finite model.

\begin{definition}[Finite payload model]\label{def:payload}
Fix positive integers $T$, $H$, $W$, $C$ (frames, height, width, channels) and an integer $q\geq 2$ (quantization levels per channel). Put $n=THWC$ and $\Omega=\{0,\dots,q-1\}^n$, so that $|\Omega|=q^n$. Fix a distance on $\Omega$, and for $v\in\Omega$ let $\operatorname{Ball}(v)\ni v$ be the set of outputs within the declared output-distortion budget of $v$. Let $\mathcal S$ be the nonempty set of allowed production states, $V(s)\in\Omega$ the ordinary output of $s$, and $U_x$ the task predicate. Fix a key $k$ and a finite nonempty message set $\mathcal M$.
\begin{enumerate}[label=(\alph*),leftmargin=2em]
\item The set of \emph{task-allowed outputs} of $s\in\mathcal S$ is
\[
 \mathcal V_U(s)=\{V(s'):s'\in\mathcal S,\ U_x(s,s')=1,\ V(s')\in\operatorname{Ball}(V(s))\}.
\]
A route may restrict this set further to the outputs that it can reach. The results below hold for such a smaller declared set, again written $\mathcal V_U(s)$.
\item A \emph{deterministic encoder} is a map $E$ with $E(s,m)\in\mathcal V_U(s)$ for all $s\in\mathcal S$ and $m\in\mathcal M$.
\item Let $\mathcal T$ be a family of deterministic maps $\tau:\Omega\to\Omega$ (output transformations) containing the identity. The \emph{reachable observation set} of $(s,m)$ is $\operatorname{Reach}(s,m)=\{\tau(E(s,m)):\tau\in\mathcal T\}$.
\item A decoder $D_k$ is a possibly randomized map from $\Omega$ to $\mathcal M\cup\{\bot\}$. It receives the observation and the key, not $s$. It is \emph{zero-error} for $E$ if $\Prb\{D_k(\tau(E(s,m)))=m\}=1$ for every $s\in\mathcal S$, every $m\in\mathcal M$, and every $\tau\in\mathcal T$.
\end{enumerate}
\end{definition}

Audio, metadata, variable duration, or additional precision need separate finite budgets if the formulation allows them. Unbounded source length or real-valued parameters do not create infinitely many distinct observations in the fixed alphabet $\Omega$. The definition projects a state-based task predicate onto pixels.

\begin{proposition}[Zero-error payload bounds]\label{prop:zero-error-count}
In the setting of Definition~\ref{def:payload}, let $D_k$ be a zero-error decoder for a deterministic encoder $E$. Then:
\begin{enumerate}[label=(\roman*),leftmargin=2em]
\item $\operatorname{Reach}(s,m)\cap\operatorname{Reach}(s',m')=\emptyset$ for all $s,s'\in\mathcal S$ and all distinct $m,m'\in\mathcal M$; equivalently, the unions $\bigcup_{s\in\mathcal S}\operatorname{Reach}(s,m)$ over states are disjoint for distinct messages $m$;
\item $|\mathcal M|\leq|\mathcal V_U(s)|\leq q^n$ for every $s\in\mathcal S$, so the payload $\log_2|\mathcal M|$ is at most $n\log_2q$ bits.
\end{enumerate}
\end{proposition}

\begin{proof}[Proof of Proposition~\ref{prop:zero-error-count}]
\begin{steps}
\step{1}{1} Part (i) holds.
\pf Suppose $o\in\operatorname{Reach}(s,m)\cap\operatorname{Reach}(s',m')$ with $m\neq m'$. Then $o=\tau(E(s,m))=\tau'(E(s',m'))$ for some $\tau,\tau'\in\mathcal T$, so zero-error gives $\Prb\{D_k(o)=m\}=1=\Prb\{D_k(o)=m'\}$. These probabilities sum to two, which contradicts Lemma~\ref{lem:decoder-mass}.
\step{1}{2} Part (ii) holds.
\pf Fix $s\in\mathcal S$. The identity lies in $\mathcal T$, so $E(s,m)\in\operatorname{Reach}(s,m)$. By \sref{1}{1} with $s'=s$, distinct messages have disjoint reachable sets, so $m\mapsto E(s,m)$ is injective into $\mathcal V_U(s)\subseteq\Omega$ (Definition~\ref{def:payload}(b)). So $|\mathcal M|\leq|\mathcal V_U(s)|\leq|\Omega|=q^n$, and $\log_2|\mathcal M|\leq n\log_2q$.
\qedstep{1}{3}
\pf By \sref{1}{1} and \sref{1}{2}.
\end{steps}
\end{proof}

Part (i) with $s\neq s'$ matters because the decoder does not receive $s$ and must separate the messages across all states. Part (ii) holds even if $\mathcal T$ contains only the identity, so it applies before any attack. These counting constraints do not determine the largest payload achievable under the declared attacks.

\begin{remark}[Worst case and average case]\label{rem:worst-average}
A worst-case formulation can require success for every allowed production state, message, and transformation. An average formulation must name distributions over production states, messages, and attacks. Let $\Psi$ be the allowed set of such cases, let $\operatorname{succ}(\psi)\in[0,1]$ be the success probability at a case $\psi\in\Psi$, and let $\varepsilon\in[0,1]$. If $\Psi$ is infinite, assume that $\operatorname{succ}$ is measurable. Both formulations use the same production, encoder, attacker, and decoder randomness.
\begin{enumerate}[label=(\alph*),leftmargin=2em]
\item If $\operatorname{succ}(\psi)\geq1-\varepsilon$ for every $\psi\in\Psi$, then $\int\operatorname{succ}\,d\mu\geq1-\varepsilon$ for every probability law $\mu$ on $\Psi$, because every value of $\operatorname{succ}$ is at least $1-\varepsilon$.
\item For a fixed law $\mu$, the converse can fail: if $\Psi=\{\psi_1,\psi_2\}$, $\operatorname{succ}=(1,0)$, and $\mu=(9/10,1/10)$, the average success is $9/10$ but the worst-case success is $0$. Here $\psi_2$ can be a rare task with no visible freedom, which contributes little to the average error.
\end{enumerate}
\end{remark}

If the formulation allows randomized encoding or verification, it must specify their probability spaces and error tolerances. The zero-error argument does not carry over unchanged.

False alarms and false acceptances need a separate model. A codebook that separates its own messages can still collide with ordinary unmarked videos (a false alarm under the null law), and a verifier can accept transplanted marked assets (a false acceptance outside the authorized relation). So recoverable payload and reliable association with a production event are different objectives. The finite model is a starting discipline, not a complete watermark specification.
\Needspace{9\baselineskip}
\section{Temporal Search and Editing Boundaries}\label{app:temporal}

\subsection*{The clock belongs to a representation}

Three clocks can coexist. A latent video model may use a latent-frame index, an animation program uses a declared timeline, and an encoded file has presentation timestamps. They need not advance together.

A renderer may sample an analytic trajectory at a new frame rate, while a video editor may interpolate already sampled frames. These operations can have different effects on the same intended motion. Specifying the verifier's observation includes declaring the carrier's clock.

\begin{definition}[Deletion channel]\label{def:deletion-channel}
Let $T\geq1$, and let $V=(V_0,\ldots,V_{T-1})$ be the original clip, a sequence of frames over a finite alphabet. A \emph{deletion-only toy channel} with output length $T'$ turns it into the received clip $Y=(Y_0,\ldots,Y_{T'-1})$ with $Y_j=V_{\varphi(j)}$. Here $\varphi$ is a strictly increasing map from $\{0,\ldots,T'-1\}$ into $\{0,\ldots,T-1\}$, so $T'\leq T$. If $T'\geq1$, the \emph{origin} (or offset) of $Y$ is $\varphi(0)$, the original index of its first frame.
\end{definition}

A merely nondecreasing $\varphi$ also allows duplicated frames, and then $T'$ may exceed $T$. Inserted frames need a separate insertion symbol and values, and interpolation generally needs more than this integer-index model. These distinctions prevent ``time warp'' from concealing several different channels.

\begin{example}[A periodic carrier]\label{ex:periodic}
Consider a 48-frame sequence sampled at 24 frames per second, with timestamps $t/24$ for $t=0,\ldots,47$. Suppose, only for this timing example, that an ideal scalar measurement of frame $t$ is
\begin{equation}
 c_t=\cos(2\pi t/8).
 \label{eq:temporal-periodic}
\end{equation}
This carrier has period 8 frames, or 3 cycles per second.
\begin{enumerate}[label=(\alph*)]
\item For every integer $t$, $c_{t+8}=c_t$ and $c_{t+4}=-c_t$.
\item Deleting the first eight frames ($\varphi(j)=j+8$) and deleting the last eight ($\varphi(j)=j$) give the same 40 scalars: $c_{j+8}=c_j$ for $j=0,\ldots,39$.
\item Deleting the first four frames and deleting the last four frames give two 44-term sequences that are negatives of each other.
\end{enumerate}
Adding $2\pi$ or $\pi$ to the argument of the cosine gives (a), and (b) and (c) follow term by term.
\end{example}

By (a) and (b), a detector that sees only the scalars does not distinguish origins eight frames apart when the kept frames are consecutive. Visible boundaries, or side information such as length constraints or timestamps, can still supply origin information. By (c), an origin error of four frames becomes a sign error. A detector that averages signs before aligning them can then cancel evidence in this ideal model. The scalar measurement is a stipulation that isolates the timing issue, here and in Example~\ref{ex:aliasing}. We do not claim that any method obtains it reliably from arbitrary video.

For code-based rendering, a scene parameter might follow $c(t)=\cos(2\pi t/\lambda)$ with period $\lambda$, sampled at times that the renderer chooses. Changing the frame rate keeps this continuous function fixed but changes its sampled sequence. Resetting the time origin changes its phase. For a latent video model, a carrier tied to latent-frame numbers needs a relation to the decoded and received frame grids. Recovering a latent state does not itself settle that relation. VideoShield uses inversion and a spatiotemporal template, while VideoMark includes temporal matching of recovered frame information \cite{videoshield2025,videomark2025}.

\begin{example}[Aliasing under subsampling]\label{ex:aliasing}
Subsampling can erase a distinction between continuous carriers. Sample at 24 samples per second, at times $t/24$ for integer $t$. Carriers of 3 and of 9 cycles per second give $c^{(3)}_t=\cos(2\pi\cdot3t/24)=\cos(\pi t/4)$ and $c^{(9)}_t=\cos(2\pi\cdot9t/24)=\cos(3\pi t/4)$. These sequences differ at $t=1$: the values are $1/\sqrt2$ and $-1/\sqrt2$.

Now keep only the even-indexed samples, that is, every second sample. This gives 12 samples per second and the deletion map $\varphi(j)=2j$. For every integer $j$,
\[
 c^{(9)}_{2j}=\cos(3\pi j/2)=\cos(2\pi j-\pi j/2)=\cos(\pi j/2)=c^{(3)}_{2j}.
\]
The third equality uses the period and evenness of the cosine. So this exact deletion pattern makes the two scalar observations identical.
\end{example}

To rule out this example, one must restrict the allowed carrier frequencies or sampling changes. It is no remedy to ask a detector to infer information that the toy channel has removed.

Repetition offers another choice: carry the same message across time and aggregate local evidence. Video Seal propagates a pixel intervention and aggregates frame-level soft bits \cite{videoseal2024}. That choice reduces reliance on a unique frame origin. Repeated information alone cannot separate two origins that differ by a multiple of the period of a periodic carrier, that is, a carrier unchanged by a shift of one period (Example~\ref{ex:periodic}). The origin is then determined at most modulo that period, unless side information or a restricted candidate set removes the ambiguity.

A changing schedule can provide alignment information. It also gives editing more ways to desynchronize recovery. Transform-domain video work explicitly addresses frame-rate conversion and temporal redundancy \cite{dtcwt2022}. These are reasons to study carrier schedules, and not to assume that every temporal strategy has the same invariances.

\subsection*{A synchronization search is also a statistical search}

Alignment often means trying candidate offsets, frame-rate ratios, windows, or keys and reporting the strongest match. A threshold calibrated for one candidate need not be valid for that maximum. The issue occurs for neural inversion and for pixel measurement of rendered motion. The search procedure is part of the detector.

\begin{definition}[Valid family of $p$-values]\label{def:valid-pvalue}
Fix the tested key or keys and the public design parameters. Let $\Prb_0$ be a single null law for unmarked received clips $Y$. For example, $\Prb_0$ can be one member of $\mathcal Q_0$ (Section~\ref{sec:recovery}) with the keys fixed. Let $L\geq1$, and fix, before observing $Y$, measurable functions $p_1,\ldots,p_L$ from received clips to $[0,1]$. The family $p_1,\ldots,p_L$ is \emph{valid under $\Prb_0$} if
\begin{equation}
 \Prb_0\{p_\ell(Y)\leq u\}\leq u
 \quad\text{for every }u\in[0,1]\text{ and every }\ell.
 \label{eq:temporal-validity}
\end{equation}
\end{definition}

\begin{proposition}[Whole-search union bound]\label{prop:union}
Let $\Prb_0$ be a null law for received clips $Y$, let $p_1,\ldots,p_L$ with $L\geq1$ be a family valid under $\Prb_0$ in the sense of Definition~\ref{def:valid-pvalue}, and let $\alpha\in[0,1]$. Then
\begin{equation}
 \Prb_0\!\left\{\min_{1\leq\ell\leq L}p_\ell(Y)\leq\alpha/L\right\}
 \leq\alpha.
 \label{eq:temporal-union}
\end{equation}
The $p_\ell$ need not be independent. Consider a detector that declares a detection only if some member is at most $\alpha/L$. Its false-alarm probability under $\Prb_0$ is at most $\alpha$, even if it chooses which members to evaluate after seeing $Y$.
\end{proposition}

\begin{proof}[Proof of Proposition~\ref{prop:union}]
\begin{steps}
\step{1}{1} \Let $A_\ell=\{p_\ell(Y)\leq\alpha/L\}$. Then $\Prb_0\{\min_\ell p_\ell(Y)\leq\alpha/L\}\leq\sum_{\ell=1}^L\Prb_0(A_\ell)$.
\pf The event is $\bigcup_\ell A_\ell$, and pointwise $\mathbf1_{\cup_\ell A_\ell}\leq\sum_\ell\mathbf1_{A_\ell}$. Take expectations under $\Prb_0$.
\step{1}{2} $\Prb_0(A_\ell)\leq\alpha/L$ for every $\ell$.
\pf Because $\alpha\in[0,1]$ and $L\geq1$, the number $\alpha/L$ lies in $[0,1]$. Apply \eqref{eq:temporal-validity} with $u=\alpha/L$.
\step{1}{3} A detector as in the statement declares a detection only on $\bigcup_\ell A_\ell$.
\pf Some member must be at most $\alpha/L$. The $A_\ell$ depend only on the fixed functions $p_\ell$, so the choice of members to evaluate does not matter.
\qedstep{1}{4}
\pf By \sref{1}{1} and \sref{1}{2}, $\Prb_0(\bigcup_\ell A_\ell)\leq L\cdot(\alpha/L)=\alpha$, which is \eqref{eq:temporal-union}. By \sref{1}{3}, the false-alarm probability of such a detector under $\Prb_0$ is at most $\Prb_0(\bigcup_\ell A_\ell)$, so at most $\alpha$.
\end{steps}
\end{proof}

Validity is a stronger premise than naming a score a $p$-value. It must hold for every member under the declared null law, including nuisance transformations in it. More precise calibration may improve the conservative finite bound, but needs its own argument.

\begin{remark}[Key averaging]\label{rem:key-average}
A guarantee averaged over random keys (Section~\ref{sec:recovery}) does not automatically become a guarantee for each fixed key.
\end{remark}

Selecting or constructing a test after seeing the data requires accounting for the whole selection procedure. Reporting one selected test does not reduce the family size to one.

\begin{example}[Adaptive selection]\label{ex:adaptive}
Let $L\geq2$, and let $Y$ be uniform on $\{1,\ldots,L\}$, with law $\Prb_0$. For each fixed $\ell\in\{1,\ldots,L\}$ define
\begin{equation}
 p_\ell(y)=\begin{cases}1/L,&y=\ell,\\1,&y\ne\ell.\end{cases}
 \label{eq:temporal-adaptive}
\end{equation}
Each fixed test is valid: $\{p_\ell(Y)\leq u\}$ has probability zero for $u<1/L$, probability $1/L\leq u$ for $1/L\leq u<1$, and probability one for $u=1$.

But choosing $\ell=Y$ makes the reported value $1/L$ with probability one. Calling this a single test at level $1/L$ is wrong. Accounting for the full family of $L$ candidates retains \eqref{eq:temporal-union}, by Proposition~\ref{prop:union}. Discarding the alternatives after looking at the observation does not.
\end{example}

A practical first calculation can precede any implementation. It lists
\begin{itemize}
\item the proposed offsets, windows, sampling ratios, and key trials;
\item the statistical family that is searched;
\item whether its local null laws are known;
\item the threshold that the whole procedure requires.
\end{itemize}
A dynamic matching routine may hide many candidate alignments behind one output. Its reported statistic still needs calibration for that procedure.

\begin{example}[A bounded search]\label{ex:search-count}
Take twelve offsets, three sampling ratios, ten windows, and four predeclared keys. Their product has $12\cdot3\cdot10\cdot4=1440$ tuples, fixed before observing $Y$. Suppose each tuple supplies a valid test under the same declared null law $\Prb_0$. Declare a detection when some tuple's value is at most $0.01/1440=1/144000$. By Proposition~\ref{prop:union} with $L=1440$ and $\alpha=0.01$, this rule has false-alarm probability at most $0.01$ under $\Prb_0$. Some tuples can give identical tests; counting them separately is conservative, not invalid.
\end{example}

The calculation supplies no power guarantee. It identifies the size and calibration obligation of this search, including the key trials.

\subsection*{Editing budgets accumulate and change the question}

Individual edits can be harmless while their allowed composition leaves too little evidence.

\begin{example}[Two maximal crops]\label{ex:crop-budget}
Suppose a 64-frame clip allows deleting at most one quarter of its current frames at each editing step. Two maximal prefix crops leave $64\to48\to36$ frames. The combined deletion is 28 frames, or $7/16$ of the original clip, rather than one quarter. These lengths are divisible by four, so rounding plays no role.
\end{example}

A budget measured against the original clip would describe a different attacker.

\begin{definition}[Edit distance]\label{def:edit-distance}
Let $X_0$ and $X_1$ be finite sequences of frames over a finite alphabet. A \emph{script} is a finite list of single-frame insertions, deletions, and substitutions, applied in order. The \emph{edit distance} $d_{\mathrm{edit}}(X_0,X_1)$ is the least length of a script that turns $X_0$ into $X_1$.
\end{definition}

\begin{lemma}[Triangle inequality and additive budgets]\label{lem:edit-triangle}
Let $X_0$, $X_1$, and $X_2$ be finite sequences over a finite alphabet. Then
$d_{\mathrm{edit}}(X_0,X_2)\leq d_{\mathrm{edit}}(X_0,X_1)+d_{\mathrm{edit}}(X_1,X_2)$.
Thus, if $d_{\mathrm{edit}}(X_0,X_1)\leq e_1$ and $d_{\mathrm{edit}}(X_1,X_2)\leq e_2$, then $d_{\mathrm{edit}}(X_0,X_2)\leq e_1+e_2$.
\end{lemma}

\begin{proof}
Deleting all frames of one sequence and inserting all frames of another gives a script between them, so the least length exists. Concatenate optimal scripts from $X_0$ to $X_1$ and from $X_1$ to $X_2$. The result is a script from $X_0$ to $X_2$ whose length is the sum of the two distances. This proves the first inequality, and the second follows by adding the given bounds.
\end{proof}

This upper bound does not say that the operations commute, and it does not preserve a watermark recovery guarantee. The second edit acts on a changed sequence, and its standalone guarantee may have assumed an original input.

Cropping changes what a later resize sees, and interpolation changes what a later frame deletion removes. The combined output distribution may differ from the unmarked distribution used for calibration. So a robustness claim must name the composed class, and a detection claim must name the associated null family. Source editing followed by rerendering is a different access model: it changes the production state instead of composing edits on an existing video.

Reordering requires a further choice of channel. The strictly increasing deletion map of Definition~\ref{def:deletion-channel} does not describe swaps.

\begin{remark}[Permutations]\label{rem:permutation}
Let $y_0,\ldots,y_{T-1}$ be local measurements of $T$ frames, and let $\varpi$ be a bijection of $\{0,\ldots,T-1\}$. Then $\sum_jy_{\varpi(j)}=\sum_jy_j$, because $\varpi$ only reindexes the same sum, so a plain average does not change. For an ordered code $(\gamma_0,\ldots,\gamma_{T-1})$, the sum $\sum_j\gamma_jy_{\varpi(j)}$ need not equal $\sum_j\gamma_jy_j$. For example, take $T=4$, code $(1,-1,1,-1)$, and measurements $(3,-1,4,1)$. Swapping the first two measurements leaves the plain sum at $7$ but changes the ordered sum from $7$ to $-1$.
\end{remark}

Neither behavior is automatically preferable. Order-free message recovery can remain possible when an editing-history claim has become ambiguous. Reporting a plausible order does not authenticate that a particular issuer emitted the recovered frame sequence. The target should say which conclusion the surviving signal is meant to support.

\subsection*{A first channel calculation}

A first approach is a small ledger of operations, retained evidence, and searched alignments. It can expose an impossible target before a detector is built. It can also locate a tractable one: a short, fixed edit family with explicit access and a test whose whole-search validity is checkable.

A small starting problem allows one deletion at an unknown location and no interpolation (Definition~\ref{def:deletion-channel} with $T'=T-1$). List the candidate alignments and the carrier values each would retain. Check whether different messages can yield the same observation under different allowed deletions. Then ask how a stated noise bound changes those ambiguities. This is a channel and distinguishability calculation. Construction and robust decoding remain research questions.
\Needspace{9\baselineskip}
\section{Candidate Locations, Glossary, and Notation}
\label{app:glossary}

\subsection*{Candidate locations}

\begin{table}[ht]
\centering\small
\begingroup
\renewcommand{\arraystretch}{1.18}
\setlength{\extrarowheight}{2pt}
\arrayrulecolor{black}
\begin{tabularx}{\linewidth}{>{\raggedright\arraybackslash\sffamily\bfseries}p{0.18\linewidth}YY}
\toprule
\rowcolor{TableHead}
Location & \textbf{\sffamily Image/video generation} & \textbf{\sffamily Code-based rendering} \\
\midrule
Request / text & Conditions or generated textual content & Prompt, code tokens, identifiers, literals \\
\rowcolor{TableBand}
Representation & Noise, latent states, model parameters & Program, component graph, scene parameters, assets \\
Emission stage & Decoder or generated-frame intervention & Renderer or captured-frame intervention \\
\rowcolor{TableBand}
Final pixels & Pixels, transforms, codec-aware representation & The same output-level possibilities as for image/video generation \\
Side information & Signed manifests, service records, disclosed states & Signed manifests, source identity, execution receipts \\
\bottomrule
\end{tabularx}
\endgroup
\caption{Candidate locations. Section~\ref{sec:stages} maps its five stages to these rows. A location is not a security guarantee: each row needs a specified carrier, observation, attacker, and assertion.}
\label{tab:locations}
\end{table}

\subsection*{Glossary}

\begin{description}[style=unboxed,leftmargin=0pt,labelsep=0.5em,itemindent=0.5em,itemsep=2pt]
\item[Route, operation, workflow, pipeline.] A route is the informal contrast between image/video generation and code-based rendering. An operation is a node of the production model, a workflow a named real system or a described composition of operations, and a pipeline a specific implemented chain.
\item[Production graph.] The adopted model of operations and dependencies producing an artifact; it can include image/video generation, rendering, editing, and human contributions. A declared graph is not an authenticated execution history.
\item[Production event.] One invocation of a production step, such as one generation call, one rendering run, or one export. It is not the event that the video depicts.
\item[Image/video generation.] An operation in which a visual generative model (an image/video model) produces pixels, or representations decoded into pixels, conditioned on declared inputs. Code can orchestrate it.
\item[Code-based rendering.] Execution or interpretation of a program or declarative graphics description under a specified environment; author and workflow controller are separate variables.
\item[Rerendering and reconstruction.] Rerendering emits fresh output from an existing description or representation. Replacement substitutes clean material for a marked asset, and reconstruction rebuilds a task-equivalent program. Question~\ref{q:reconstruction} treats rerendering, replacement, and reconstruction as different access types. A marked asset can retain a mark in fresh views.
\item[Render equivalence.] Equality of rendered outputs under specified execution conditions; equality at one fixed environment is weaker than across an environment family.
\item[Task preservation.] Compliance with the application requirements declared for the problem. It can permit image changes that exact render equivalence forbids.
\item[Nuisance policy.] A rule by which a route chooses implementation details, such as which program or seed realizes a scene, that do not change the intended content.
\item[Verification specification.] The tuple $\mathfrak C$ of Section~\ref{sec:formal-contract}, which declares the seven items of Section~\ref{sec:spec-list}. Key generation, issuance queries, and verifier permissions complete it.
\item[Observation and side information.] An observation is the information actually supplied to a verifier: pixels, metadata, source, states, logs, or credentials. Side information means any information beyond the pixel content that the specification declares, such as metadata, timestamps, records, or disclosed states. The specification declares model access and public design parameters as auxiliary inputs.
\item[Auxiliary channel.] A channel other than the video pixels, such as audio, metadata, duration, or additional precision (Definition~\ref{def:payload} and the paragraph after it).
\item[Verifier, decoder, detector.] The verifier is the party or procedure that receives the observation. The decoder $D$ is its message-recovery map, and a detector is a presence test or a passive-inference label.
\item[Assertion.] The statement whose support the verifier tests, together with the verification task (the symbol $\chi$).
\item[Payload.] The information carried by the embedded message, measured in bits as $\log_2|\mathcal M|$ for a finite message set $\mathcal M$ (Definition~\ref{def:payload}; Proposition~\ref{prop:zero-error-count}).
\item[In-band watermark.] A watermark embedded in the content itself, not in separate metadata.
\item[Carrier.] A mark-bearing variation in text, a representation, an emitted signal, or a side channel; one observable at one stage may be absent downstream.
\item[Passive inference.] Estimating production labels from observations without arranging an active mark. The specification must declare the relevant distributions and unknown cases.
\item[Watermark detection and verification.] Detection tests for the presence of a mark; verification covers detection and message recovery.
\item[Message recovery.] Estimating a deliberately embedded message (payload extraction). A recovery statement requires its error model, side information, auxiliary inputs, and message population.
\item[False alarm and false acceptance.] A false alarm is acceptance of an unmarked observation as marked, under a declared family of unmarked observation laws. Searching many alignments or keys changes the statistical decision rule and its calibration. A false acceptance is acceptance of a credential outside the authorized relation (Question~\ref{q:credentials}).
\item[Authentication (authenticated provenance).] Checking an issuer-bound content or event assertion under declared trust assumptions. It does not establish a complete origin history.
\item[C2PA assertion and claim.] An assertion is a data structure that states something about an asset. A claim is a digitally signed data structure that references a set of assertions \cite[Sections~2.3.1 and~2.3.2]{c2pa}.
\item[Manifest.] In C2PA, the provenance information about an asset: one or more assertions (including content bindings), a single claim, and a claim signature. Content Credential is its preferred non-technical term \cite[Sections~2.3.4 and~2.3.6]{c2pa}.
\item[Content binding.] In C2PA, information that associates digital content with a specific manifest, either as a hard binding or as a soft binding \cite[Section~2.3.11]{c2pa}.
\item[Hard binding.] In C2PA, one or more cryptographic hashes that uniquely identify an entire asset or a portion of it \cite[Section~2.3.12]{c2pa}.
\item[Soft binding.] In C2PA, a content identifier that is not statistically unique, such as a fingerprint, or is an invisible watermark in the content \cite[Section~2.3.13]{c2pa}.
\item[Issuer binding and event binding.] Issuer binding ties an assertion to the issuer that makes it. Event binding, this article's term, ties an artifact or an allowed derivative to an authorized production event (Question~\ref{q:credentials}).
\end{description}

\subsection*{Notation}

Table~\ref{tab:notation} lists the symbols.

\begin{table}[!ht]
\centering\footnotesize
\begingroup
\renewcommand{\arraystretch}{1.02}
\setlength{\extrarowheight}{0.5pt}
\setlength{\tabcolsep}{4pt}
\arrayrulecolor{black}
\begin{tabularx}{\linewidth}{>{\raggedright\arraybackslash}p{0.19\linewidth}Y>{\raggedright\arraybackslash}p{0.255\linewidth}}
\toprule
\rowcolor{TableHead}
\textbf{\sffamily Symbol} & \textbf{\sffamily Meaning} & \textbf{\sffamily Where} \\
\midrule
$\mathcal G=(\mathcal N,\mathcal E),\;\operatorname{pa}(v)$ & Production graph (nodes, edges); parents of node $v$ & Section~\ref{sec:formal-contract} \\
\rowcolor{TableBand}
$Y_v,\;F_v,\;Z_v$ & Artifact, node map, and randomness at node $v$ & Section~\ref{sec:formal-contract} \\
$x,\;\mu_x$ & Task instance; joint law of the additional inputs $Z=(Z_v)_v$ & Section~\ref{sec:formal-contract} \\
\rowcolor{TableBand}
$S,\;s,\;\mathcal S$ & Production state (random, realized); set of allowed states & Section~\ref{sec:formal-contract}, Appendix~\ref{app:rendering} \\
$O(S),\;o$ & Verifier's observation; realized value & Section~\ref{sec:formal-contract}, Appendix~\ref{app:observation} \\
\rowcolor{TableBand}
$\mathfrak C$ & Verification specification $(\mathfrak P,\mathcal I,O,\mathcal A,U,\chi,\mathcal Q_0)$ & Section~\ref{sec:formal-contract} \\
$\mathfrak P,\;\mathcal I$ & Production specification $(\mathcal G,(F_v)_{v\in\mathcal N},(\mu_x)_x)$; encoder's allowed observations and interventions & Section~\ref{sec:formal-contract} \\
\rowcolor{TableBand}
$\mathcal A,\;a$ & Attacker's view and allowed policies; one allowed policy & Sections~\ref{sec:formal-contract}, \ref{sec:recovery} \\
$U_x(s_0,s')$ & Task-preservation predicate (baseline $s_0$, candidate $s'$) & Section~\ref{sec:formal-contract}, Appendix~\ref{app:rendering} \\
\rowcolor{TableBand}
$\chi$ & Assertion under test, with its verification task & Section~\ref{sec:formal-contract} \\
$\mathcal Q_0,\;Q$ & Family of unmarked observation laws; one member & Sections~\ref{sec:formal-contract}, \ref{sec:recovery} \\
\rowcolor{TableBand}
$(K_e,K_v),\;k,\;k_v$ & Issuance and verification keys; a fixed key; a verifier key value & Section~\ref{sec:recovery}, Appendix~\ref{app:observation} \\
$m,\;M,\;\mathcal M,\;N$ & Message, random message, message set, its size & Section~\ref{sec:recovery}, Appendices~\ref{app:observation}, \ref{app:rendering} \\
\rowcolor{TableBand}
$D,\;\bot$ & Decoder; abstention & Section~\ref{sec:recovery} \\
$\varepsilon$ & Message-error level; completeness-failure level & Section~\ref{sec:recovery}, Questions~\ref{q:synchronization}, \ref{q:credentials} \\
\rowcolor{TableBand}
$\alpha$ & False-alarm level; in Question~\ref{q:hybrid}, bound on accepting relation-violating negatives & Section~\ref{sec:recovery}, Appendix~\ref{app:temporal}, Questions~\ref{q:synchronization}, \ref{q:hybrid} \\
$\delta$ & False-acceptance level & Question~\ref{q:credentials} \\
\rowcolor{TableBand}
$\kappa$ & Minimum coverage or acceptance level & Questions~\ref{q:identifiable}, \ref{q:hybrid} \\
$\eta,\;r$ & Alignment error; span of the retained original-time window, in frame intervals & Question~\ref{q:synchronization} \\
\rowcolor{TableBand}
$\beta,\;b$ & Bound on average recovery; coefficient-substitution bound & Question~\ref{q:reconstruction} \\
$g,\;\sigma,\;R_{\mathrm{asset}}$ & Reference asset; object authenticating the composition masks; retained-asset relation & Question~\ref{q:hybrid} \\
\rowcolor{TableBand}
$v$ & A node of $\mathcal G$; elsewhere, a video & Section~\ref{sec:formal-contract}, Question~\ref{q:hybrid} \\
$P,\;R,\;z$ & Program, renderer, environment & Appendices~\ref{app:observation}, \ref{app:rendering} \\
\rowcolor{TableBand}
$E,\;W_k,\;E_k$ & Encoder; keyed program and neural-state encoders & Question~\ref{q:reconstruction}, Appendices~\ref{app:observation}, \ref{app:rendering} \\
$\Lambda,\;F$ & Neural state; deterministic decoder of a neural state to pixels & Appendix~\ref{app:rendering} \\
\rowcolor{TableBand}
$\Omega,\;f_0,\;f_1$ & Finite observation alphabet; mass functions of the two experiments & Appendices~\ref{app:observation}, \ref{app:rendering} \\
$\TV$ & Total variation distance & Appendix~\ref{app:observation} \\
\rowcolor{TableBand}
$T,\;q$ & Frame count; quantization levels per channel & Section~\ref{sec:agenda}, Appendices~\ref{app:rendering}, \ref{app:temporal} \\
$V(s),\;\mathcal V_U(s)$ & Ordinary output video of state $s$; task-allowed outputs of $s$ & Appendix~\ref{app:rendering}, Question~\ref{q:capacity} \\
\rowcolor{TableBand}
$\mathcal T$ & Edit class: deterministic output transformations & Appendix~\ref{app:rendering} \\
$\varphi,\;\lambda$ & Deletion map; period of a carrier & Appendix~\ref{app:temporal} \\
\rowcolor{TableBand}
$\Prb_0,\;p_\ell,\;L$ & Null law (for example, one member of $\mathcal Q_0$ with the keys fixed); $\ell$th test value; size of the searched family & Appendix~\ref{app:temporal} \\
\bottomrule
\end{tabularx}
\endgroup
\caption{Principal notation. Symbols local to one proof or example are defined where they occur and omitted here. Some letters, such as $b$, $d$, $x$, $V$, $W$, $X$, $Y$, $\mu$, and $K$, also have local meanings there.}
\label{tab:notation}
\end{table}

\section{Scoped Mechanism Map}\label{app:mechanism}

The map separates an intervention location from the verifier's observation. Entries identify precedents, not interchangeable security guarantees. The map includes abstract-screened methods to locate a mechanism. We did not independently audit their complete threat models and proofs.

\begin{table}[ht]
\centering\small
\begingroup
\renewcommand{\arraystretch}{1.12}
\setlength{\extrarowheight}{1pt}
\arrayrulecolor{black}
\begin{tabularx}{\linewidth}{>{\raggedright\arraybackslash\sffamily\bfseries}p{0.21\linewidth}>{\raggedright\arraybackslash\hsize=0.85\hsize}X>{\raggedright\arraybackslash\hsize=1.15\hsize}X}
\toprule
\rowcolor{TableHead}
Family & \textbf{\sffamily Examples} & \textbf{\sffamily Observation boundary} \\
\midrule
Video-asset marks & RivaGAN, DVMark, ItoV, Video Seal \cite{rivagan2019,dvmark2021,itov2023,videoseal2024} & Marked frames are the received signal; this does not establish source authorship. \\
\rowcolor{TableBand}
Video noise interventions & VideoShield, VideoMark \cite{videoshield2025,videomark2025} & Corresponding inversion paths and keys are auxiliary access. \\
Video decoder / model interventions & LVMark, Video Signature, SPDMark \cite{lvmark2024,vidsig2025,spdmark2025} & Learned emission changes differ from marking source tokens. \\
\rowcolor{TableBand}
Graphical video payload & Safe-Sora \cite{safesora2025} & A graphical message is the payload; abstract screening alone does not fix every embedding-stage detail. \\
Stage inheritance & LoT-Pass \cite{lotpass2025} & A marked input image passes through image-to-video generation; persistence needs its own channel model. \\
\rowcolor{TableBand}
Localized / image neighbors & TrustMark, Watermark Anything \cite{trustmark2023,wam2024} & Raster-image recovery and localization precede additional temporal questions. \\
Source-code marks & SWEET, ACW, SrcMarker, AST-guided marking, STONE \cite{lee2024sweet,li2024acw,yang2024srcmarker,guo2025astwatermark,kim2026stone} & The verifier receives code; raster recovery is a separate claim. \\
\rowcolor{TableBand}
Rendering-aware representations & CopyRNeRF, WateRF, GaussianMarker, 3D-GSW \cite{luo2023copyrnerf,jang2024waterf,huang2024gaussianmarker,jang2025gsw} & Representation changes deliberately reach rendered images. \\
Temporal / evaluation boundaries & DTCWT-SVD, VideoMarkBench \cite{dtcwt2022,videomarkbench2025} & Temporal processing and tested attack families require explicit scope. \\
\bottomrule
\end{tabularx}
\endgroup
\caption{Mechanism map for the compared production workflows, with an observation boundary for each family; the accompanying notes record exact versions and reading depth.}
\label{tab:mechanism-map}
\end{table}

Recent video-watermark variants include SIGMark, FlowMark, and COVER \cite{sigmark2026,flowmark2026,cover2026}. We retain their primary abstracts in the source notes. We import no unreviewed security claims from them.

Adjacent work concerns other tasks. Passive synthetic-video detection is a different task: GenVidBench and UNITE study specified detection populations \cite{genvidbench2025,unite2025}. Those populations do not automatically define inference about LLM-authored drawings or programs. Likewise, dataset-use marking in CodeMark and semantic text marking in SemaMark address different artifacts and observations \cite{sun2023codemark,ren2024semamark}. These distinctions keep adjacent work useful without adopting its success criterion as ours.


\begin{thebibliography}{99}
\small\raggedright
\bibitem{ho2022} Jonathan Ho et al. \emph{Video Diffusion Models}. 2022. arXiv:2204.03458. \url{https://arxiv.org/abs/2204.03458}.

\bibitem{rombach2022ldm} Robin Rombach et al. \emph{High-Resolution Image Synthesis with Latent Diffusion Models}. CVPR 2022. arXiv:2112.10752. \url{https://arxiv.org/abs/2112.10752v2}.

\bibitem{claudeCodeOverview} Anthropic. \emph{Overview -- Claude Code Docs}. 2026. Official documentation/source, accessed 4 October 2026. \url{https://code.claude.com/docs/en/overview}.

\bibitem{remotionCodingAgents} Remotion AG. \emph{Prompting videos with coding agents}. 2026. Official documentation/source, accessed 4 October 2026. \url{https://www.remotion.dev/docs/ai/coding-agents}.

\bibitem{openaiCodeInterpreter} OpenAI. \emph{Code Interpreter}. 2026. Official documentation/source, accessed 4 October 2026. \url{https://developers.openai.com/api/docs/guides/tools-code-interpreter}.

\bibitem{openaiImageGeneration} OpenAI. \emph{Image generation tool}. 2026. Official documentation/source, accessed 4 October 2026. \url{https://developers.openai.com/api/docs/guides/tools-image-generation}.

\bibitem{claudeCustomVisuals} Anthropic. \emph{Custom visuals in chat and Cowork}. 2026. Official documentation/source, accessed 4 October 2026. \url{https://support.claude.com/en/articles/13979539-custom-visuals-in-chat-and-cowork}.

\bibitem{remotionPromptToVideo} Remotion AG. \emph{Prompt to Video -- Remotion Template}. 2026. Official documentation/source, accessed 4 October 2026. \url{https://www.remotion.dev/templates/prompt-to-video}.

\bibitem{videoseal2024} Pierre Fernandez, Hady Elsahar, I. Zeki Yalniz, Alexandre Mourachko. \emph{Video Seal: Open and Efficient Video Watermarking}. 2024. arXiv:2412.09492. \url{https://arxiv.org/abs/2412.09492v1}.

\bibitem{yang2024srcmarker} Borui Yang, Wei Li, Liyao Xiang, Bo Li. \emph{SrcMarker: Dual-Channel Source Code Watermarking via Scalable Code Transformations}. 2024. 2024 IEEE Symposium on Security and Privacy (SP). \url{https://ieeexplore.ieee.org/document/10646683/}.

\bibitem{luo2023copyrnerf} Ziyuan Luo et al. \emph{CopyRNeRF: Protecting the CopyRight of Neural Radiance Fields}. 2023. arXiv:2307.11526. \url{https://openaccess.thecvf.com/content/ICCV2023/papers/Luo_CopyRNeRF_Protecting_the_CopyRight_of_Neural_Radiance_Fields_ICCV_2023_paper.pdf}.

\bibitem{c2pa} Coalition for Content Provenance and Authenticity. \emph{Content Credentials: {C2PA} Technical Specification}. 2025. Official documentation/source, accessed 4 October 2026. \url{https://spec.c2pa.org/specifications/specifications/2.2/specs/C2PA_Specification.html}.

\bibitem{remotionSSR} Remotion AG. \emph{Rendering using SSR APIs}. 2026. Official documentation/source, accessed 4 October 2026. \url{https://www.remotion.dev/docs/ssr-node}.

\bibitem{remotionClientRender} Remotion AG. \emph{Client-side rendering}. 2026. Official documentation/source, accessed 4 October 2026. \url{https://www.remotion.dev/docs/client-side-rendering}.

\bibitem{remotionLambda} Remotion AG. \emph{Remotion Lambda}. 2026. Official documentation/source, accessed 4 October 2026. \url{https://www.remotion.dev/docs/lambda}.

\bibitem{tewari2020neuralrendering} Ayush Tewari et al. \emph{State of the Art on Neural Rendering}. 2020. arXiv:2004.03805. \url{https://arxiv.org/abs/2004.03805v1}.

\bibitem{w3cPROV} W3C. \emph{PROV-DM: The PROV Data Model}. W3C Recommendation, 30 April 2013. Editors: Luc Moreau and Paolo Missier. \url{https://www.w3.org/TR/2013/REC-prov-dm-20130430/}.

\bibitem{li2026watermarkforensics} Xiaoyu Li et al. \emph{Watermark Forensics for Generative Models: An Information-Theoretic Perspective}. 2026. arXiv:2607.13003. \url{https://arxiv.org/abs/2607.13003v1}.

\bibitem{hidden2018} Jiren Zhu, Russell Kaplan, Justin Johnson, Li Fei-Fei. \emph{HiDDeN: Hiding Data With Deep Networks}. 2018. arXiv:1807.09937. \url{https://arxiv.org/abs/1807.09937v1}.

\bibitem{stegastamp2019} Matthew Tancik, Ben Mildenhall, Ren Ng. \emph{StegaStamp: Invisible Hyperlinks in Physical Photographs}. 2019. arXiv:1904.05343. \url{https://arxiv.org/abs/1904.05343v2}.

\bibitem{trustmark2023} Tu Bui, Shruti Agarwal, John Collomosse. \emph{TrustMark: Universal Watermarking for Arbitrary Resolution Images}. 2023. arXiv:2311.18297. \url{https://arxiv.org/abs/2311.18297v1}.

\bibitem{wam2024} Tom Sander et al. \emph{Watermark Anything with Localized Messages}. 2024. arXiv:2411.07231. \url{https://arxiv.org/abs/2411.07231v2}.

\bibitem{rivagan2019} Kevin Alex Zhang, Lei Xu, Alfredo Cuesta-Infante, Kalyan Veeramachaneni. \emph{Robust Invisible Video Watermarking with Attention}. 2019. arXiv:1909.01285. \url{https://arxiv.org/abs/1909.01285v1}.

\bibitem{dvmark2021} Xiyang Luo et al. \emph{DVMark: A Deep Multiscale Framework for Video Watermarking}. 2021. arXiv:2104.12734. \url{https://arxiv.org/abs/2104.12734v1}.

\bibitem{itov2023} Guanhui Ye et al. \emph{ItoV: Efficiently Adapting Deep Learning-based Image Watermarking to Video Watermarking}. 2023. arXiv:2305.02781. \url{https://arxiv.org/abs/2305.02781v1}.

\bibitem{dtcwt2022} Yifei Wang et al. \emph{A DTCWT-SVD Based Video Watermarking resistant to frame rate conversion}. 2022. arXiv:2206.01094. \url{https://arxiv.org/abs/2206.01094v1}.

\bibitem{flowmark2026} Vishal Asnani, Shruti Agarwal, John Collomosse. \emph{FlowMark: Mask-Guided Video Watermarking}. 2026. arXiv:2607.05261. \url{https://arxiv.org/abs/2607.05261v1}.

\bibitem{treering2023} Yuxin Wen, John Kirchenbauer, Jonas Geiping, Tom Goldstein. \emph{Tree-Ring Watermarks: Fingerprints for Diffusion Images that are Invisible and Robust}. 2023. arXiv:2305.20030. \url{https://arxiv.org/abs/2305.20030v3}.

\bibitem{gaussianshading2024} Zijin Yang et al. \emph{Gaussian Shading: Provable Performance-Lossless Image Watermarking for Diffusion Models}. 2024. arXiv:2404.04956. \url{https://arxiv.org/abs/2404.04956v3}.

\bibitem{gao2026slice} Zheng Gao et al. \emph{SLICE: Semantic Latent Injection via Compartmentalized Embedding for Image Watermarking}. 2026. arXiv:2603.12749. \url{https://arxiv.org/abs/2603.12749v1}.

\bibitem{videoshield2025} Runyi Hu et al. \emph{VideoShield: Regulating Diffusion-based Video Generation Models via Watermarking}. 2025. arXiv:2501.14195. \url{https://arxiv.org/abs/2501.14195v2}.

\bibitem{videomark2025} Xuming Hu et al. \emph{VideoMark: A Distortion-Free Robust Watermarking Framework for Video Diffusion Models}. 2025. arXiv:2504.16359. \url{https://arxiv.org/abs/2504.16359v3}.

\bibitem{sigmark2026} Xinjie Zhu et al. \emph{SIGMark: Scalable In-Generation Watermark with Blind Extraction for Video Diffusion}. 2026. arXiv:2603.02882. \url{https://arxiv.org/abs/2603.02882v1}.

\bibitem{stablesignature2023} Pierre Fernandez et al. \emph{The Stable Signature: Rooting Watermarks in Latent Diffusion Models}. 2023. arXiv:2303.15435. \url{https://arxiv.org/abs/2303.15435v2}.

\bibitem{lvmark2024} Youngdong Jang et al. \emph{LVMark: Robust Watermark for Latent Video Diffusion Models}. 2024. arXiv:2412.09122. \url{https://arxiv.org/abs/2412.09122v4}.

\bibitem{vidsig2025} Yu Huang et al. \emph{Video Signature: Implicit Watermarking for Video Diffusion Models}. 2025. arXiv:2506.00652. \url{https://arxiv.org/abs/2506.00652v4}.

\bibitem{spdmark2025} Samar Fares, Nurbek Tastan, Karthik Nandakumar. \emph{SPDMark: Selective Parameter Displacement for Robust Video Watermarking}. 2025. arXiv:2512.12090. \url{https://arxiv.org/abs/2512.12090v2}.

\bibitem{lee2024sweet} Taehyun Lee et al. \emph{Who Wrote this Code? Watermarking for Code Generation}. 2024. Proceedings of the 62nd Annual Meeting of the Association for Computational Linguistics (Volume 1: Long Papers). \url{https://aclanthology.org/2024.acl-long.268/}.

\bibitem{kim2026stone} Jungin Kim, Shinwoo Park, Yo-Sub Han. \emph{Marking Code Without Breaking It: Code Watermarking for Detecting LLM-Generated Code}. 2026. Findings of the Association for Computational Linguistics: EACL 2026. \url{https://aclanthology.org/2026.findings-eacl.207/}.

\bibitem{guo2025astwatermark} Zhimeng Guo, Minhao Cheng. \emph{Practical and Effective Code Watermarking for Large Language Models}. 2025. Advances in Neural Information Processing Systems. \url{https://papers.neurips.cc/paper_files/paper/2025/hash/6b4cb812f234a92ab757d1544912b4a8-Abstract-Conference.html}.

\bibitem{li2024acw} Boquan Li et al. \emph{Resilient Watermarking for AI-Generated Codes}. 2024. arXiv:2402.07518. \url{https://arxiv.org/pdf/2402.07518v2}.

\bibitem{jang2024waterf} Youngdong Jang et al. \emph{WateRF: Robust Watermarks in Radiance Fields for Protection of Copyrights}. 2024. arXiv:2405.02066. \url{https://openaccess.thecvf.com/content/CVPR2024/html/Jang_WateRF_Robust_Watermarks_in_Radiance_Fields_for_Protection_of_Copyrights_CVPR_2024_paper.html}.

\bibitem{huang2024gaussianmarker} Xiufeng Huang et al. \emph{GaussianMarker: Uncertainty-Aware Copyright Protection of 3D Gaussian Splatting}. 2024. arXiv:2410.23718. \url{https://papers.nips.cc/paper_files/paper/2024/hash/39cee562b91611c16ac0b100f0bc1ea1-Abstract-Conference.html}.

\bibitem{jang2025gsw} Youngdong Jang et al. \emph{3D-GSW: 3D Gaussian Splatting for Robust Watermarking}. 2025. arXiv:2409.13222. \url{https://openaccess.thecvf.com/content/CVPR2025/papers/Jang_3D-GSW_3D_Gaussian_Splatting_for_Robust_Watermarking_CVPR_2025_paper.pdf}.

\bibitem{lotpass2025} Guanjie Wang, Zehua Ma, Han Fang, Weiming Zhang. \emph{LoT-Pass: Long-term-robust Image Watermarking for Image to Video Generation}. 2025. arXiv:2509.17773. \url{https://arxiv.org/abs/2509.17773v3}.

\bibitem{remotionRendererMedia} Remotion AG. \emph{Remotion Renderer -- render-media.ts}. 2026. Official documentation/source, accessed 4 October 2026. \url{https://github.com/remotion-dev/remotion/blob/e385a83dbde54179c0457ad90b7d7c3a4b6ab44a/packages/renderer/src/render-media.ts}.

\bibitem{remotionRendererPerFrame} Remotion AG. \emph{Remotion Renderer -- render-frame-with-option-to-reject.ts}. 2026. Official documentation/source, accessed 4 October 2026. \url{https://github.com/remotion-dev/remotion/blob/e385a83dbde54179c0457ad90b7d7c3a4b6ab44a/packages/renderer/src/render-frame-with-option-to-reject.ts}.

\bibitem{remotionRendererCapture} Remotion AG. \emph{Remotion Renderer -- take-frame.ts}. 2026. Official documentation/source, accessed 4 October 2026. \url{https://github.com/remotion-dev/remotion/blob/e385a83dbde54179c0457ad90b7d7c3a4b6ab44a/packages/renderer/src/take-frame.ts}.

\bibitem{videomarkbench2025} Zhengyuan Jiang et al. \emph{VideoMarkBench: Benchmarking Robustness of Video Watermarking}. 2025. arXiv:2505.21620. \url{https://arxiv.org/abs/2505.21620v1}.

\bibitem{googleVeo31Example} Google AI for Developers. \emph{Generate videos with Veo 3.1 in Gemini API}. Official documentation, updated 17 September 2026; accessed 5 October 2026. \url{https://ai.google.dev/gemini-api/docs/veo?hl=en#extend-prompt}.

\bibitem{gao2026StillLoading} Zheng Gao. \emph{I created a music video ``still loading\textquotesingle\textquotesingle} [English translation of the Chinese title]. Bilibili, 1 October 2026; accessed 5 October 2026. \url{https://www.bilibili.com/video/BV15qYA6nE6u/}.

\bibitem{remotionTravelGallery} Remotion AG. \emph{Travel Route on Map with 3D Landmarks}. Remotion Prompts, credited to @JNYBGR; Claude Code / Opus 4.5. Accessed 5 October 2026. \url{https://www.remotion.dev/prompts/travel-route-on-map-with-3d-landmarks}.

\bibitem{sanjaya2026GamesVideo} Dilum Sanjaya (@DilumSanjaya). [Video post: top 20 games by copies sold]. X, 3 February 2026; accessed 5 October 2026. Descriptive title. \url{https://x.com/DilumSanjaya/status/2018367621381620142}.

\bibitem{remotionGamesGallery} Remotion AG. \emph{Three.js ``Top 20 Games Sold\textquotesingle\textquotesingle\ Ranking}. Remotion Prompts, credited to @DilumSanjaya; Claude Code. Accessed 5 October 2026. \url{https://www.remotion.dev/prompts/threejs-top-20-games-sold-ranking-1}.

\bibitem{zeng2026lava} Bokang Zeng et al. \emph{LAVA: Layered Audio-Visual Anti-tampering Watermarking for Robust Deepfake Detection and Localization}. 2026. arXiv:2604.23957. \url{https://arxiv.org/abs/2604.23957v2}.

\bibitem{genvidbench2025} Zhenliang Ni et al. \emph{GenVidBench: A 6-Million Benchmark for AI-Generated Video Detection}. 2025. arXiv:2501.11340. \url{https://arxiv.org/abs/2501.11340v3}.

\bibitem{unite2025} Rohit Kundu et al. \emph{Towards a Universal Synthetic Video Detector: From Face or Background Manipulations to Fully AI-Generated Content}. 2025. Proceedings of the IEEE/CVF Conference on Computer Vision and Pattern Recognition. \url{https://openaccess.thecvf.com/content/CVPR2025/html/Kundu_Towards_a_Universal_Synthetic_Video_Detector_From_Face_or_Background_CVPR_2025_paper.html}.

\bibitem{matplotlibSavefig} Matplotlib development team. \emph{matplotlib.pyplot.savefig}. 2026. Official documentation/source, accessed 4 October 2026. \url{https://matplotlib.org/stable/api/_as_gen/matplotlib.pyplot.savefig.html}.

\bibitem{w3cSVGRender} W3C. \emph{Rendering Model -- SVG 2}. 2026. Official documentation/source, accessed 4 October 2026. \url{https://www.w3.org/TR/SVG/render.html}.

\bibitem{gao2026trace} Zheng Gao et al. \emph{TRACE: A Two-Channel Robust Attribution Watermark via Complementary Embeddings for LLM-Agent Trajectories}. 2026. arXiv:2607.08400. \url{https://arxiv.org/abs/2607.08400v1}.

\bibitem{remotionCodeGeneration} Remotion AG. \emph{Generate Remotion Code using LLMs}. 2026. Official documentation/source, accessed 4 October 2026. \url{https://www.remotion.dev/docs/ai/generate}.

\bibitem{sanjaya2026GamesWorkflow} Dilum Sanjaya (@DilumSanjaya). [Workflow: Claude Code and Remotion Skills]. X thread post, 3 February 2026; accessed 5 October 2026. Descriptive title. \url{https://x.com/DilumSanjaya/status/2018367624019911007}.

\bibitem{sanjaya2026GamesAssets} Dilum Sanjaya (@DilumSanjaya). [Imported images and a Sketchfab background]. X thread post, 3 February 2026; accessed 5 October 2026. Descriptive title. \url{https://x.com/DilumSanjaya/status/2018367629141086242}.

\bibitem{remotionRenderMediaAPI} Remotion AG. \emph{renderMedia() API reference}. 2026. Official documentation/source, accessed 4 October 2026. \url{https://www.remotion.dev/docs/renderer/render-media}.

\bibitem{remotionRendererFrames} Remotion AG. \emph{Remotion Renderer -- render-frames.ts}. 2026. Official documentation/source, accessed 4 October 2026. \url{https://github.com/remotion-dev/remotion/blob/e385a83dbde54179c0457ad90b7d7c3a4b6ab44a/packages/renderer/src/render-frames.ts}.

\bibitem{claudeArtifactsHelp} Anthropic. \emph{What are artifacts and how do I use them?}. 2026. Official documentation/source, accessed 4 October 2026. \url{https://support.claude.com/en/articles/17153992-what-are-artifacts-and-how-do-i-use-them}.

\bibitem{remotionSkillsMarkup} Remotion AG. \emph{Remotion Markup Agent Skill}. 2026. Official documentation/source, accessed 4 October 2026. \url{https://github.com/remotion-dev/skills/blob/0b5db9daae40f42c73544d1cc0a8c733bd530eaa/skills/remotion-markup/SKILL.md}.

\bibitem{claudeMarking} Anthropic. \emph{How Claude marks AI-generated content}. 2026. Official documentation/source, accessed 4 October 2026. \url{https://support.claude.com/en/articles/16266773-how-claude-marks-ai-generated-content}.

\bibitem{chen2001qim} Brian Chen and Gregory W. Wornell. \emph{Quantization Index Modulation: A Class of Provably Good Methods for Digital Watermarking and Information Embedding}. IEEE Transactions on Information Theory, 47(4):1423--1443, 2001. \url{https://dsp-group.mit.edu/wp-content/uploads/2024/11/QuantizationIndexMod.pdf}.

\bibitem{bao2026shift} Rui Bao et al. \emph{SHIFT: Stochastic Hidden-Trajectory Deflection for Removing Diffusion-based Watermark}. 2026. arXiv:2603.29742. \url{https://arxiv.org/abs/2603.29742v2}.

\bibitem{gao2026csi} Zheng Gao et al. \emph{Breaking Semantic-Aware Watermarks via LLM-Guided Coherence-Preserving Semantic Injection}. 2026. arXiv:2602.21593. \url{https://arxiv.org/abs/2602.21593v1}.

\bibitem{christ2023} Miranda Christ, Sam Gunn, Or Zamir. \emph{Undetectable Watermarks for Language Models}. 2023. arXiv:2306.09194. \url{https://arxiv.org/abs/2306.09194}.

\bibitem{lin2026signatures} Huijia Lin, Kameron Shahabi, Min Jae Song. \emph{Unforgeable Watermarks for Language Models via Robust Signatures}. 2026. arXiv:2602.15323. \url{https://arxiv.org/abs/2602.15323}.

\bibitem{zhang2023sand} Hanlin Zhang et al. \emph{Watermarks in the Sand: Impossibility of Strong Watermarking for Generative Models}. 2023. arXiv:2311.04378. \url{https://arxiv.org/abs/2311.04378}.

\bibitem{safesora2025} Zihan Su et al. \emph{Safe-Sora: Safe Text-to-Video Generation via Graphical Watermarking}. 2025. arXiv:2505.12667. \url{https://arxiv.org/abs/2505.12667v2}.

\bibitem{cover2026} Yuxin Cao et al. \emph{COVER: Codec-Robust Video Watermarking with Generative Video Priors}. 2026. arXiv:2609.26236. \url{https://arxiv.org/abs/2609.26236v1}.

\bibitem{sun2023codemark} Zhensu Sun, Xiaoning Du, Fu Song, Li Li. \emph{CodeMark: Imperceptible Watermarking for Code Datasets against Neural Code Completion Models}. 2023. arXiv:2308.14401. \url{https://arxiv.org/abs/2308.14401}.

\bibitem{ren2024semamark} Jie Ren et al. \emph{A Robust Semantics-based Watermark for Large Language Model against Paraphrasing}. 2024. Findings of the Association for Computational Linguistics: NAACL 2024. \url{https://aclanthology.org/2024.findings-naacl.40/}.

\bibitem{lamport1995write} Leslie Lamport. \emph{How to Write a Proof}. The American Mathematical Monthly 102(7):600--608, 1995. doi:10.1080/00029890.1995.12004627. \url{https://doi.org/10.1080/00029890.1995.12004627}.

\bibitem{lamport2012write} Leslie Lamport. \emph{How to Write a 21st Century Proof}. Journal of Fixed Point Theory and Applications 11(1):43--63, 2012. doi:10.1007/s11784-012-0071-6. \url{https://doi.org/10.1007/s11784-012-0071-6}.
\end{thebibliography}
\end{document}